\documentclass[12pt]{article}

\usepackage[margin=1in]{geometry}
\usepackage{amsmath,amssymb,amsthm,mathtools}
\usepackage{natbib}
\usepackage{bibunits}
\usepackage[colorlinks,citecolor=blue,linkcolor=blue,urlcolor=blue,hyperfootnotes=false]{hyperref}
\usepackage{enumerate}
\usepackage{booktabs}
\usepackage{graphicx}
\usepackage{subcaption}
\usepackage{adjustbox}
\usepackage[flushleft]{threeparttable}
\usepackage{multirow,rotating,setspace}
\usepackage{xcolor}
\usepackage[normalem]{ulem}
\usepackage{etoolbox}
\makeatletter
\let\BUoriginalanchor\hyper@natanchorstart
\let\BUoriginallink\hyper@natlinkstart
\let\BUoriginalbreak\hyper@natlinkbreak
\renewcommand{\hyper@natanchorstart}[1]{\BUoriginalanchor{\@bibunitname.#1}}
\renewcommand{\hyper@natlinkstart}[1]{\BUoriginallink{\@bibunitname.#1}}
\renewcommand{\hyper@natlinkbreak}[2]{\BUoriginalbreak{#1}{\@bibunitname.#2}}
\makeatother
\allowdisplaybreaks

\renewcommand{\qed}{\rule{2mm}{2mm}}

\DeclareMathOperator*{\E}{E}
\DeclareMathOperator*{\Var}{Var}

\newcommand{\Prob}{\mathbb P}
\newcommand{\R}{\mathbb R}
\newcommand{\1}{\mathbf 1}
\newcommand{\calZ}{\mathcal Z}
\newcommand{\calA}{\mathcal A}
\newcommand{\calD}{\mathcal D}
\newcommand{\calS}{\mathcal S}
\newcommand{\calI}{\mathcal I}
\newcommand{\calE}{\mathcal E}
\newcommand{\calR}{\mathcal R}

\newtheorem{theorem}{Theorem}[section]
\newtheorem{lemma}{Lemma}[section]
\newtheorem{corollary}{Corollary}[section]
\newtheorem{proposition}{Proposition}[section]

\theoremstyle{plain}
\newtheorem{definition}{Definition}[section]

\theoremstyle{definition}
\newtheorem{remark}{Remark}[section]
\newtheorem{assumption}{Assumption}[section]

\newtheorem{testprocedure}{Procedure}[section]

\AtEndEnvironment{remark}{~\qed}
\AtEndEnvironment{example}{}

\begin{document}

\def\spacingset#1{\renewcommand{\baselinestretch}{#1}\small\normalsize}
\spacingset{1}
\begin{bibunit}[ims]
\hypersetup{pageanchor=false}

\title{Randomization Tests in Randomized Saturation Designs\thanks{We thank Toru Kitagawa, Yuya Sasaki, Takahide Yanagi, participants at EcoSta 2026, and seminar participants at Kyoto University, National University of Singapore, Nanyang Technological University, and Singapore Management University for helpful comments. All errors are our own.}}

\author{
Jizhou Liu \\
PHBS Business School\\
Peking University\\
\url{jizhou.liu@phbs.pku.edu.cn}
\and
Azeem M.\ Shaikh \\
Department of Economics\\
University of Chicago\\
\url{amshaikh@uchicago.edu}
\and
Liang Zhong \\
Faculty of Business and Economics\\
The University of Hong Kong\\
\url{samzl@hku.hk}
}

\date{November 14, 2026}
\maketitle

\onehalfspacing

\begin{abstract}
Randomized saturation designs are widely used to study spillover effects in clustered populations.
In these designs, clusters are first assigned to treatment saturation levels, and units are then randomized within clusters according to the assigned saturation.
This paper develops randomization tests for such experiments under several null hypotheses that arise naturally in spillover analysis.
We first compare each unit’s untreated potential outcomes at two saturation levels: the partially sharp null requires zero spillover for every unit, while the bounded null imposes a prespecified upper bound on each unit’s spillover effect.
Both hypotheses can be tested using a common conditional randomization framework that provides finite-sample validity.
We next develop a finite-sample valid unconditional randomization test of whether each unit’s untreated potential outcome is nondecreasing across saturation levels.
Finally, we study weak average-spillover nulls and show that our conditional randomization framework delivers asymptotically valid inference using studentized statistics.
Simulations and an application to the Zomba Cash Transfer experiment illustrate the finite-sample behavior and practical implementation of the methods.
\end{abstract}


\hypersetup{pageanchor=false}
\thispagestyle{empty}

\clearpage
\setcounter{page}{1}
\hypersetup{pageanchor=true}
\spacingset{1.72}

\section{Introduction}

Randomized saturation designs are a central tool for studying spillovers in clustered populations.
In a typical design, clusters are first randomized to different treatment saturation levels, and then units within each cluster are randomized to treatment subject to the assigned saturation.
The first stage creates variation in the intensity of exposure to treated peers, while the second stage separates own treatment from peer treatment intensity.
Such designs have been used in a wide range of applications and have generated an active literature on estimation and inference under partial interference \citep{hudgens2008,toulis2013,Basse2018,Basse2019,Imai2021,liutwostage}.

Despite this flexibility, valid inference in randomized saturation designs remains challenging for two related reasons.
First, some applications have few clusters per saturation level, making large-sample approximations potentially unreliable (see Appendix Table~\ref{tab:empirical}).
Fisherian randomization tests \citep{Fisher1935Design} are attractive in this setting because they rely only on the known assignment mechanism and can deliver exact finite-sample \(p\)-values under sharp null hypotheses \citep[see][]{imbens2015causal}.
Second, however, many hypotheses of substantive interest concern \emph{spillover effects} or \emph{total effects}, and these hypotheses are typically only partially sharp on the full assignment space under interference.
They restrict selected exposures and therefore do not impute the full schedule of missing potential outcomes \citep{zhong2024unconditional}.
As a result, a naive randomization test that resamples the full assignment vector need not generate a valid reference distribution.
This motivates conditioning on \emph{focal units} and restricting the reference distribution to assignments for which the relevant potential outcomes are imputable \citep{Athey2018,Basse2019}.

This paper develops conditional randomization tests (CRT) for randomized saturation designs under several spillover null hypotheses.
For a fixed pair of saturation levels \(s\) and \(s'\), we first study two individual-level nulls.
The first is a partially sharp null requiring \(Y_i(0,s)=Y_i(0,s')\) for every unit \(i\), where \(Y_i(0,s)\) denotes unit \(i\)'s potential outcome under control treatment status \(0\) and cluster saturation level \(s\).
The second is a bounded null requiring \(Y_i(0,s)-Y_i(0,s')\le \delta\) for every unit \(i\), where \(\delta\) is a researcher-chosen constant.
For both nulls, we use cluster-level permutations of the saturation labels, conditioning on untreated focal units in clusters assigned to \(s\) or \(s'\). Under the partially sharp null, the focal outcomes remain unchanged under these permutations. Under the bounded null, setting each unit’s spillover effect equal to the upper bound \(\delta\) gives a least-favorable imputation for the one-sided alternative.
Both tests therefore have finite-sample conditional validity.

We next study hypotheses involving more than two saturation levels.
In many applications the object of interest is not a single contrast on two levels, but the entire shape of the spillover response as the saturation level changes.
For example, researchers may want to test whether untreated outcomes are nondecreasing in the saturation level.
One approach tests individual ordered-exposure contrasts and combines the resulting tests.
We instead develop a pairwise imputation randomization test (PIRT) \citep{zhong2024unconditional}.
The test compares the observed assignment with each reassignment in the design support, using only units that are untreated under both assignments and whose exposures lie in the ordered set covered by the null.
This construction yields finite-sample size control and, as explained in Remark~\ref{rem:general-ordered-exposure}, extends naturally to other experimental designs with ordered exposures.

Finally, we consider weak average-spillover effects.
These nulls do not require every individual spillover effect to satisfy a unit-level restriction; instead, they set a prespecified weighted average of spillover effects equal to zero.
Leading examples include the  equally weighted and the cluster-size weighted cluster-average spillover effect.
Such weak nulls are scientifically less restrictive than the individual-level equality null, but they do not permit us to impute the missing potential outcomes of the focal units.
Consequently, the conditional permutation distribution is not an exact finite-sample null distribution.
We show that, when paired with an appropriately studentized cluster-level statistic, the same permutation device yields an asymptotically valid randomization-based test under regularity conditions for the corresponding weighted cluster-level array.

We illustrate the methods using the Zomba Cash Transfer experiment in Malawi
\citep{baird2011cash,Baird2018}. The application compares untreated schoolgirls across areas offering cash transfers to different proportions of girls, with few areas in each group.
Calibrated simulations show that conventional EA-clustered regression can overreject individual spillover bounds, while additional focal controls improve power at fixed cluster weights. Monotonicity comparisons reveal shape-dependent power, and the observed-data tests do not reject the examined restrictions at 5\%.

The paper makes three contributions.
First, it develops finite-sample conditional randomization tests for partially sharp and bounded spillover nulls under multi-saturation two-stage randomization.
Existing work on randomized saturation designs has primarily studied estimation and large-sample inference for direct and spillover effects \citep{hudgens2008,Basse2018,Imai2021,Gonzalo-two-stage,liutwostage}, while the conditional randomization-test literature has emphasized nonsharp nulls under interference more generally \citep{Aronow2012,Athey2018,Basse2019,puelz2021,basse2024,liu2026randomizationinferencetwosidedmarket,liu2026randomizationtestsswitchbackexperiments,hoshino2023}.
Unlike \citet{Athey2018}, we allow the selection of focal units to depend on the realized assignment, enabling us to select observed untreated units from clusters whose realized saturation is one of the two levels being compared. \citet{Basse2019} develop a general framework for conditional randomization tests and illustrate it using a two-stage design with either zero or one treated unit per cluster. Their implementation randomly selects one untreated focal unit per cluster. Building on this framework, we construct a focal-sampling rule under which any fixed focal subset has the same probability of selection under either contrasted saturation label. This permits multiple untreated focal units per cluster while preserving a uniform conditional permutation distribution. The construction accommodates heterogeneous cluster sizes and multiple saturation levels.

Second, the paper extends randomization inference in randomized saturation designs beyond equality nulls.
Recent work has emphasized that scientifically meaningful null hypotheses need not require exact equality of potential outcomes, but may instead impose bounded, ordered, or monotonic restrictions \citep{Caughey2023,huang2025randomization}.
For bounded spillover nulls, we show how conditional randomization tests can be adapted to allow nonzero but bounded spillover effects.
For monotone spillover effects across ordered saturation levels, we develop a randomization test that targets the monotonicity restriction directly.
Whereas \citet{huang2025randomization} construct and combine tests of ordered-exposure contrasts, we adapt pairwise-imputation testing to the joint inequality restriction without constructing separate contrast tests. The construction also applies to other known assignment distributions and ordered exposure mappings.

Third, we establish unconditional asymptotic validity for tests of
prespecified weighted average spillover effects, including equally
weighted cluster-average and unit-average targets. A population
weak null need not hold in the realized focal sample, so the
argument must account for the randomness introduced by focal
selection. We account jointly for the set of clusters assigned to
the contrasted saturation levels, the selection of focal units
within those clusters, and the assignment of saturation labels.
Under regularity conditions on the weighted cluster-level array,
studentized permutations of the observed focal outcomes provide
asymptotic size control for the weighted population null. This
extends the use of studentized randomization tests for weak nulls
\citep{chung2013,diciccio2017robust,ZHAO2021278,Wu2021,toulis2025asymptotic}
to a setting in which the outcomes entering the permutation
statistic are themselves selected through an auxiliary randomization.

The rest of the paper is organized as follows.
Section~\ref{sec:setup} introduces the design, potential outcomes, and four null hypotheses.
Section~\ref{sec:two-level-crt} begins with a motivating example and develops finite-sample conditional tests for the individual equality and bounded nulls.
Section~\ref{sec:monotone} develops a finite-sample pairwise-imputation test of increasing monotonicity.
Section~\ref{sec:weak-null} establishes unconditional asymptotic validity for weighted average nulls.
Section~\ref{sec:application_malawi} presents calibrated simulations and the empirical application.

\section{Setup and Notation}\label{sec:setup}

We work under a finite-population framework, following \citet{Basse2018,Basse2019}.
There are \(K\) clusters indexed by \(j=1,\ldots,K\).
Cluster \(j\) contains the unit set \(\mathcal I_j\), with \(|\mathcal I_j|=n_j\), and the total number of units is \(N=\sum_{j=1}^K n_j\).
Let \([i]\) denote the cluster containing unit \(i\).
Potential outcomes are treated as fixed throughout; randomness comes from the known assignment mechanism and, when used, auxiliary focal selection.

\subsection{Randomized saturation design}
\label{subsec:design}

The assignment proceeds in two stages: clusters are first assigned to treatment saturation labels, and units are then randomized within clusters conditional on the assigned label.
Let
\(
\calS=\{s_0,s_1,\ldots,s_L\}\subset[0,1]
\)
be the set of retained cluster-level saturation labels. If zero saturation is included, write \(s_0=0\).
Fix positive integers \(K_0,\ldots,K_L\) satisfying
\(
\sum_{\ell=0}^L K_\ell=K.
\)
These counts are fixed by design, with \(K_\ell\) denoting the number of clusters assigned to saturation label \(s_\ell\); labels with zero assignment probability are omitted.

Let \(A_j\in\calS\) denote the saturation assigned to cluster \(j\), and write \(A=(A_1,\ldots,A_K)\).
In the first stage, \(A\) is assigned by complete randomization over
\[
\calA
=
\Bigl\{
 a\in\calS^K:
 \#\{j:a_j=s_\ell\}=K_\ell,\ \ell=0,\ldots,L
\Bigr\},
\]
so that
\[
\Pr(A=a)
=
\binom{K}{K_0,\ldots,K_L}^{-1}\mathbf 1\{a\in\calA\}.
\]

For each cluster \(j\) and saturation label \(s\in\calS\), let
\(
m_j(s)\in\{0,1,\ldots,n_j\}
\)
denote the prespecified number of treated units if \(A_j=s\), with \(m_j(0)=0\) when \(0\in\calS\).
When \(s n_j\) is an integer, a natural choice is \(m_j(s)=s n_j\); otherwise \(m_j(s)\) may be determined by any rule fixed ex ante, such as \(\lfloor s n_j\rfloor\) or \(\lceil s n_j\rceil\).

Let \(D_i\in\{0,1\}\) denote the treatment indicator for unit \(i\), and write \(D=(D_i)_{i=1}^N\).
For any realization \(d\) of \(D\), let \(d_{\mathcal I_j}=(d_i)_{i\in\mathcal I_j}\) denote the vector of treatment indicators for the units in cluster \(j\).
Conditional on \(A=a\), the second stage randomizes independently across clusters, selecting exactly \(m_j(a_j)\) treated units uniformly without replacement within cluster \(j\).
Define
\[
\calD_j(s)
=
\left\{
 d_{\calI_j}\in\{0,1\}^{\calI_j}:
 \sum_{i\in\calI_j} d_i=m_j(s)
\right\}.
\]
Then
\[
\Pr(D=d\mid A=a)
=
\prod_{j=1}^K
\binom{n_j}{m_j(a_j)}^{-1}
\mathbf 1\{d_{\calI_j}\in\calD_j(a_j)\}.
\]
The full assignment is \(Z=(A,D)\), with support
\[
\calZ
=
\left\{
(a,d):
 a\in\calA,
 d_{\calI_j}\in\calD_j(a_j)\ \text{for all }j
\right\}.
\]
The designs studied in \citet{Basse2018,Basse2019,liutwostage} arise as the special case \(\calS=\{0,\pi_2\}\) for some \(\pi_2\in(0,1)\).

\subsection{Potential outcomes}\label{subsec:potential-outcomes}

For each unit \(i\) and assignment \(z\in\calZ\), let \(Y_i(z)\) denote the potential outcome of unit \(i\) under assignment \(z\), and
the observed outcome is
\(
Y_i^{\mathrm{obs}}=Y_i(Z^{\mathrm{obs}}).
\)
We impose the standard homogeneous partial-interference restriction for randomized saturation designs \citep[see, e.g.,][]{hudgens2008,Basse2018,Basse2019,Forastiere2021,Imai2021,Gonzalo-two-stage,liutwostage}.

\begin{assumption}[Homogeneous partial interference]\label{ass:hpi}
For any unit \(i\) and any two feasible assignments \(z=(a,d),z'=(a',d')\in\calZ\),
\[
Y_i(z)=Y_i(z')
\quad\text{whenever}\quad
d_i=d_i'
\quad\text{and}\quad
\sum_{k\in\calI_{[i]}}d_k
=
\sum_{k\in\calI_{[i]}}d_k'.
\]
\end{assumption}

Assumption~\ref{ass:hpi} has two components.
First, it rules out interference across clusters: assignments outside unit \(i\)'s cluster do not affect \(i\)'s outcome.
Second, within a cluster, the outcome depends on the assignment only through unit \(i\)'s own treatment status and the number of treated units in the cluster.
Conditional on these two quantities, the identities of the treated peers are irrelevant. 

Assumption~\ref{ass:hpi} is an exposure-mapping restriction, not a consequence of random assignment. If it is misspecified, then a rejection can reflect failure of the exposure mapping rather than of the stated spillover restriction \citep{basse2024}.

Because the design fixes the number of treated units in cluster \(j\) at \(m_j(s)\) whenever \(A_j=s\), Assumption~\ref{ass:hpi} allows us to index potential outcomes by own treatment status and cluster saturation whenever that exposure is feasible.
For cluster \(j\), define the feasible exposure set from the assignment support:
\[
\calE_j=\{(r,s):\exists (a,d)\in\calZ,\ a_j=s,\ d_i=r\},\qquad i\in\calI_j.
\]
Within-cluster randomization makes this set the same for every \(i\in\calI_j\).
Under the design in Section~\ref{subsec:design}, it equals
\(\{(0,s):s\in\calS,\ m_j(s)<n_j\}\cup\{(1,s):s\in\calS,\ m_j(s)>0\}\).
For \(i\in\calI_j\), the shorthand \(Y_i(r,s)\) is used only for \((r,s)\in\calE_j\).
The observed outcome can therefore be written as
\[
Y_i^{\mathrm{obs}}
=
Y_i\!\left(D_i^{\mathrm{obs}},A_{[i]}^{\mathrm{obs}}\right),
\]
where the realized exposure is necessarily feasible.

Write \(\calS_{0j}=\{s\in\calS:(0,s)\in\calE_j\}\) for the feasible untreated labels in cluster \(j\). The main results compare labels in the common untreated-support set
\[
\calS_0
=
\bigcap_{j=1}^K\calS_{0j}
=\{s\in\calS:m_j(s)<n_j\ \text{for every }j=1,\ldots,K\}.
\]
\subsection{Null hypotheses}\label{subsec:null-hypotheses}
We consider four null hypotheses. Fix distinct \(s,s'\in\calS_0\) for the pairwise contrasts.

\paragraph{Individual no-spillover equality.}
The first null is
\begin{equation}\label{eq:null-ps}
H_{0,\mathrm{PS}}^{s,s'}:
\qquad
Y_i(0,s)=Y_i(0,s')
\quad\text{for all } i=1,\ldots,N.
\end{equation}
This null is partially sharp on the full assignment space, which means that
it links two untreated exposures but does not determine treated outcomes or outcomes under other saturation levels, such as $Y_i(1, s)$.

\paragraph{Bounded individual spillovers.}
For a prespecified constant \(\delta\in\R\), define
\begin{equation}\label{eq:null-bounded}
H_{\delta,\mathrm{B}}^{s,s'}:
\qquad
Y_i(0,s)-Y_i(0,s')\le \delta
\quad\text{for all } i=1,\ldots,N.
\end{equation}
Rejection provides evidence that the bound fails for at least one unit, rather than that an average effect exceeds the bound.

\paragraph{Increasing monotonicity.}
Let
\[
\calS_M=\{q_0,q_1,\ldots,q_M\}\subseteq \calS_0,
\qquad
q_0<q_1<\cdots<q_M,
\]
be the ordered set of feasible untreated saturation levels over which the researcher wants to test monotonicity.
The case \(\calS_M=\calS_0\) corresponds to the null over the full untreated-support set.
Allowing \(\calS_M\) to be a strict subset is useful when the substantive question concerns monotonicity only over a particular range of saturation levels.

We focus on the monotone increasing spillover effects for untreated units,
\begin{equation}\label{eq:null-monotone}
H_{0,\mathrm M}(\calS_M):
\qquad
Y_i(0,q_0)\le Y_i(0,q_1)\le\cdots\le Y_i(0,q_M)
\quad
\text{for all } i=1,\ldots,N.
\end{equation}
The direction can be reversed if the substantive theory predicts that higher saturation weakly lowers untreated outcomes. The null is unitwise: it is stronger than monotonicity of the average spillover response.

\paragraph{Weighted average spillovers.}\label{subsec:weak-nulls}
For cluster \(j\), define
\[
\mu_{a,j}
:=
\frac{1}{n_j}\sum_{i\in\calI_j}Y_i(0,a),
\qquad a\in\{s,s'\},
\]
and
\[
\tau_j^{s,s'}
:=
\mu_{s,j}-\mu_{s',j}
=
\frac{1}{n_j}\sum_{i\in\calI_j}\bigl(Y_i(0,s)-Y_i(0,s')\bigr).
\]
Let \(\lambda_{1K},\ldots,\lambda_{KK}\) be prespecified nonnegative weights with \(\sum_{j=1}^K\lambda_{jK}=1\).
The weighted weak null is
\begin{equation}\label{eq:null-weak-weighted}
H_{0,\mathrm W}^{s,s'}(\lambda):
\qquad
\bar\tau_{\lambda,K}^{s,s'}
:=
\sum_{j=1}^K\lambda_{jK}\tau_j^{s,s'}
=
0.
\end{equation}

The choices \(\lambda_{jK}=1/K\) and \(\lambda_{jK}=n_j/N\) give, respectively,
\begin{align}
H_{0,\mathrm W,\mathrm{eq}}^{s,s'}:\quad
\bar\tau_K^{s,s'}&:=K^{-1}\sum_{j=1}^K\tau_j^{s,s'}=0,
\label{eq:null-weak-ew}\\
H_{0,\mathrm W,N}^{s,s'}:\quad
\bar\tau_N^{s,s'}&:=N^{-1}\sum_{j=1}^K n_j\tau_j^{s,s'}=0.
\label{eq:null-weak-cw}
\end{align}
The two targets coincide when cluster sizes are constant. These nulls relax individual equality by allowing positive and negative effects to cancel; they are not generally nested with the bounded null.

In addition to the spillover nulls, Appendix~\ref{app:other-estimands} extends our methods to null hypotheses on treated-spillover, direct, and total effects, including all-treated versus all-control comparisons. It also considers joint no-spillover nulls.

\section{Conditional Tests for Two Saturation Levels}\label{sec:two-level-crt}

This section develops finite-sample conditional tests of the individual equality and bounded nulls in \eqref{eq:null-ps} and \eqref{eq:null-bounded}.
For a fixed pair \(s,s'\in\calS_0\), the construction has two steps.
First, we select untreated focal units from clusters whose observed saturation is either \(s\) or \(s'\).
Second, we generate the reference distribution by simply permuting the two saturation labels across those clusters.
Under the partially sharp null, this conditioning makes the focal outcomes imputable.
Under the bounded null, setting each individual spillover effect equal to the upper bound \(\delta\) gives a least-favorable imputation.

\subsection{A motivating example}
\label{subsec:toy-crt}

We use a toy example with five clusters of five units each to illustrate the conditional randomization test for a single hypothesis, the partially sharp null of no spillover between 60\% and 20\% saturation. Specifically,  clusters 1 and 2 are observed at
$60\%$ saturation, clusters 3 and 4 at $20\%$ saturation, and cluster 5 at
$0\%$ saturation; these labels treat three, one, and zero units per cluster.
We illustrate the null
$H_{0,\mathrm{PS}}^{0.6,0.2}:Y_i(0,0.6)=Y_i(0,0.2)$ for every $i$.
In the panels below, bold $\mathbf 1$ denotes treatment.
For $j=1,\ldots,4$, let $U_j$ denote the set of two focal controls marked by boxes in cluster $j$ of Panel A.
Write
$y_j=\widetilde Y_j^{\mathrm{obs}}
=2^{-1}\sum_{i\in U_j}Y_i^{\mathrm{obs}}$ for their mean observed outcome.

\begin{center}
\renewcommand{\arraystretch}{1.06}

\begin{minipage}{\textwidth}
\centering
\textit{Panel A: observed assignment and focal controls}
\resizebox{\textwidth}{!}{%
\begin{tabular}{@{}lccccc@{}}
\toprule
& Cluster 1 & Cluster 2 & Cluster 3 & Cluster 4 & Cluster 5 \\
\midrule
Saturation
& $60\%$ & $60\%$ & $20\%$ & $20\%$ & $0\%$ \\
Treatment vector
& $(\mathbf 1,\mathbf 1,\mathbf 1,\boxed{0},\boxed{0})$
& $(\mathbf 1,\boxed{0},\mathbf 1,\boxed{0},\mathbf 1)$
& $(\mathbf 1,\boxed{0},\boxed{0},0,0)$
& $(\boxed{0},\mathbf 1,0,\boxed{0},0)$
& $(0,0,0,0,0)$ \\
Focal mean
& $y_1$ & $y_2$ & $y_3$ & $y_4$ & -- \\
\bottomrule
\end{tabular}%
}
\footnotesize

Cluster 5 lies outside the pairwise comparison and is therefore non-focal.
\end{minipage}

\vspace{15pt}
\begin{minipage}{\textwidth}
\centering
\textit{Panel B: naive draw of a new two-stage assignment}
\resizebox{\textwidth}{!}{%
\begin{tabular}{@{}lccccc@{}}
\toprule
& Cluster 1 & Cluster 2 & Cluster 3 & Cluster 4 & Cluster 5 \\
\midrule
New saturation
& $20\%$ & $60\%$ & $60\%$ & $20\%$ & $0\%$ \\
New treatment vector
& $(\mathbf 1,0,0,\boxed{0},\boxed{0})$
& $(\mathbf 1,\boxed{0},\mathbf 1,\boxed{0},\mathbf 1)$
& $(\mathbf 1,\boxed{\mathbf 1}^{\dagger},\boxed{0},\mathbf 1,0)$
& $(\boxed{0},\mathbf 1,0,\boxed{0},0)$
& $(0,0,0,0,0)$ \\
\bottomrule
\end{tabular}%
}
\footnotesize

$^{\dagger}$A focal control from Panel A becomes treated, requiring
$Y_i(1,0.6)$, which the null does not restrict.
\end{minipage}

\vspace{15pt}
\begin{minipage}{\textwidth}
\centering
\textit{Panel C: naive shuffle of control-exposure labels across units}
\resizebox{\textwidth}{!}{%
\begin{tabular}{@{}lccccc@{}}
\toprule
& Cluster 1 & Cluster 2 & Cluster 3 & Cluster 4 & Cluster 5 \\
\midrule
Observed
& $(\mathbf 1,\mathbf 1,\mathbf 1,0.6,0.6)$
& $(\mathbf 1,0.6,\mathbf 1,0.6,\mathbf 1)$
& $(\mathbf 1,0.2,0.2,0.2,0.2)$
& $(0.2,\mathbf 1,0.2,0.2,0.2)$
& $(0,0,0,0,0)$ \\
After one unit-level swap
& $(\mathbf 1,\mathbf 1,\mathbf 1,\underline{0.2},0.6)$
& $(\mathbf 1,0.6,\mathbf 1,0.6,\mathbf 1)$
& $(\mathbf 1,\underline{0.6},0.2,0.2,0.2)$
& $(0.2,\mathbf 1,0.2,0.2,0.2)$
& $(0,0,0,0,0)$ \\
\bottomrule
\end{tabular}%
}
\footnotesize

Bold $\mathbf 1$ denotes a treated unit; numerical entries are the
saturation labels attached to observed controls.
\end{minipage}
\end{center}

Panel A reports the observed saturation and treatment assignments and the mean outcome of the two boxed focal controls in each of Clusters 1--4.
Panel B shows a feasible draw from the two-stage design in which the second unit in Cluster 3, a focal control in Panel A, becomes treated at 60\% saturation.
The null equates this unit’s untreated potential outcomes at 20\% and 60\% saturation, but it leaves \(Y_i(1,0.6)\) unrestricted.
We therefore cannot compute the focal mean for Cluster 3 under this draw.
In Panel C, swapping exposure labels between individual controls gives two controls in Cluster 1 different saturation labels, even though the design assigns a single label to the entire cluster.

Nor does fixing the mean of all observed controls and permuting cluster labels generally give a valid test: the control sample changes with saturation (Appendix~\ref{app:naive-aggregation}). Our method instead selects the two boxed focal controls in each of Clusters 1--4 and conditions on this selection. Under the null, their means remain \(y_1,\ldots,y_4\), and the conditional distribution permutes the two saturation labels uniformly across those clusters.

To illustrate the $p$-value calculation, consider the difference in focal-cluster means as the test statistic,
\[
T(a)
=
\frac{1}{2}\sum_{j:a_j=0.6}y_j
-
\frac{1}{2}\sum_{j:a_j=0.2}y_j,
\qquad a=(a_1,\ldots,a_4).
\]
The six equally likely permutations and their statistic values are

\begin{center}
\footnotesize
\renewcommand{\arraystretch}{1.10}
\setlength{\tabcolsep}{4.2pt}
\begin{tabular}{@{}cc@{\qquad}cc@{}}
\toprule
Permuting $a$ & $T(a)$ & Permuting $a$ & $T(a)$ \\
\midrule
$(0.6,0.6,0.2,0.2)$ (observed)
& $\tfrac12(y_1+y_2-y_3-y_4)$
& $(0.2,0.2,0.6,0.6)$
& $\tfrac12(y_3+y_4-y_1-y_2)$ \\
$(0.6,0.2,0.6,0.2)$
& $\tfrac12(y_1+y_3-y_2-y_4)$
& $(0.2,0.6,0.2,0.6)$
& $\tfrac12(y_2+y_4-y_1-y_3)$ \\
$(0.6,0.2,0.2,0.6)$
& $\tfrac12(y_1+y_4-y_2-y_3)$
& $(0.2,0.6,0.6,0.2)$
& $\tfrac12(y_2+y_3-y_1-y_4)$ \\
\bottomrule
\end{tabular}
\end{center}
For an upper-tail test, the conditional \(p\)-value is the fraction of these
six values that are at least as large as the observed statistic:
\[
p_T
=
\frac{1}{6}
\sum_{r=1}^{6}
\1\!\left\{
T\!\left(a^{(r)}\right)
\ge
T\!\left(A_J^{\mathrm{obs}}\right)
\right\},
\]
where \(a^{(1)},\ldots,a^{(6)}\) denote the six permuted vectors listed
above and \(A_J^{\mathrm{obs}}=(0.6,0.6,0.2,0.2)\).

\subsection{Focal units and randomization distribution}\label{subsec:conditioning-event}

Fix two distinct saturation levels \(s,s'\in\calS_0\).
The comparison uses only clusters whose observed saturation is one of these two levels.
Let
\(
J^{\mathrm{obs}}
:=
\{j:A_j^{\mathrm{obs}}\in\{s,s'\}\}
\)
be this set of contrasted clusters.
For each cluster \(j\), choose an integer \(k_j\) before observing the assignment such that
\begin{equation}\label{eq:kj-condition}
1\le k_j\le \min\{n_j-m_j(s),\,n_j-m_j(s')\}.
\end{equation}
Because \(s,s'\in\calS_0\), the right-hand side is positive for every cluster.
Condition~\eqref{eq:kj-condition} guarantees that cluster \(j\) has at least \(k_j\) untreated units available whether its saturation label is \(s\) or \(s'\). 

After observing the assignment, for each \(j\in J^{\mathrm{obs}}\), sample
\[
U_j\subseteq \{i\in\calI_j:D_i^{\mathrm{obs}}=0\},
\qquad |U_j|=k_j,
\]
uniformly without replacement, independently across contrasted clusters.
Let
\(
U=\bigcup_{j\in J^{\mathrm{obs}}}U_j
\)
be the focal unit set. In the toy example in Section~\ref{subsec:toy-crt}, \(J^{\mathrm{obs}}=\{1,2,3,4\}\), with \(k_j=2\) focal controls selected from each of these clusters. In Panel A of Section~\ref{subsec:toy-crt}, \(U_j\) consists of the two boxed controls in cluster \(j\), and \(U\) consists of all eight boxed controls.

The focal-set size $k_j$ must not depend on which saturation label is
realized. For any fixed subset $u_j\subseteq\calI_j$ of size $k_j$
and either $q\in\{s,s'\}$,
\begin{equation}\label{eq:focal-selection-cancellation}
\Pr(U_j=u_j\mid A_j=q)
=
\frac{\binom{n_j-k_j}{m_j(q)}}
{\binom{n_j}{m_j(q)}\binom{n_j-m_j(q)}{k_j}}
=\binom{n_j}{k_j}^{-1}.
\end{equation}
The numerator counts assignments leaving $u_j$ untreated; the denominator
accounts for treatment assignment and focal selection. Thus each focal
subset has the same selection probability under either label.
By Bayes' rule, this label-independent probability leaves the relative probabilities of the admissible saturation assignments unchanged after conditioning on the focal sets and the outside labels. Thus focal selection ensures both imputability under the equality null and the uniform conditional randomization distribution. Appendix~\ref{app:crt-framework} gives the formal conditioning event, and Lemma~\ref{lem:conditional-relabeling} establishes the distribution.

A simple choice is \(k_j=\min\{n_j-m_j(s),n_j-m_j(s')\}\), the largest feasible focal-set size.

The reference distribution is obtained by permuting the two saturation labels \(s\) and \(s'\) across the contrasted clusters.
Let
\(
A_J^{\mathrm{obs}}=(A_j^{\mathrm{obs}}:j\in J^{\mathrm{obs}}),
\)
and define the observed margins
\[
K_s:=\sum_{j\in J^{\mathrm{obs}}}\1\{A_j^{\mathrm{obs}}=s\},
\qquad
K_{s'}:=\sum_{j\in J^{\mathrm{obs}}}\1\{A_j^{\mathrm{obs}}=s'\}.
\]
The randomization space is
\begin{equation}\label{eq:relabel-space}
\calA_J^{\mathrm{obs}}
=
\left\{
 a_J\in\{s,s'\}^{|J^{\mathrm{obs}}|}:
 \sum_{j\in J^{\mathrm{obs}}}\1\{a_j=s\}=K_s
\right\}.
\end{equation}
The randomization distribution is uniform over \(\calA_J^{\mathrm{obs}}\).
Equivalently, a reference draw selects exactly \(K_s\) of the contrasted clusters to carry label \(s\), assigns label \(s'\) to the remaining \(K_{s'}\) contrasted clusters, and leaves all other clusters outside the comparison.

For each contrasted cluster, define the observed focal-cluster mean
\begin{equation}\label{eq:focal-mean}
\widetilde Y_j^{\mathrm{obs}}
:=
\frac{1}{k_j}\sum_{i\in U_j}Y_i^{\mathrm{obs}},
\qquad j\in J^{\mathrm{obs}}.
\end{equation}
In the toy example in Section~\ref{subsec:toy-crt}, \(\widetilde Y_j^{\mathrm{obs}}=y_j\) is the average outcome of the two boxed controls in cluster \(j\) in Panel A, for \(j=1,\ldots,4\). The observed labels are \(A_J^{\mathrm{obs}}=(0.6,0.6,0.2,0.2)\), with \(K_{0.6}=K_{0.2}=2\). The set \(\mathcal A_J^{\mathrm{obs}}\) contains the six equally likely permuted labels listed in that example. Under the null of no-spillover effect, the means \(y_1,\ldots,y_4\) remain unchanged as these labels are permuted.

\begin{remark}[Choice of focal-set size]\label{rem:choice-kj}
The counts \(k_j\) are chosen before observing the assignment. Equation~\eqref{eq:focal-selection-cancellation} holds for every such choice satisfying \eqref{eq:kj-condition}; larger focal samples use more observed controls without changing the permutation space.
\end{remark}

\subsection{Test statistics}\label{subsec:test-statistics}

For the partially sharp null \(H_{0,\mathrm{PS}}^{s,s'}\), conditioning makes
the focal outcomes imputable over \(\calA_J^{\mathrm{obs}}\). Hence any
statistic constructed from the fixed focal outcomes, or their cluster-level
summaries, and the cluster-level saturation labels can be used. Our default
choice is the difference in focal-cluster means,
\begin{equation}\label{eq:stat-diffmean}
\hat\tau_U(a_{J})
=
\frac{1}{K_s}
\sum_{j\in J^{\mathrm{obs}}}
\widetilde Y_j^{\mathrm{obs}}\1\{a_j=s\}
-
\frac{1}{K_{s'}}
\sum_{j\in J^{\mathrm{obs}}}
\widetilde Y_j^{\mathrm{obs}}\1\{a_j=s'\},
\end{equation}
where \(a_J\in\{s,s'\}^{|J^{\mathrm{obs}}|}\) is a fixed vector assigning saturation \(s\) or \(s'\) to each cluster in \(J^{\mathrm{obs}}\).

For the bounded null \(H_{\delta,\mathrm B}^{s,s'}\), the focal outcomes
need not remain fixed across permutations. We therefore use the equality $Y_i(0,s)-Y_i(0,s')= \delta$ to construct a conservative reference distribution and restrict
attention to effect-increasing statistics \citep{Caughey2023}.

\begin{samepage}
\begin{definition}[Effect-increasing focal statistic]
\label{def:effect-increasing-focal}
A statistic \(t:\calA_J^{\mathrm{obs}}\times\mathbb R^{|J^{\mathrm{obs}}|}\to\mathbb R\)
is effect increasing if, for every fixed \(a_J\) and outcome vectors \(y,y'\),
\[
\left.
\begin{array}{ll}
y'_j\ge y_j, & a_j=s,\\
y'_j\le y_j, & a_j=s',
\end{array}
\right\}\quad\text{for every }j\in J^{\mathrm{obs}}
\quad\Longrightarrow\quad
t(a_J,y')\ge t(a_J,y).
\]
\end{definition}
\end{samepage}
Thus raising an outcome carrying label \(s\), or lowering one carrying
\(s'\), cannot reduce the statistic. The difference-in-means statistic in \eqref{eq:stat-diffmean} satisfies
this property, as does its focal-count weighted version.

To see how the bounded-null test is constructed, consider the equality
null \(H_{\delta,\mathrm E}^{s,s'}: Y_i(0,s)-Y_i(0,s')= \delta\)  for all $i$. Under this null, the focal-cluster
outcome corresponding to a candidate label \(a_j\) is imputed as
\begin{equation}\label{eq:boundary-imputation}
\widetilde Y_{j,\delta}^{\mathrm{imp}}(a_j)
=
\widetilde Y_j^{\mathrm{obs}}
+\delta\bigl[\1\{a_j=s\}-\1\{A_j^{\mathrm{obs}}=s\}\bigr].
\end{equation}
Write \(\widetilde Y_{\delta}^{\mathrm{imp}}(a_J)\) for the vector of these
imputed means and \(\widetilde Y(a_J)\) for the corresponding vector of true
focal means, whose \(j\)th component is
\(k_j^{-1}\sum_{i\in U_j}Y_i(0,a_j)\).
Under \(H_{\delta,\mathrm B}^{s,s'}\), the imputation weakly raises
the true outcome for a cluster carrying candidate label \(s\) and weakly
lowers it for one carrying \(s'\). Hence, for any prespecified
effect-increasing statistic,
\begin{equation}\label{eq:boundary-statistic-dominance}
T_{U,\delta}^{\mathrm{full}}(a_J)
:=
t\bigl(a_J,\widetilde Y_{\delta}^{\mathrm{imp}}(a_J)\bigr)
\ge
t\bigl(a_J,\widetilde Y(a_J)\bigr).
\end{equation}
At the observed saturation labels the imputed and observed outcomes coincide.
The imputed reference statistics are therefore weakly larger than
their true counterparts, while the observed threshold is unchanged. This
argument leads to the conservative test established in
Theorem~\ref{thm:two-level-exact}. Our default statistic for the bounded null is
\begin{equation}\label{eq:bounded-common-shift}
T_{U,\delta}^{\mathrm{imp}}(a_J)
:=
\frac{1}{K_s}
\sum_{j\in J^{\mathrm{obs}}}
\widetilde Y_{j,\delta}^{\mathrm{imp}}(a_j)\1\{a_j=s\}
-
\frac{1}{K_{s'}}
\sum_{j\in J^{\mathrm{obs}}}
\widetilde Y_{j,\delta}^{\mathrm{imp}}(a_j)\1\{a_j=s'\}
\end{equation}

\begin{remark}[Heterogeneous bounds]\label{rem:heterogeneous-bounds}
For prespecified individual bounds \(\delta_i\), including cluster-specific
bounds \(\delta_{[i]}\), replace \(\delta\) in
\eqref{eq:boundary-imputation} by
\[
\bar\delta_{U,j}
:=
k_j^{-1}\sum_{i\in U_j}\delta_i.
\]
The same effect-increasing argument gives a conservative test. 
Proposition~\ref{prop:heterogeneous-bounds} in
Appendix~\ref{app:heterogeneous-bounds} gives the corresponding construction
and proof.
\end{remark}

\subsection{The conditional randomization test}\label{subsec:crt-procedure}

The following procedure summarizes the finite-sample conditional testing algorithm.

\begin{testprocedure}[Two-saturation conditional randomization test]\label{proc:two-level-crt}
Fix two saturation levels \(s\neq s'\) in \(\calS_0\), focal-set sizes \(\{k_j\}\) satisfying \eqref{eq:kj-condition}, and a statistic \(T_U(a_J)\) defined on \(\calA_J^{\mathrm{obs}}\).
\begin{enumerate}[(1)]
\item \textit{Select focal controls.}
For each \(j\in J^{\mathrm{obs}}\), sample \(k_j\) units uniformly without replacement from
\(
\{i\in\calI_j:D_i^{\mathrm{obs}}=0\}.
\)
Let \(U_j\) be the selected set in cluster \(j\), and let \(U=\cup_{j\in J^{\mathrm{obs}}}U_j\).

\item \textit{Construct the randomization space.}
Form \(\calA_J^{\mathrm{obs}}\) as in \eqref{eq:relabel-space}.

\item \textit{Evaluate the statistic.}
For each \(a_J\in\calA_J^{\mathrm{obs}}\), compute \(T_U(a_J)\).
For the partially sharp null, one may use \eqref{eq:stat-diffmean} or any other focal statistic.
For the bounded null, evaluate an effect-increasing statistic on the imputed outcomes in \eqref{eq:boundary-imputation}, such as \eqref{eq:bounded-common-shift}.

\item \textit{Compute the conditional \(p\)-value.}
For a one-sided alternative with large values of the statistic unfavorable to the null, set
\begin{equation}\label{eq:generic-crt-pvalue}
p_T
=
\frac{1}{|\calA_J^{\mathrm{obs}}|}
\sum_{a_J\in\calA_J^{\mathrm{obs}}}
\1\left\{
T_U(a_J)\ge T_U(A_J^{\mathrm{obs}})
\right\}.
\end{equation}
\end{enumerate}
\end{testprocedure}

The test can be implemented by enumeration or Monte Carlo. Draw
\(A_J^{(1)},\ldots,A_J^{(B)}\) independently and uniformly from
\(\calA_J^{\mathrm{obs}}\), retain repeated draws, and use
\begin{equation}\label{eq:crt-monte-carlo}
\widehat p_T=
\frac{1+\sum_{b=1}^B\1\{T_U(A_J^{(b)})\ge T_U(A_J^{\mathrm{obs}})\}}{B+1}.
\end{equation}
Equivalently, include the identity together with \(B\) independent uniform cluster permutations \citep[Remark~2.1]{RitzwollerRomanoShaikh2025}. For the bounded null, the same formula is conservative by the pointwise comparison in \eqref{eq:boundary-statistic-dominance}.

\begin{theorem}[Finite-sample conditional validity for two-saturation tests]\label{thm:two-level-exact}
Suppose Assumption~\ref{ass:hpi} holds. Fix distinct saturation levels \(s,s'\in\calS_0\), and let \(U\) be the focal set generated by Procedure~\ref{proc:two-level-crt}.

\begin{enumerate}[(a)]
\item \textit{Partially sharp null.}
For the partially sharp null in \eqref{eq:null-ps}, the permutation test in Procedure~\ref{proc:two-level-crt} is finite-sample valid for any statistic based only on the focal outcomes and saturation labels.
Specifically, for every \(\alpha\in[0,1]\),
\[
\Prob\{p_T\le \alpha\mid U\}\le \alpha
\qquad\text{under }H_{0,\mathrm{PS}}^{s,s'}.
\]

\item \textit{Bounded null.}
Let \(t\) be a prespecified effect-increasing statistic satisfying Definition~\ref{def:effect-increasing-focal}.
For the bounded null in \eqref{eq:null-bounded}, Procedure~\ref{proc:two-level-crt} with \(T_U=T_{U,\delta}^{\mathrm{full}}\) as defined in \eqref{eq:boundary-statistic-dominance} is finite-sample valid.
Specifically, for every \(\alpha\in[0,1]\),
\[
\Prob\{p_T\le \alpha\mid U\}\le \alpha
\qquad\text{under }H_{\delta,\mathrm B}^{s,s'}.
\]
\end{enumerate}
\end{theorem}

The two tests use the same focal selection and cluster permutations. Under equality, any focal statistic is valid; under the bounded null, validity uses an effect-increasing statistic evaluated at the upper-bound imputation. Appendix~\ref{app:stratified-crt} gives the stratified implementation used in the application.

\section{Multiple Saturation Levels and Monotone Effects}
\label{sec:monotone}

This section develops a finite-sample valid pairwise imputation randomization test (PIRT) of the increasing-monotonicity null \(H_{0,\mathrm M}(\calS_M)\) in \eqref{eq:null-monotone}, following \citet{zhong2024unconditional}.

\subsection{Testing the joint monotonicity restriction}
\label{subsec:monotone-null}

One approach is based on testing the adjacent inequalities
\[
Y_i(0,q_\ell)\le Y_i(0,q_{\ell+1}),
\qquad \ell=0,\ldots,M-1,
\]
and then combining the resulting \(p\)-values \citep[see, e.g.,][]{huang2025randomization}.
For related methods that combine \(p\)-values across repeated data splits, see \citet{Romano2020}.
With an appropriate combination rule, this also tests the joint monotonicity restriction. We instead construct a single pairwise-imputation test without forming separate contrast tests.

\subsection{Pairwise-imputable monotone statistics}
\label{subsec:monotone-statistics}

The test is based on pairwise comparisons between two assignments.
For any assignments \(z=(a,d)\) and \(z'=(a',d')\) in \(\calZ\), define
\begin{equation}\label{eq:pairwise-imputable-set}
\mathbb I_M(z,z')
=
\left\{
 i\in[N]:
 d_i=d_i'=0
 \text{ and }
 a_{[i]},a'_{[i]}\in\calS_M
\right\}.
\end{equation}
These units are the ones that are untreated under both assignments and whose two saturation exposures are both covered by the monotone null.
Units whose exposures fall outside \(\calS_M\) are excluded because the null imposes no restriction on their outcome ordering; a label excluded from common support may still leave untreated units in some clusters.

The following definition states the requirements on the statistic for any pair of assignments, following the pairwise-imputation approach of \citet{zhong2024unconditional}.

\begin{definition}[Pairwise-imputable spillover-monotone statistic]
\label{def:pairwise-monotone-stat}
A statistic
\[
T:\mathbb R^N\times\calZ\times\calZ\to\mathbb R\cup\{\infty\}
\]
is a pairwise-imputable spillover-monotone statistic for \(H_{0,\mathrm M}(\calS_M)\) if, for every pair \(z=(a,d)\) and \(z'=(a',d')\), the following two conditions hold.
\begin{enumerate}[(i)]
\item \textit{Pairwise imputability.}
If two outcome vectors \(y\) and \(y'\) agree on \(\mathbb I_M(z,z')\), then
\[
T(y,z,z')=T(y',z,z').
\]

\item \textit{Monotonicity.}
If, for every \(i\in\mathbb I_M(z,z')\),
\[
\begin{cases}
y_i\ge y_i', & \text{when } a_{[i]}>a'_{[i]},\\
y_i\le y_i', & \text{when } a_{[i]}<a'_{[i]},\\
y_i= y_i', & \text{when } a_{[i]}=a'_{[i]},
\end{cases}
\]
then
\[
T(y,z,z')\ge T(y',z,z').
\]
\end{enumerate}
\end{definition}

Condition (i) says that the statistic uses only units whose two untreated exposures are covered by the null.
Condition (ii) says that the statistic respects the direction of the monotone restriction.
When the saturation exposure is higher under \(z\) than under \(z'\), the corresponding outcome is weakly higher; when the saturation exposure is lower under \(z\), the corresponding outcome is weakly lower; and when the exposure is unchanged, the outcome is unchanged.
The last case is natural because equal exposure labels correspond to the same untreated exposure.
For a monotone decreasing null, the inequalities in Condition (ii) are reversed.

A simple example is a transformed difference in means.
Let \(\psi_1,\psi_0:\mathbb R\to\mathbb R\) be nondecreasing functions.
\begin{equation}\label{eq:transformed-diffmean}
T^{\mathrm{dm}}(y,z,z')
=
\frac{
\sum_{i\in\mathbb I_M(z,z')}
\1\{a_{[i]}> a'_{[i]}\}\psi_1(y_i)
}{
\sum_{i\in\mathbb I_M(z,z')}
\1\{a_{[i]}> a'_{[i]}\}
}
-
\frac{
\sum_{i\in\mathbb I_M(z,z')}
\1\{a_{[i]}<a'_{[i]}\}\psi_0(y_i)
}{
\sum_{i\in\mathbb I_M(z,z')}
\1\{a_{[i]}<a'_{[i]}\}
}
\end{equation}
with both scores \(T^{\mathrm{dm}}(y,z,z')\) and \(T^{\mathrm{dm}}(y,z',z)\) set to zero whenever either directional group is empty. When \(\psi_1(u)=\psi_0(u)=u\), this statistic compares outcomes for units whose saturation is higher under \(z\) than under \(z'\) with outcomes for units whose saturation is lower under \(z\) than under \(z'\). Appendix~\ref{app:pirt-statistics} gives normalized rank and adjacent-exposure alternatives.

A rank-based statistic can be defined similarly, following \citet{Caughey2023}.
Fix an ordering of the units before observing assignments or outcomes.
Let \(r_i(y_{\mathbb I_M(z,z')})\) be the integer rank of \(y_i\) among
\(\{y_\ell:\ell\in\mathbb I_M(z,z')\}\), breaking ties by this fixed unit
ordering restricted to the pairwise set. The same ordering is used for every
assignment pair and every outcome vector.\footnote{An alternative assigns
each tied observation the average of \(\varphi(r)\) over the integer rank
positions occupied by its tie group, rather than applying \(\varphi\) to the
average rank. This also preserves monotonicity; see
\citet[Supplementary Material, Section~A1.1]{Caughey2023}.
For the Wilcoxon score \(\varphi(r)=r\), the two averaging operations coincide.}
For a nondecreasing score function \(\varphi\), define
\begin{equation}\label{eq:rank-statistic}
T^{\mathrm{rk}}(y,z,z')
=
\sum_{i\in\mathbb I_M(z,z')}
\1\{a_{[i]}\ge a'_{[i]}\}
\varphi\!\left(r_i(y_{\mathbb I_M(z,z')})\right).
\end{equation}
The transformed difference-in-means statistic satisfies
Definition~\ref{def:pairwise-monotone-stat} directly by the monotonicity of
\(\psi_1\) and \(\psi_0\).
For the rank statistic, fix \(z,z'\) and arrange the ranks occupied by the
weakly higher group in increasing order. Raising outcomes in that group and
lowering outcomes in the lower group cannot decrease any entry of this
ordered rank vector. Since \(\varphi\) is nondecreasing, the group's rank-score
sum cannot decrease. Thus both \(T^{\mathrm{dm}}\) and \(T^{\mathrm{rk}}\) are
pairwise-imputable spillover-monotone statistics.

The definition implies the following pairwise ordering.

\begin{proposition}[Pairwise ordering under monotonicity]
\label{prop:pairwise-ordering}
Suppose Assumption~\ref{ass:hpi} holds and the monotone null \(H_{0,\mathrm M}(\calS_M)\) in \eqref{eq:null-monotone} is true.
If \(T\) is a pairwise-imputable spillover-monotone statistic, then, for every \(z,z'\in\calZ\),
\[
T\bigl(Y(z),z,z'\bigr)
\ge
T\bigl(Y(z'),z,z'\bigr).
\]
\end{proposition}

\subsection{PIRT for the monotone null}
\label{subsec:pirt}

Proposition~\ref{prop:pairwise-ordering} leads to an unconditional randomization test following \citet{zhong2024unconditional}.
For the realized assignment \(Z^{\mathrm{obs}}\), define
\begin{equation}\label{eq:pirt-pvalue}
p_M^{\mathrm{PIRT}}(Z^{\mathrm{obs}})
=
\sum_{z\in\calZ}
\1\!\bigl[
T\bigl(Y^{\mathrm{obs}},Z^{\mathrm{obs}},z\bigr)
\ge
T\bigl(Y^{\mathrm{obs}},z,Z^{\mathrm{obs}}\bigr)
\bigr]
P(z).
\end{equation}
Every term in \eqref{eq:pirt-pvalue} is observable because the statistic is pairwise imputable: for each reference assignment \(z\), both statistics depend only on units in \(\mathbb I_M(Z^{\mathrm{obs}},z)\).
When \(|\calZ|\) is too large for enumeration, the sum in \eqref{eq:pirt-pvalue} can be approximated by Monte Carlo draws from the known design distribution \(P\).
Appendix~\ref{app:conditional-pirt} gives a finite-sample valid plus-one Monte Carlo implementation, including its factor-of-two adjustment. An unusually small comparison fraction is evidence against increasing monotonicity.

\begin{testprocedure}[PIRT for the monotone spillover null]
\label{proc:pirt}
Fix an ordered saturation set \(\calS_M\subseteq\calS_0\) and a pairwise-imputable spillover-monotone statistic \(T\).
\begin{enumerate}[(1)]
\item Draw or enumerate reference assignments \(z\) from the design distribution \(P\).
\item For each reference assignment, compute
\[
h(z;Z^{\mathrm{obs}},Y^{\mathrm{obs}})
=
\1\!\bigl[
T\bigl(Y^{\mathrm{obs}},Z^{\mathrm{obs}},z\bigr)
\ge
T\bigl(Y^{\mathrm{obs}},z,Z^{\mathrm{obs}}\bigr)
\bigr].
\]
\item Compute
\[
p_M^{\mathrm{PIRT}}(Z^{\mathrm{obs}})
=
\sum_{z\in\calZ}
h(z;Z^{\mathrm{obs}},Y^{\mathrm{obs}})P(z),
\]
or the plus-one Monte Carlo version in \eqref{eq:conditional-pirt-mc} with the original design law.
\item Reject \(H_{0,\mathrm M}(\calS_M)\) at level \(\alpha\) when
\[
p_M^{\mathrm{PIRT}}(Z^{\mathrm{obs}})\le \alpha/2.
\]
Equivalently, \(\min\{2p_M^{\mathrm{PIRT}},1\}\) is the reported finite-sample valid \(p\)-value under exact enumeration.
\end{enumerate}
\end{testprocedure}

\begin{theorem}[Finite-sample validity of monotone PIRT]
\label{thm:pirt-valid}
Suppose Assumption~\ref{ass:hpi} holds and the monotone null \(H_{0,\mathrm M}(\calS_M)\) in \eqref{eq:null-monotone} is true.
If \(T\) is a pairwise-imputable spillover-monotone statistic, then the rejection rule in Procedure~\ref{proc:pirt} satisfies
\[
\E_P\!
\left(
\1\{p_M^{\mathrm{PIRT}}(Z^{\mathrm{obs}})\le \alpha/2\}
\right)
\le \alpha
\qquad
\text{for every } \alpha\in(0,1).
\]
Thus Procedure~\ref{proc:pirt} is a finite-sample valid unconditional randomization test of the monotone spillover null over \(\calS_M\).
\end{theorem}

The factor-of-two adjustment converts the pairwise ordering into finite-sample size control; the resulting test is conservative in general.

\begin{remark}[Conditioning on labels outside the tested support]
\label{rem:conditional-pirt}
The same PIRT construction can use a conditional design law. For a
prespecified $\calS_M$, let
\[
\Gamma_{\mathrm{out}}(z)=\{(j,a_j):a_j\notin\calS_M\},
\qquad
P_g(\cdot)=P\{\cdot\mid\Gamma_{\mathrm{out}}(Z)=g\}.
\]
The observed assignment selects $g=\Gamma_{\mathrm{out}}(Z^{\mathrm{obs}})$.
Replace $P$ by $P_g$ in \eqref{eq:pirt-pvalue}; clusters outside the tested
support retain their identities and labels, while the remaining labels
and within-cluster assignments follow their conditional design law.
Corollary~\ref{cor:conditional-monotone-pirt} proves conditional validity
with the same factor-of-two adjustment, following the partition formulation
of conditional PIRT in \citet[Section~4.2]{liu2026randomizationtestsswitchbackexperiments}.
\end{remark}

\begin{remark}[General ordered-exposure nulls]\label{rem:general-ordered-exposure}
The PIRT construction does not rely on any special feature of the randomized saturation design beyond the exposure mapping and the known assignment distribution.
More generally, suppose an experiment has assignment space \(\calZ\), assignment law \(P\), and exposure mapping \(e_i:\calZ\to\mathcal E\).
Let \(\mathcal E_0\subseteq\mathcal E\) be the subset of exposures on which the null is stated, and let \(\preceq\) be a partial order on \(\mathcal E_0\).
Consider the ordered-exposure null
\[
Y_i(e)\le Y_i(e')
\quad\text{for all }i\text{ and all }e,e'\in\mathcal E_0\text{ with }e\preceq e'.
\]
For a pair \(z,z'\), the pairwise imputable set is
\[
\mathbb I_{\preceq}(z,z')
=
\left\{
 i\in[N]:
 e_i(z),e_i(z')\in\mathcal E_0
 \text{ and }
 \bigl(e_i(z)\preceq e_i(z')\text{ or }e_i(z')\preceq e_i(z)\bigr)
\right\}.
\]
Units whose two exposures are not both covered by the null, or whose exposures are incomparable, are excluded from the pairwise comparison.
Any statistic satisfying the analogues of Definition~\ref{def:pairwise-monotone-stat} with respect to \(\mathbb I_{\preceq}(z,z')\) yields the same PIRT validity argument.
The saturation-design null in \eqref{eq:null-monotone} is the special case \(e_i(z)=(D_i,A_{[i]})\), 
\(
\mathcal E_0=\{(0,q):q\in\calS_M\},
\)
and
\(
(0,q)\preceq (0,q')
\Longleftrightarrow
q\le q'.
\)
\end{remark}

\section{Weak Average Nulls as an Asymptotic Extension}\label{sec:weak-null}

Sections~\ref{sec:two-level-crt} and~\ref{sec:monotone} establish finite-sample validity under individual-level restrictions.
We now test the weighted average null \(H_{0,\mathrm W}^{s,s'}(\lambda)\) in \eqref{eq:null-weak-weighted}, which relaxes individual equality but does not determine the missing focal potential outcomes.
Consequently, the conditional permutation distribution is no longer an exact finite-sample null distribution. However, it can be used with a studentized statistic to provide unconditional asymptotic size control.

\subsection{Studentized statistic and asymptotic validity}\label{subsec:weak-statistic}

Fix \(s\neq s'\) in \(\calS_0\), and use the focal set and permutation space from Section~\ref{subsec:conditioning-event}.
For the weighted weak null \(H_{0,\mathrm W}^{s,s'}(\lambda)\), define
\[
w_{jK}:=K\lambda_{jK},
\qquad
X_j^{\mathrm{obs}}:=w_{jK}\widetilde Y_j^{\mathrm{obs}}.
\]
The normalization \(w_{jK}=K\lambda_{jK}\) is convenient because \(K^{-1}\sum_{j=1}^K w_{jK}=1\).

For \(q\in\{s,s'\}\) and \(a_J\in\calA_J^{\mathrm{obs}}\), define
\[
\bar X_U(q;a_J)
=
\frac{1}{K_q}
\sum_{j\in J^{\mathrm{obs}}}
X_j^{\mathrm{obs}}\1\{a_j=q\},
\]
and, when \(K_q\ge 2\),
\[
\hat S_{X,U}^2(q;a_J)
=
\frac{1}{K_q-1}
\sum_{j\in J^{\mathrm{obs}}}
\bigl(X_j^{\mathrm{obs}}-\bar X_U(q;a_J)\bigr)^2
\1\{a_j=q\}.
\]
The Neyman variance estimator for the transformed focal-cluster outcomes is
\[
\hat V_{X,U}^{\mathrm{Ney}}(a_J)
=
\frac{\hat S_{X,U}^2(s;a_J)}{K_s}
+
\frac{\hat S_{X,U}^2(s';a_J)}{K_{s'}}.
\]
The studentized statistic is
\begin{equation}\label{eq:stat-studentized}
T_{U,\lambda}^{\mathrm{Ney}}(a_J)
=
\frac{\bar X_U(s;a_J)-\bar X_U(s';a_J)}
{\sqrt{\hat V_{X,U}^{\mathrm{Ney}}(a_J)}}.
\end{equation}
When the denominator is zero, we use a fixed deterministic convention, such as setting the statistic to zero.
Assumption~\ref{ass:weak-regularity} in Appendix~\ref{app:weak-regularity}
rules out this degeneracy in the limit.
The corresponding upper-tail \(p\)-value is defined by \eqref{eq:generic-crt-pvalue} with \(T_U=T_{U,\lambda}^{\mathrm{Ney}}\); write this \(p\)-value as \(p_{T,\lambda}\).

The equally weighted cluster-average weak null corresponds to \(w_{jK}=1\), in which case \(T_{U,\lambda}^{\mathrm{Ney}}\) reduces to the unweighted studentized statistic.
The unit-average weak null corresponds to
\[
w_{jK}=\frac{n_j}{N/K}.
\]
Since multiplication of all transformed cluster outcomes by a common positive constant cancels after studentization, the unit-average version can equivalently be implemented by replacing each focal-cluster mean by \(n_j\widetilde Y_j^{\mathrm{obs}}\).

\begin{theorem}[Unconditional asymptotic validity for weighted weak average nulls]\label{thm:weak-null-asymptotic}
Fix \(s\neq s'\) in \(\calS_0\), and let \(\lambda_{1K},\ldots,\lambda_{KK}\) be nonnegative, nonrandom weights summing to one, and let \(p_{T,\lambda}\) be the \(p\)-value defined by \eqref{eq:generic-crt-pvalue} with studentized statistic \(T_{U,\lambda}^{\mathrm{Ney}}\) in \eqref{eq:stat-studentized}.
If the weighted weak null \(H_{0,\mathrm W}^{s,s'}(\lambda)\) in \eqref{eq:null-weak-weighted} is true and Assumption~\ref{ass:weak-regularity} holds, then
\[
\limsup_{K\to\infty}
\Prob\{p_{T,\lambda}\le \alpha\}
\le \alpha
\qquad
\text{for every } \alpha\in(0,1/2).
\]
\end{theorem}

The guarantee averages over focal selection: a zero population average effect need not remain zero in the realized focal sample. For unequal weights, we recommend the centered implementation in Appendix~\ref{app:weak-centering} when the conditions of Proposition~\ref{prop:weak-centered} hold. It is invariant to adding a common constant to all outcomes. With equal weights, centering leaves the test unchanged.

\section{Empirical Application: Spillovers in the Zomba Cash Transfer Experiment}
\label{sec:application_malawi}
We illustrate the methods using the Zomba Cash Transfer Program in Malawi \citep{baird2011cash,Baird2018}. We first describe the design and spillover hypotheses, then examine the tests in application-calibrated simulations, and finally apply them to the observed outcomes.

\subsection{Design, exposure labels, and hypotheses}
\label{subsec:malawi_design_hypotheses}
The clusters are enumeration areas (EAs), and the units are girls who were in school at baseline. An EA's saturation determines the proportion of schoolgirls randomly offered transfers. The original UCT-plus-zero roster contains 665 schoolgirls in 42 EAs at saturation levels \(0,0.33,0.66,1\). We use the 664 schoolgirls in the 41 EAs that leave at least one schoolgirl untreated at each tested offer fraction.\footnote{One singleton EA has no untreated exposure at saturation \(0.66\). We exclude it on the basis of its roster size and condition on its observed label when forming the retained randomization distribution; see Appendix~\ref{app:malawi-implementation}.} Table~\ref{tab:malawi_design_cells} describes the retained roster.

\begin{table}[t]
\centering
\begin{threeparttable}
\caption{Retained application sample by saturation level}
\label{tab:malawi_design_cells}
\begin{tabular}{@{}llrrr@{}}
\toprule
Label & Schoolgirl policy & EAs & Offered & Untreated \\
\midrule
$\ell_{0}$ & No schoolgirl offer & 14 & 0 & 200 \\
$\ell_{0.33}$ & UCT, $33\%$ & 9 & 68 & 135 \\
$\ell_{0.66}$ & UCT, $66\%$ & 9 & 87 & 44 \\
$\ell_{1}$ & UCT, $100\%$ & 9 & 130 & 0 \\

\bottomrule
\end{tabular}
\begin{tablenotes}\normalsize\setstretch{1.72}
\item Counts include missing-outcome rows; the singleton EA is excluded.

\end{tablenotes}
\end{threeparttable}
\end{table}

The experiment used three geographic strata: Zomba Town, nearby rural areas, and more distant rural areas \citep{baird2011cash}. We permute saturation labels within each stratum, preserving its retained observed margins. Conditional PIRT also fixes the full-offer EA labels, as in Remark~\ref{rem:conditional-pirt}. Appendix~\ref{app:malawi-implementation} gives the assignment distribution and implementation details. The application maintains the stratified, fixed-count assignment law specified there and assumes that outcome-observation indicators are unchanged over the reference assignments. The finite-sample guarantees apply under these maintained assumptions and Assumption~\ref{ass:hpi}.

Let \(\ell_0\) denote no schoolgirl offer and \(\ell_s\) the UCT schoolgirl-offer label at saturation \(s\in\{0.33,0.66,1\}\). Under Assumption~\ref{ass:hpi}, write schoolgirl \(i\)'s potential outcome as \(Y_i(d,\ell_s)\). Each outcome has 639 observed values in the retained roster. For its observed-case set \(\calI_{\mathrm{app}}\), we test
\[
H_{0,\mathrm B}^{s,s'}:\quad
Y_i(0,\ell_s)\le Y_i(0,\ell_{s'})
\quad\text{for all }i\in\calI_{\mathrm{app}}\text{ with both exposures feasible},
\]
for \((s,s')=(0.33,0),(0.66,0.33),(0.66,0)\). A small \(p\)-value indicates a positive spillover for at least one girl. For every \(i\in\calI_{\mathrm{app}}\), we also test \(Y_i(0,\ell_q)\le Y_i(0,\ell_{q'})\) for all \(q<q'\) in \(\{0,0.33,0.66\}\) with both untreated exposures supported by the conditional reference design. These nulls concern opposite directions. The outcomes are current enrollment, ability to read an English letter, ever married, and ever pregnant, in their original coding.

The CRT uses the effect-increasing coefficient from regressing EA focal means on the higher-saturation indicator and geographic fixed effects, weighted by observed focal counts. PIRT uses the strict changed-exposure mean contrast in Appendix~\ref{app:pirt-statistics}; we also report adjacent CRTs combined by Bonferroni.

\subsection{Calibrated simulations}
\label{subsec:malawi_simulation}
The simulations retain the roster, geographic structure, and outcome-observation indicators. A binomial-logit model of untreated outcomes includes saturation-cell fixed effects and Gaussian EA random intercepts. Drawing \(b_{jk}^*\) from its fitted distribution and removing the saturation effect gives
\begin{equation}\label{eq:sim-potential-outcome}
Y_i^k(0,q)=\1\!\left\{U_i^k\le
\operatorname{expit}(\widehat\alpha_k+b_{[i]k}^*+\gamma_{[i]k}(q))\right\},
\qquad U_i^k\stackrel{\mathrm{iid}}{\sim}\operatorname{Unif}(0,1).
\end{equation}
Shared unit uniforms and equal shifts at identical treated counts enforce exposure consistency. Schedules are fixed before assignment; Appendix~\ref{app:malawi-calibration} gives their calibration and feasible grids.

We assess size at individual equality \eqref{eq:null-ps}; Appendix~\ref{app:malawi-robustness} also checks bounded spillovers \eqref{eq:null-bounded} at \(\delta=0\) and increasing monotonicity \eqref{eq:null-monotone}. Our regression benchmark follows the linear probability approach with EA-clustered standard errors used by \citet{baird2011cash}, adapted to saturation comparisons as in \citet{Baird2018}. We regress untreated outcomes on the higher-saturation indicator and geographic fixed effects, using EA-clustered CR1 standard errors; Appendix~\ref{app:regression-benchmark} gives matched-focal and rounded-count diagnostics.

Each configuration uses 2,000 outer replications at level 5\%, with 1,999 independent reference draws plus the observed comparison per sampled test. Additional results appear in Appendix~\ref{app:malawi-robustness}.

\begin{table}[t]
\centering
\begin{threeparttable}
\caption{Rejection frequencies under individual no-spillover equality}
\label{tab:malawi_pairwise_size}
\begin{tabular}{@{}lcc@{}}
\toprule
Outcome & CRT & LPM--CR1 \\
\midrule
Currently enrolled & 4.98 (0.30) & 7.67 (0.39) \\
Reads an English letter & 5.25 (0.30) & 12.28 (0.49) \\
Ever married & 4.75 (0.29) & 6.08 (0.34) \\
Ever pregnant & 4.68 (0.30) & 5.95 (0.34) \\

\bottomrule
\end{tabular}
\begin{tablenotes}\normalsize\setstretch{1.72}
\item Bounded-null tests are evaluated under \eqref{eq:null-ps}. Entries give percent rejection at nominal 5\%, averaged over three contrasts within replication; parentheses give Monte Carlo standard errors.

\end{tablenotes}
\end{threeparttable}
\end{table}

Table~\ref{tab:malawi_pairwise_size} first assesses size: CRT rejection frequencies range from 4.68\% to 5.25\%, whereas LPM--CR1 overrejects, particularly for enrollment and literacy. Figure~\ref{fig:malawi-pairwise-power} then compares power for marriage and pregnancy, where excess rejection is smaller, including CRTs with maximal and single focal samples.

\begin{figure}[t]
\centering
\includegraphics[width=\textwidth,height=0.32\textheight,keepaspectratio]{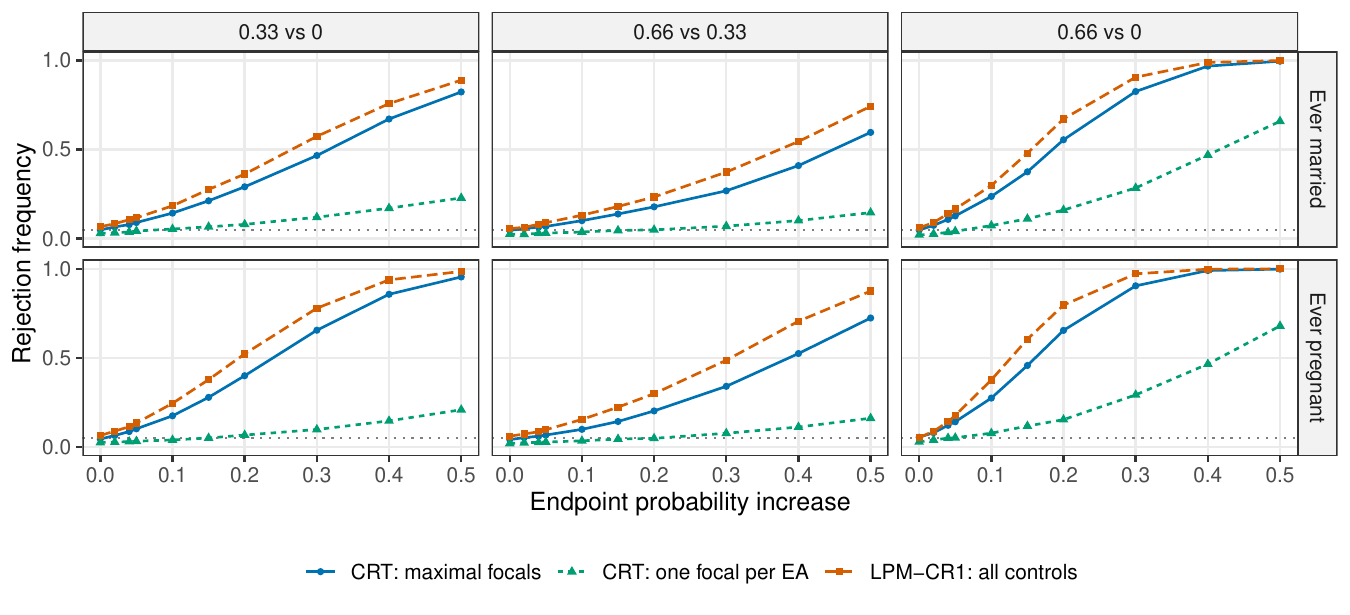}
\caption{Pairwise bounded-null rejection frequencies for marriage and pregnancy. Columns show the three saturation contrasts; magnitude is the probability increase between 0 and 0.66. All procedures use geographic fixed effects and observation-count-weighted EA means, varying their included controls and inference method. Further outcomes, sizes, and a fixed-equal-weight focal comparison appear in Appendix~\ref{app:malawi-robustness}.}
\label{fig:malawi-pairwise-power}
\end{figure}

For monotonicity, Figure~\ref{fig:malawi-monotone-power} compares PIRT, adjacent CRT--Bonferroni, the endpoint CRT, and all-pairs CRT--Bonferroni under diffuse declines. The three-level design follows the application, while the four-level design distributes the partial-offer slots across three positive saturation levels. In each case, the target average decline is spread equally across adjacent levels.

\begin{figure}[t]
\centering
\includegraphics[width=\textwidth,height=0.32\textheight,keepaspectratio]{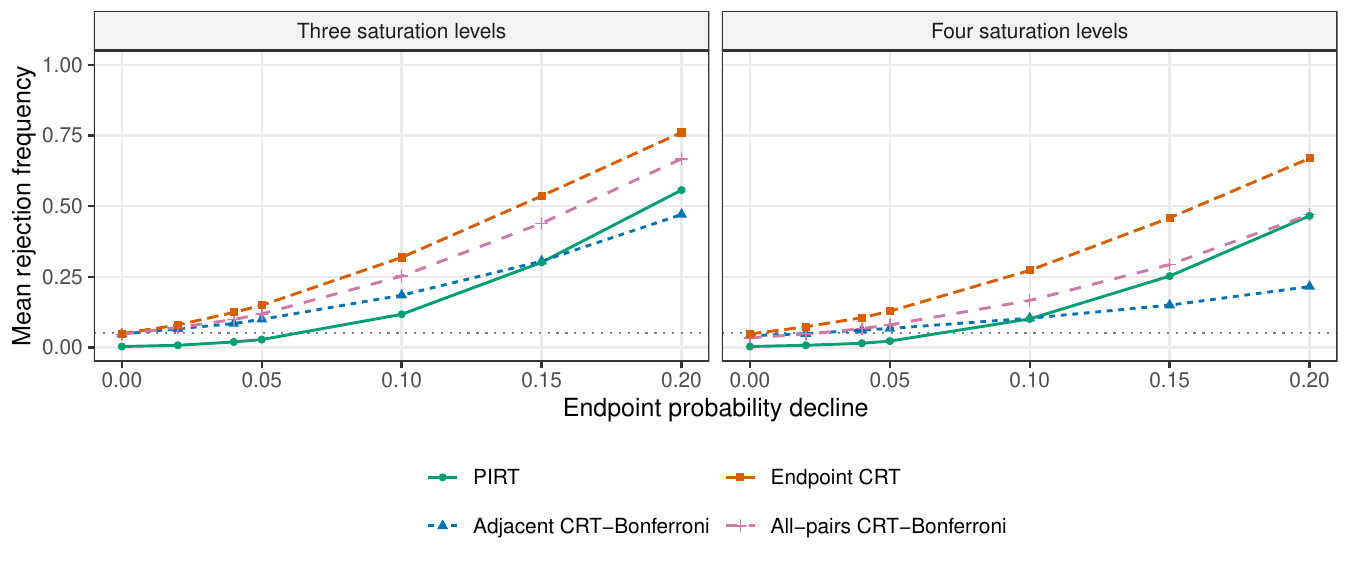}
\caption{Tests of increasing monotonicity under diffuse declines, averaged over the same four outcomes. The panels use three levels, $\{0,.33,.66\}$, and four levels, $\{0,.22,.44,.66\}$. Magnitude is the average probability decline from the lowest to the highest level. Appendix~\ref{app:malawi-robustness} reports localized, hump, and trough alternatives and outcome-specific comparisons.}
\label{fig:malawi-monotone-power}
\end{figure}

For the four-level diffuse decline at magnitude 0.20, rejection frequencies are 46.6\% for PIRT, 21.5\% for adjacent CRT--Bonferroni, 66.9\% for the endpoint CRT, and 47.0\% for all-pairs CRT--Bonferroni. PIRT's advantage over adjacent CRT--Bonferroni is therefore specific to the alternative, not uniform across CRT-based procedures or response shapes. The supplementary hump and trough experiments favor CRT-based procedures; their equal endpoints also illustrate why an endpoint test alone need not detect an interior violation of monotonicity.

Appendix~\ref{app:weak-numerics} separately evaluates studentized tests of weighted weak nulls.

\subsection{Application to observed outcomes}
\label{subsec:malawi_observed}
Table~\ref{tab:malawi_application_results} reports focal coefficients with one-sided CRT \(p\)-values below, followed by the two joint monotonicity \(p\)-values. Appendix~\ref{app:malawi-robustness} reports descriptive means.

\begin{table}[t]
\centering
\begin{threeparttable}
\caption{Focal-regression coefficients and randomization \(p\)-values}
\label{tab:malawi_application_results}
\setlength{\tabcolsep}{3.5pt}
\begin{tabular*}{\textwidth}{@{\extracolsep{\fill}}lccccc@{}}
\toprule
&\multicolumn{3}{c}{Two-saturation non-increase bounds}&\multicolumn{2}{c}{Increasing monotonicity}\\
\cmidrule(lr){2-4}\cmidrule(l){5-6}
Outcome & $0.33$ vs. $0$ & $0.66$ vs. $0.33$ & $0.66$ vs. $0$ & PIRT & CRT--Bonferroni\\
\midrule
Currently enrolled & -0.047 & 0.073 & -0.001 & 0.822 & 0.556 \\
 & (0.723) & (0.261) & (0.527) & & \\
Reads an English letter & -0.044 & 0.096 & 0.018 & 0.532 & 0.172 \\
 & (0.915) & (0.103) & (0.349) & & \\
Ever married & 0.001 & -0.051 & -0.055 & 0.846 & 0.590 \\
 & (0.480) & (0.708) & (0.702) & & \\
Ever pregnant & -0.007 & -0.136 & -0.105 & 0.646 & 0.101 \\
 & (0.522) & (0.953) & (0.901) & & \\

\bottomrule
\end{tabular*}
\begin{tablenotes}\normalsize\setstretch{1.72}
\item Parentheses contain one-sided CRT \(p\)-values, not standard errors. PIRT values include the factor-two adjustment. Pairwise values are not multiplicity-adjusted. Implementation details appear in Appendix~\ref{app:malawi-implementation}.
\end{tablenotes}
\end{threeparttable}
\end{table}

No test rejects at 5\%; the pregnancy adjacent-CRT Bonferroni value is 0.101. Nonrejection establishes neither absent spillovers nor increasing monotonicity. The pairwise and joint tests concern opposite directions.

\section*{Data Availability Statement}
\noindent The data used in the empirical application are publicly available in Harvard Dataverse at \href{https://doi.org/10.7910/DVN/7K1A25}{doi:10.7910/DVN/7K1A25}.

\spacingset{1.5}
\putbib[biblio]
\end{bibunit}

\spacingset{1.72}
\newpage

\appendix
\spacingset{1.5}
\small
\begin{bibunit}[ims]
\section{Extensions to Other Exposure Contrasts and Global No-Spillover Nulls}
\label{app:other-estimands}

The main text focuses on spillover contrasts between two untreated exposures,
$Y_i(0,s)-Y_i(0,s')$.  This focus is useful because untreated focal units lead
to a particularly transparent permutation test, but it is not a restriction of
the general conditioning idea.  This section records extensions to other
supported exposure contrasts and then considers the joint null of no spillover
effects across all feasible untreated saturation levels.  The discussion is
stated for individual-level sharp and bounded nulls.  Average-effect nulls are
addressed separately below because, as in Section~\ref{sec:weak-null}, they do
not generally yield finite-sample exact tests.

\subsection{Pairwise contrasts between general supported exposures}
\label{app:general-pairwise-exposures}

For an exposure $e=(r,s)$, where $r\in\{0,1\}$ and $s\in\calS$, define the
number of units in cluster $j$ that can have own-treatment status $r$ under
saturation $s$ by
\begin{equation}
\label{eq:app-exposure-count}
c_j(r,s)
:=
 r\,m_j(s)+(1-r)\{n_j-m_j(s)\}.
\end{equation}
Thus $c_j(1,s)=m_j(s)$ and $c_j(0,s)=n_j-m_j(s)$.  The exposure $(r,s)$ is
feasible for every unit in the common-support sense used here when
$c_j(r,s)>0$ for every cluster $j$.

Fix two exposures
\[
e=(r,s),
\qquad
e'=(r',s'),
\qquad s\neq s',
\]
that are commonly feasible.  Write $Y_i(e)=Y_i(r,s)$ and
$Y_i(e')=Y_i(r',s')$.  The corresponding individual-level equality and
bounded nulls are
\begin{align}
H_{0,\mathrm{PS}}^{e,e'}:
&\qquad
Y_i(r,s)=Y_i(r',s')
\quad\text{for all }i,
\label{eq:app-general-pair-sharp}
\\
H_{\delta,\mathrm B}^{e,e'}:
&\qquad
Y_i(r,s)-Y_i(r',s')\leq\delta
\quad\text{for all }i.
\label{eq:app-general-pair-bounded}
\end{align}
The original untreated-spillover null is the special case
$(r,r')=(0,0)$.  The same construction also covers treated-spillover and
mixed own-treatment contrasts and specializes the general exposure-contrast
framework of \citet{Basse2019} to the randomized saturation design.

Let
\[
J^{\mathrm{obs}}_{s,s'}
=
\{j:A_j^{\mathrm{obs}}\in\{s,s'\}\},
\]
and define
\[
K_q=\sum_{j\in J^{\mathrm{obs}}_{s,s'}}\1\{A_j^{\mathrm{obs}}=q\},
\qquad q\in\{s,s'\}.
\]
Choose integers $k_j$ before observing the assignment such that
\begin{equation}
\label{eq:app-general-kj}
1\leq k_j
\leq
\min\{c_j(r,s),c_j(r',s')\}.
\end{equation}
For a contrasted cluster $j$, define the own-treatment status associated with
its observed exposure by
\[
r_j^{\mathrm{obs}}
=
 r\,\1\{A_j^{\mathrm{obs}}=s\}
+r'\,\1\{A_j^{\mathrm{obs}}=s'\}.
\]
Thus, from a cluster observed at label $s$, we select focal units having
$D_i^{\mathrm{obs}}=r$; from a cluster observed at label $s'$, we select focal
units having $D_i^{\mathrm{obs}}=r'$.  Condition~\eqref{eq:app-general-kj}
ensures that the same number of focal units can be supported under either
exposure.

For each $j\in J^{\mathrm{obs}}_{s,s'}$, sample
\[
U_j\subseteq
\{i\in\calI_j:D_i^{\mathrm{obs}}=r_j^{\mathrm{obs}}\},
\qquad |U_j|=k_j,
\]
uniformly without replacement, independently across contrasted clusters, and
let $U=\cup_j U_j$.  Define the observed focal-cluster mean
\begin{equation}
\label{eq:app-general-focal-mean}
\widetilde Y_j^{\mathrm{obs}}
=
\frac{1}{k_j}\sum_{i\in U_j}Y_i^{\mathrm{obs}}.
\end{equation}
The randomization space is the same fixed-margin space as in the main text,
\begin{equation}
\label{eq:app-general-relabel-space}
\calA_{J}^{\mathrm{obs}}
=
\left\{
 a_J\in\{s,s'\}^{|J|}:
 \sum_{j\in J}\1\{a_j=s\}=K_s
\right\},
\qquad J=J^{\mathrm{obs}}_{s,s'}.
\end{equation}
Under a reference permutation, a focal unit in a cluster relabeled to $s$ is
required to have own-treatment status $r$, whereas a focal unit in a cluster
relabeled to $s'$ is required to have own-treatment status $r'$.  Formally,
the conditional assignment set contains the second-stage assignments
satisfying these focal-status requirements and the fixed treated count
$m_j(a_j)$ in every cluster.  These second-stage assignments can be integrated
out: the induced conditional law of $A_J$ is uniform over
$\calA_J^{\mathrm{obs}}$, so implementation only requires the cluster-level
relabeling.

For an upper-tail comparison in the direction $e-e'$, a natural statistic is
\begin{equation}
\label{eq:app-general-diffmean}
\widehat\tau_U^{e,e'}(a_J)
=
\frac{1}{K_s}
\sum_{j\in J^{\mathrm{obs}}}\widetilde Y_j^{\mathrm{obs}}\1\{a_j=s\}
-
\frac{1}{K_{s'}}
\sum_{j\in J^{\mathrm{obs}}}\widetilde Y_j^{\mathrm{obs}}\1\{a_j=s'\}.
\end{equation}
For a two-sided equality test, one may instead use
$|\widehat\tau_U^{e,e'}(a_J)|$, a rank statistic, or any other prespecified
focal statistic.

\begin{testprocedure}[Conditional randomization test for a general exposure pair]
\label{proc:app-general-pair-crt}
Fix two commonly feasible exposures $e=(r,s)$ and $e'=(r',s')$ with
$s\neq s'$, and choose $k_j$ satisfying \eqref{eq:app-general-kj}.
\begin{enumerate}[(1)]
\item \textit{Select exposure-specific focal units.}
For every $j\in J^{\mathrm{obs}}_{s,s'}$, select $k_j$ units uniformly from
those with $D_i^{\mathrm{obs}}=r$ when $A_j^{\mathrm{obs}}=s$, and from those
with $D_i^{\mathrm{obs}}=r'$ when $A_j^{\mathrm{obs}}=s'$.

\item \textit{Construct the permutation space.}
Form $\calA_J^{\mathrm{obs}}$ in
\eqref{eq:app-general-relabel-space}.  Under a permutation, focal units are
assigned own-treatment status $r$ in clusters labeled $s$ and status $r'$ in
clusters labeled $s'$; the remaining second-stage assignments are integrated
out.

\item \textit{Evaluate the statistic.}
For each $a_J\in\calA_J^{\mathrm{obs}}$, compute a focal statistic such as
\eqref{eq:app-general-diffmean}.  Under the bounded null \eqref{eq:app-general-pair-bounded},
impute the focal-cluster means as
\[
\widetilde Y_{j,\delta}^{\mathrm{imp}}(a_j)
=
\widetilde Y_j^{\mathrm{obs}}
+\delta\bigl[
\1\{a_j=s\}-\1\{A_j^{\mathrm{obs}}=s\}
\bigr].
\]
Evaluate a prespecified effect-increasing statistic on these
imputed means, such as the difference in means in
\eqref{eq:app-general-diffmean} with
\(\widetilde Y_j^{\mathrm{obs}}\) replaced by
\(\widetilde Y_{j,\delta}^{\mathrm{imp}}(a_j)\).

\item \textit{Compute the conditional $p$-value.}
For an upper-tail test, set
\[
p_{e,e'}
=
\frac{1}{|\calA_J^{\mathrm{obs}}|}
\sum_{a_J\in\calA_J^{\mathrm{obs}}}
\1\!
\left\{
T_U(a_J)\geq T_U(A_J^{\mathrm{obs}})
\right\}.
\]
Enumeration may be replaced by uniform Monte Carlo relabeling with the usual
plus-one correction.
\end{enumerate}
\end{testprocedure}

The uniform relabeling law follows from the same cancellation used in
Lemma~\ref{lem:conditional-relabeling}.  To see the key calculation, fix a
cluster $j$, a candidate label $q\in\{s,s'\}$, and let $r(q)=r$ if $q=s$ and
$r(q)=r'$ if $q=s'$.  For any fixed subset $u\subseteq\calI_j$ of size $k_j$,
the probability that the second-stage randomization makes every member of $u$
have status $r(q)$ and that the focal-sampling rule then selects exactly $u$
is
\begin{align}
&\binom{n_j}{m_j(q)}^{-1}
\binom{n_j-k_j}{m_j(q)-r(q)k_j}
\binom{c_j(r(q),q)}{k_j}^{-1}
\notag\\
&\hspace{3cm}=
\binom{n_j}{k_j}^{-1}.
\label{eq:app-general-cancellation}
\end{align}
The final expression does not depend on either $q$ or $r(q)$.  Consequently,
conditioning on the selected focal sets does not favor one placement of the
two saturation labels over another.  Under
$H_{0,\mathrm{PS}}^{e,e'}$, every focal outcome is invariant because its
exposure always lies in $\{e,e'\}$.  The resulting randomization $p$-value is
therefore finite-sample valid.  Under $H_{\delta,\mathrm B}^{e,e'}$, the
equality $Y_i(e)-Y_i(e')=\delta$ is least favorable for the same
upper-tail, effect-increasing statistics used in Section~\ref{sec:two-level-crt};
hence the shifted version is finite-sample conservative.  These statements
follow by the same argument as Theorem~\ref{thm:two-level-exact}, with no new
substantive assumption.

The weighted weak-null extension also carries over.  Define
\[
\mu_{e,j}
=
\frac{1}{n_j}\sum_{i\in\calI_j}Y_i(r,s),
\qquad
\mu_{e',j}
=
\frac{1}{n_j}\sum_{i\in\calI_j}Y_i(r',s'),
\]
and consider
\[
H_{0,\mathrm W}^{e,e'}(\lambda):
\qquad
\sum_{j=1}^K\lambda_{jK}
\{\mu_{e,j}-\mu_{e',j}\}=0.
\]
Equation~\eqref{eq:app-general-cancellation} implies an equivalent sequential
representation in which the focal subset in each contrasted cluster is drawn
uniformly from all $k_j$-subsets before the two labels are permuted.  Thus the
studentized construction and the distinction between finite-sample sharp-null
validity and asymptotic weak-null validity in Section~\ref{sec:weak-null}
continue to apply after replacing $(0,s)$ and $(0,s')$ by $e$ and $e'$.

\subsection{Important special cases}
\label{app:other-estimands-special-cases}

\paragraph{Treated-spillover effects.}
Define the common treated-support set
\begin{equation}
\label{eq:app-treated-support}
\calS_1
=
\{s\in\calS:m_j(s)>0\ \text{for every }j\}.
\end{equation}
For $s,s'\in\calS_1$, Procedure~\ref{proc:app-general-pair-crt} with
$e=(1,s)$ and $e'=(1,s')$ tests
\[
Y_i(1,s)=Y_i(1,s')
\quad\text{for all }i.
\]
The focal units are treated under both contrasted labels.  Unlike the
untreated-support set $\calS_0$, the set $\calS_1$ may include a full-saturation
label.  Thus, when $m_j(1)=n_j$, comparisons such as
$Y_i(1,1)-Y_i(1,s')$ remain well defined and directly testable.

\paragraph{Total effects between two exposure conditions.}
For any $s\neq s'$ such that $(1,s)$ and $(0,s')$ are commonly feasible,
Procedure~\ref{proc:app-general-pair-crt} tests the sharp total-effect null
\[
Y_i(1,s)=Y_i(0,s')
\quad\text{for all }i.
\]
Focal units are selected from the treated units in clusters observed at $s$
and from the untreated units in clusters observed at $s'$.  Under a reference
permutation, their own-treatment statuses switch with the cluster label.  The
second-stage assignments that implement these status changes are integrated
out, and the label relabeling remains uniform by
\eqref{eq:app-general-cancellation}.

A particularly important case is the all-treated versus all-control contrast.
Suppose $0,1\in\calS$ and
\begin{equation}
\label{eq:app-polar-support}
m_j(0)=0,
\qquad
m_j(1)=n_j
\quad\text{for every }j.
\end{equation}
Then the individual contrast is
\[
Y_i(1,1)-Y_i(0,0).
\]
Because every unit in a cluster has the relevant exposure under either polar
label, one may set $k_j=n_j$ and $U_j=\calI_j$.  The test then simply relabels
the $0$ and $1$ saturation labels across the clusters observed in those two
cells and compares their observed cluster means.  It is an exact test of
\begin{equation}
\label{eq:app-polar-sharp-null}
H_{0,\mathrm{global}}:
\qquad
Y_i(1,1)=Y_i(0,0)
\quad\text{for every }i.
\end{equation}

For the two supported cluster exposures, the corresponding unit-average contrast
\begin{equation}
\label{eq:app-gate}
\tau_{\mathrm G}
=
\frac{1}{N}
\sum_{j=1}^K\sum_{i\in\calI_j}
\{Y_i(1,1)-Y_i(0,0)\}
\end{equation}
coincides with the global average treatment effect comparing treatment of everyone with treatment of no one if homogeneous partial interference also holds for those two global assignments. This additional interpretation is not needed to test the supported exposure contrast.  This is the
global estimand emphasized by \citet{viviano2026causal}.  It is important to
distinguish the sharp null \eqref{eq:app-polar-sharp-null} from the weak null
$\tau_{\mathrm G}=0$.  Procedure~\ref{proc:app-general-pair-crt} is
finite-sample exact for the former.  The latter allows individual global effects
to cancel and should instead be handled by the studentized weak-null extension,
with $\lambda_{jK}=n_j/N$ for the unit-average target.


\paragraph{Direct effects at a fixed saturation.}
A direct-effect contrast has the same saturation label on both sides and is
therefore not generated by permuting first-stage saturation assignments.
Define
\begin{equation}
\label{eq:app-direct-support}
\calS_{01}
=
\{s\in\calS:0<m_j(s)<n_j\ \text{for every }j\}.
\end{equation}
For $s\in\calS_{01}$, consider
\begin{equation}
\label{eq:app-direct-sharp-null}
H_{0,\mathrm D}^{s}:
\qquad
Y_i(1,s)=Y_i(0,s)
\quad\text{for every }i.
\end{equation}
Condition on the realized first-stage assignment and let
$J_s^{\mathrm{obs}}=\{j:A_j^{\mathrm{obs}}=s\}$.  The appropriate reference
space is the product of the original second-stage randomizations,
\begin{equation}
\label{eq:app-direct-space}
\calD_{J_s}(s)
=
\prod_{j\in J_s^{\mathrm{obs}}}\calD_j(s).
\end{equation}
For $d\in\calD_{J_s}(s)$, one possible cluster-equally-weighted statistic is
\begin{align}
T_{\mathrm D}^{s}(d)
=
\frac{1}{|J_s^{\mathrm{obs}}|}
\sum_{j\in J_s^{\mathrm{obs}}}
\Bigg\{
&\frac{1}{m_j(s)}
\sum_{i\in\calI_j}Y_i^{\mathrm{obs}}d_i
\notag\\[-2pt]
&-
\frac{1}{n_j-m_j(s)}
\sum_{i\in\calI_j}Y_i^{\mathrm{obs}}(1-d_i)
\Bigg\}.
\label{eq:app-direct-stat}
\end{align}
The randomization $p$-value compares
$T_{\mathrm D}^{s}(D^{\mathrm{obs}})$ with its distribution when, independently
across clusters in $J_s^{\mathrm{obs}}$, exactly $m_j(s)$ treated units are
redrawn uniformly.  Under \eqref{eq:app-direct-sharp-null}, every outcome in
these clusters is invariant to this second-stage reassignment, so the test is
finite-sample exact.  Two-sided tests may use $|T_{\mathrm D}^{s}|$.  Constant
or bounded direct-effect nulls can be handled by the usual sharp imputation or
least-favorable effect-increasing construction.  A pure-control label or a
full-saturation label is excluded from $\calS_{01}$ whenever it sets
$m_j(s)=0$ or $m_j(s)=n_j$, because one of the two own-treatment conditions is
then absent from the design support.

\subsection{Complete absence of spillovers among untreated units}
\label{app:global-no-spillover}

Let
\[
\mathcal Q_0
=
\{q_0,q_1,\ldots,q_M\}
:=
\{q\in\calS_0:K_q>0\},
\qquad
q_0<q_1<\cdots<q_M,
\]
where
\[
K_q=\sum_{j=1}^K\1\{A_j^{\mathrm{obs}}=q\}
\]
is the fixed first-stage count at label $q$.  Labels with $K_q=0$ are omitted
because the experiment contains no randomization support for them.  The joint
individual-level null
of complete absence of spillovers among untreated units is
\begin{equation}
\label{eq:app-global-no-spillover-null}
H_{0,\mathrm{NS}}(\mathcal Q_0):
\qquad
Y_i(0,q)=Y_i(0,q')
\quad
\text{for every }i
\text{ and every }q,q'\in\mathcal Q_0.
\end{equation}
We describe three finite-sample approaches.

\subsubsection{PIRT based on the two monotonicity directions}
\label{app:global-no-spillover-pirt}

The null in \eqref{eq:app-global-no-spillover-null} implies both the increasing
and decreasing monotonicity nulls,
\begin{align*}
H_{0,\mathrm M}^{\uparrow}:
&\quad
Y_i(0,q_0)\leq\cdots\leq Y_i(0,q_M),
\\
H_{0,\mathrm M}^{\downarrow}:
&\quad
Y_i(0,q_0)\geq\cdots\geq Y_i(0,q_M),
\end{align*}
for every unit $i$.  Consequently, the PIRT in
Procedure~\ref{proc:pirt} can be applied directly in either direction.  Let
$p_{\uparrow}^{\mathrm{PIRT}}$ denote its raw PIRT value for the increasing
null, and let $p_{\downarrow}^{\mathrm{PIRT}}$ be the analogous value obtained
by applying the same procedure to $-Y$.  By
Theorem~\ref{thm:pirt-valid},
\[
\widetilde p_{\uparrow}
=
\min\{2p_{\uparrow}^{\mathrm{PIRT}},1\},
\qquad
\widetilde p_{\downarrow}
=
\min\{2p_{\downarrow}^{\mathrm{PIRT}},1\}
\]
are individually valid $p$-values under
$H_{0,\mathrm{NS}}(\mathcal Q_0)$.  An omnibus equality test that is sensitive to
departures in either direction can report
\begin{equation}
\label{eq:app-two-direction-pirt}
p_{\mathrm{NS}}^{\mathrm{PIRT}}
=
\min\left\{
2\min(\widetilde p_{\uparrow},\widetilde p_{\downarrow}),1
\right\}.
\end{equation}
The second factor of two is a Bonferroni adjustment over the two directions.
If scientific theory prespecifies only one direction of departure, the
corresponding single directional PIRT is already a valid test of the equality
null.  It should not, however, be described as an omnibus no-spillover test: an
increasing-monotonicity test has no reason to reject a strictly increasing
spillover response, even though that response violates equality.

\subsubsection{Pairwise CRTs with a multiplicity correction}
\label{app:global-no-spillover-bonferroni}

The joint equality null is equivalent to the intersection of the adjacent
pairwise nulls
\begin{equation}
\label{eq:app-adjacent-equality-nulls}
H_{0,\ell}:
\qquad
Y_i(0,q_\ell)=Y_i(0,q_{\ell+1})
\quad\text{for every }i,
\qquad
\ell=0,\ldots,M-1.
\end{equation}
For each adjacent pair, apply Procedure~\ref{proc:two-level-crt} using a
two-sided focal statistic, such as the absolute difference in focal-cluster
means.  Let $p_\ell^{\mathrm{CRT}}$ be the resulting finite-sample valid
$p$-value.  Bonferroni gives the joint $p$-value
\begin{equation}
\label{eq:app-bonferroni-crt}
p_{\mathrm{NS}}^{\mathrm{Bonf}}
=
\min\left\{
M\min_{0\leq\ell<M}p_\ell^{\mathrm{CRT}},1
\right\}.
\end{equation}
No dependence restriction on the local $p$-values is required.  One may
instead test all $\binom{M+1}{2}$ pairs, but this is unnecessary for validity:
by transitivity, equality for every adjacent pair implies equality for every
pair.  The all-pairs version may be useful for reporting which contrasts drive
a rejection, at the cost of a larger multiplicity adjustment.

\subsubsection{A single multi-saturation conditional randomization test}
\label{app:global-no-spillover-single-crt}

A single omnibus CRT can test
$H_{0,\mathrm{NS}}(\mathcal Q_0)$ directly.  Relative to
Procedure~\ref{proc:two-level-crt}, use all clusters whose observed saturation
belongs to $\mathcal Q_0$.  For each such cluster, choose $k_j$ before
observing the assignment, with
$k_j\leq\min_{q\in\mathcal Q_0}\{n_j-m_j(q)\}$, and sample $k_j$ observed
untreated focal units.  Then jointly permute all labels in $\mathcal Q_0$
across these clusters while preserving each first-stage margin $K_q$, rather
than relabeling only one pair of saturation levels.  Finally, use a
prespecified multi-group statistic.  In particular, the Kruskal--Wallis or
average cross-difference statistic proposed by \citet{hoshino2023} can be
applied with the saturation labels defining the exposure groups.  The
conditional $p$-value is computed as in Procedure~\ref{proc:two-level-crt}, but
over this multi-arm relabeling space; under the joint null, the same
focal-outcome imputability argument applies.

If $\mathcal Q_0=\calS$, this procedure rerandomizes the entire first-stage vector $A$ over $\calA$. In this construction, clusters whose observed labels lie outside $\mathcal Q_0$ retain those labels, while the labels in $\mathcal Q_0$ are jointly permuted.

\section{Formal CRT framework for the two-saturation permutation test}
\label{app:crt-framework}

The main text describes the two-saturation test through focal units and a
cluster-level permutation distribution. This section records the more formal
conditional randomization framework used to prove validity. The notation follows
the conditional randomization-test device of \citet{Basse2019}, but it is kept
in the appendix because the implementation of the method only requires the focal
sets \(U_j\) and the permutation space \(\calA_J^{\mathrm{obs}}\). The purpose of
this section is therefore purely technical: it introduces the conditioning event
and restricted assignment space needed for the proofs below.

\subsection{General conditional-randomization notation}
\label{app:general-crt-notation}

Let \(P\) denote the known randomization distribution on the assignment space
\(\calZ\). A conditional randomization test draws a conditioning event
\[
C=(U,\calR),
\]
where \(U\) is a set of focal units and \(\calR\subseteq\calZ\) is a restricted
assignment set. A conditioning mechanism specifies how \(C\) is generated from
the realized assignment. We write it as
\[
m(C\mid Z)=f(U\mid Z)\,g(\calR\mid U,Z),
\]
where \(f\) selects the focal units and \(g\) selects the restricted assignment
set once the focal units are fixed.

Given \(C\), the conditional assignment distribution is
\begin{equation}\label{eq:app-crt-conditional-law}
\Pr(Z=z\mid C)
=
\frac{m(C\mid z)P(z)}{\sum_{\tilde z\in\calZ}m(C\mid \tilde z)P(\tilde z)},
\qquad z\in\calZ,
\end{equation}
with the convention that assignments with \(m(C\mid z)=0\) receive zero
conditional probability. For a statistic \(T(z\mid Y,C)\), the key requirement
is imputability under the null: for all assignments \(z,z'\) with positive
conditional probability,
\begin{equation}\label{eq:app-crt-imputability}
T\bigl(z'\mid Y(z'),C\bigr)
=
T\bigl(z'\mid Y(z),C\bigr).
\end{equation}
When \eqref{eq:app-crt-imputability} holds, the statistic under the reference
assignments can be computed using the observed outcomes, and the usual
conditional randomization \(p\)-value is finite-sample valid.

\subsection{Specialization to the randomized saturation design}
\label{app:crt-specialization}

Fix \(s\neq s'\). For a generic assignment \(z=(a,d)\in\calZ\), define the
contrasted clusters
\[
J_{s,s'}(z):=\{j:a_j\in\{s,s'\}\}.
\]
For each cluster \(j\), choose \(k_j\) satisfying \eqref{eq:kj-condition}. Given
\(z=(a,d)\), the admissible focal sets are
\[
\mathcal U_{s,s'}(z)
=
\left\{
U\subseteq \bigcup_{j\in J_{s,s'}(z)}\calI_j:
|U\cap\calI_j|=k_j\ \forall j\in J_{s,s'}(z),\
 d_i=0\ \forall i\in U
\right\}.
\]
The focal-sampling component of the conditioning mechanism is uniform on this
set:
\[
f(U\mid z)=
|\mathcal U_{s,s'}(z)|^{-1}\1\{U\in\mathcal U_{s,s'}(z)\}.
\]
Equivalently, it samples \(k_j\) untreated units uniformly without replacement
from every contrasted cluster and does so independently across contrasted
clusters.

Given \(U\) and \(z=(a,d)\), define the restricted assignment set
\begin{equation}\label{eq:app-restricted-assignment-set}
\calZ_{s,s'}(U,z)
=
\left\{
z'=(a',d')\in\calZ:
\begin{array}{l}
 d_i'=0 \ \text{for all } i\in U,\\[2pt]
 a_j'=a_j \ \text{for all } j\notin J_{s,s'}(z),\\[2pt]
 a_j'\in\{s,s'\} \ \text{for all } j\in J_{s,s'}(z)
\end{array}
\right\}.
\end{equation}
Thus every focal unit remains untreated, labels outside the contrasted clusters
are held fixed, and only the two labels \(s\) and \(s'\) may vary across the
contrasted clusters. The restricted-assignment component of the conditioning
mechanism is
\[
g(\calR\mid U,z)=\1\{\calR=\calZ_{s,s'}(U,z)\}.
\]
For the realized assignment, the formal conditioning event is therefore
\begin{equation}\label{eq:app-conditioning-event}
C=
\left(U,\calZ_{s,s'}(U,Z^{\mathrm{obs}})\right).
\end{equation}

Let
\[
J=J^{\mathrm{obs}}=J_{s,s'}(Z^{\mathrm{obs}}),
\qquad
A_J^{\mathrm{obs}}=(A_j^{\mathrm{obs}}:j\in J),
\]
and define \(K_s\) and \(K_{s'}\) as in Section~\ref{subsec:conditioning-event}.
The first-stage design fixes the number of clusters assigned to each saturation
level, so the only first-stage variation left inside \eqref{eq:app-restricted-assignment-set}
is the permutation of \(s\) and \(s'\) on the clusters in \(J\), preserving the
observed margins. Hence the relevant permutation space is exactly
\(\calA_J^{\mathrm{obs}}\) in \eqref{eq:relabel-space}.
Lemma~\ref{lem:conditional-relabeling} below proves that
\[
A_J\mid C\sim \mathrm{Unif}(\calA_J^{\mathrm{obs}}).
\]
This is the formal justification for the simple cluster-level permutation
distribution used in the main text.

The main-text theorem states finite-sample validity conditional on the focal
set \(U\), because \(U\) is the object selected and reported in the testing
algorithm. The proofs below establish the stronger statement conditional on the
formal event \(C\) in \eqref{eq:app-conditioning-event}. Since \(C\) contains
\(U\), the main-text statement follows by iterated expectations. For example, if
\(\Pr(p_T\le \alpha\mid C)\le \alpha\) for every realized \(C\), then
\[
\Pr(p_T\le \alpha\mid U)
=
\E\{\Pr(p_T\le \alpha\mid C)\mid U\}
\le \alpha.
\]
The same argument applies to the bounded-null \(p\)-value.

\subsection{Why aggregation alone does not justify permutation}\label{app:naive-aggregation}
\label{app:naive-cluster-permutation}

A cluster-label permutation applied to the mean of all observed controls
need not be valid under the individual no-spillover null. Consider
$K=20$ clusters, each with two units, and saturation labels $0$ and $1/2$.
The first stage assigns exactly one cluster to $1/2$ and the other nineteen
to $0$, uniformly over the twenty possible label vectors. At saturation
$1/2$, one of the two units is treated with equal probability; at saturation
$0$, neither unit is treated. Write $Y_{j\ell}(r,s)$ for the potential outcome
of unit $\ell\in\{1,2\}$ in cluster $j$, and set, for every $j$,
\[
Y_{j1}(0,0)=Y_{j1}(0,1/2)=-1,
\qquad
Y_{j2}(0,0)=Y_{j2}(0,1/2)=1.
\]
The individual no-spillover null holds exactly; treated potential outcomes
may be arbitrary subject to Assumption~\ref{ass:hpi}.

Let $j^\star$ be the observed $1/2$-saturation cluster and let
$\bar Y_{j,0}^{\mathrm{obs}}$ denote the mean of all observed controls in
cluster $j$. Then $\bar Y_{j,0}^{\mathrm{obs}}=0$ for $j\ne j^\star$, while
$\bar Y_{j^\star,0}^{\mathrm{obs}}=\varepsilon$, where
$\Pr(\varepsilon=1)=\Pr(\varepsilon=-1)=1/2$.
Consider a naive test that holds these twenty means fixed, uniformly
relabels the single $1/2$-saturation label, and uses
\[
T_{\mathrm{naive}}(a)
=
\sum_{j:a_j=1/2}\bar Y_{j,0}^{\mathrm{obs}}
-
\frac{1}{19}\sum_{j:a_j=0}\bar Y_{j,0}^{\mathrm{obs}}.
\]
The observed statistic is $\varepsilon$. Each of the other nineteen
relabelings assigns the $1/2$ label to a zero-mean cluster and gives
$-\varepsilon/19$. The upper-tail permutation probability, including the
observed labeling and counting ties weakly, is therefore
\[
p_{\mathrm{naive}}
=
\begin{cases}
1/20, & \varepsilon=1,\\
1, & \varepsilon=-1.
\end{cases}
\qquad
\Pr\{p_{\mathrm{naive}}\le 0.05\}=\frac12.
\]
Thus aggregation respects the first-stage assignment unit but does not make
the observed control mean invariant to the second-stage sample of controls.
In this example, that mean averages both units at saturation $0$ and only
one randomly selected unit at saturation $1/2$.

By contrast, Procedure~\ref{proc:two-level-crt} uses $k_j=1$ under either
label: it selects one control uniformly at saturation $0$ and retains the
sole control at saturation $1/2$. Under either label, each of the two
singleton focal sets has probability $1/2$. Consequently, conditioning on
the focal sets leaves the first-stage label vector uniform, as in
\eqref{eq:focal-selection-cancellation}, and the null makes the selected
focal outcomes invariant over that conditional support. The failure above
concerns permutation of the fixed all-control means, not cluster-level
permutation of appropriately selected focal outcomes.

\subsection{Stratification, sample restrictions, and observed outcomes}
\label{app:stratified-crt}
Suppose the first-stage distribution is uniform within prespecified geographic blocks, conditional on the number of clusters at each saturation level in each block. For a pair \(s,s'\), let \(J_h\) be the contrasted clusters in block \(h\), and let \(K_{h,s}\) and \(K_{h,s'}\) be their fixed margins. Choose positive focal counts before assignment, feasible under every contrasted label attainable in the block. The calculation in \eqref{eq:focal-selection-cancellation} is unchanged. Conditional on the focal sets, the outside labels, and the block margins,
\[
\Pr(A_J=a_J\mid C)
=\prod_h\binom{|J_h|}{K_{h,s}}^{-1},
\qquad a_J\in\prod_h\calA_{J_h}^{\mathrm{obs}}.
\]
Here \(C\) also records the block margins, and \(\calA_{J_h}^{\mathrm{obs}}\) is the fixed-margin space in \eqref{eq:relabel-space} for block \(h\). A block with only one contrasted label has a singleton randomization space. No unsupported exposure is imputed in such a block: the equality or bounded restriction is imposed only where both untreated exposures are attainable. The proof of Theorem~\ref{thm:two-level-exact} then applies block by block. Likewise, the PIRT ordering argument requires only the ordered exposure pairs that can occur in the conditional design, since a doubly untreated unit has both of its compared exposures in the support.

A cluster may be excluded on the basis of a prespecified roster characteristic, such as having no untreated unit at a tested saturation. Conditioning on its observed saturation label gives the retained design by subtracting that cluster from its block-by-label margin. The excluded EA therefore retains its observed label in the reference distribution.

For outcome \(k\), suppose its observation indicator \(O_{ik}\) is invariant over the reference assignments. Select focal units from the full retained roster and then use \(F_{jk}=\{i\in U_j:O_{ik}=1\}\) in the statistic. Conditional on \(U\), these sets and their counts are fixed. A cluster with \(F_{jk}=\varnothing\) has zero outcome weight but keeps its label in the permutation space; no mean is evaluated by dividing by zero. The imputability and bounded-null comparisons apply to the observed focal outcomes. These are finite-sample statements; they do not by themselves establish a stratified or missing-outcome version of Theorem~\ref{thm:weak-null-asymptotic}.

\section{Proofs for the Two-Saturation CRT}\label{app:two-level-proofs}

Throughout this appendix, fix two distinct saturation levels \(s\neq s'\).
Write
\[
J=J^{\mathrm{obs}}=J_{s,s'}(Z^{\mathrm{obs}}),
\qquad
M=|J|,
\qquad
U_j=U\cap\calI_j,
\]
and recall that \(|U_j|=k_j\). We also write \(K_s\) and \(K_{s'}\) for the
observed numbers of clusters in \(J\) assigned to \(s\) and \(s'\), respectively.
All probabilities in the finite-sample conditional statements are taken with
respect to the randomization induced by the randomized saturation design and by the
focal-sampling rule.

\subsection{Two auxiliary lemmas}\label{app:auxiliary-lemmas}

We first record two elementary facts used repeatedly in the proofs. The first is
a purely probabilistic statement about randomization \(p\)-values on a finite
support. The second verifies that the conditioning event used in
Procedure~\ref{proc:two-level-crt} reduces the randomized saturation design to a uniform
cluster-level relabeling of \(s\) and \(s'\) over the contrasted clusters.

\begin{lemma}[Super-uniformity of randomization \(p\)-values]
\label{lem:randomization-pvalue}
Let \(\mathcal R\) be a finite set and let \(R\) be distributed according to a
probability mass function \(q\) on \(\mathcal R\). For any real-valued function
\(t:\mathcal R\to\mathbb R\), define
\[
p(R)
=
\sum_{r\in\mathcal R}
q(r)\1\{t(r)\ge t(R)\}.
\]
Then, for every \(\alpha\in[0,1]\),
\[
\Pr\{p(R)\le \alpha\}\le \alpha .
\]
\end{lemma}

\begin{proof}
Let \(c_1>c_2>\cdots>c_L\) be the distinct values taken by \(t(R)\). For each
\(\ell=1,\ldots,L\), define the point mass and upper-tail mass
\[
\Delta_\ell=\Pr\{t(R)=c_\ell\},
\qquad
G_\ell=\Pr\{t(R)\ge c_\ell\}=\sum_{m=1}^\ell \Delta_m .
\]
If \(t(R)=c_\ell\), then \(p(R)=G_\ell\). Hence
\[
\Pr\{p(R)\le \alpha\}
=
\sum_{\ell:G_\ell\le \alpha}\Delta_\ell .
\]
If the index set \(\{\ell:G_\ell\le \alpha\}\) is empty, the probability is zero.
Otherwise, let \(\ell_\alpha=\max\{\ell:G_\ell\le \alpha\}\). Since the sequence
\(G_\ell\) is nondecreasing in \(\ell\),
\[
\sum_{\ell:G_\ell\le \alpha}\Delta_\ell
\le
\sum_{\ell=1}^{\ell_\alpha}\Delta_\ell
=
G_{\ell_\alpha}
\le
\alpha .
\]
This proves the desired super-uniformity. The argument allows ties in the
statistic and is therefore the usual conservativeness property of finite
randomization \(p\)-values.
\end{proof}

\begin{lemma}[Conditional distribution of saturation assignments]\label{lem:conditional-relabeling}
Under the randomized saturation design and the focal-sampling rule in
Procedure~\ref{proc:two-level-crt},
\[
A_J\mid C
\sim
\mathrm{Unif}\bigl(\calA_J^{\mathrm{obs}}\bigr),
\qquad
C=(U,\calZ_{s,s'}(U,Z^{\mathrm{obs}})).
\]
\end{lemma}

\begin{proof}
Let
\[
Z^{\mathrm{obs}}=(A^{\mathrm{obs}},D^{\mathrm{obs}}),
\qquad
J=J^{\mathrm{obs}}=J_{s,s'}(Z^{\mathrm{obs}}),
\qquad
U_j=U\cap\calI_j,
\]
and write
\[
\mathcal R^{\mathrm{obs}}
=
\calZ_{s,s'}(U,Z^{\mathrm{obs}}).
\]
Thus the realized conditioning event is
\[
C=(U,\mathcal R^{\mathrm{obs}}).
\]
Recall that the conditioning mechanism first samples focal units and then sets
the restricted assignment set equal to the assignments that keep the focal units
untreated, keep noncontrasted saturation labels fixed, and allow only the labels
\(s\) and \(s'\) to vary on the contrasted clusters. Equivalently, for a generic
assignment \(z=(a,d)\), the conditioning probability can be written as
\[
m(C\mid z)
=
f(U\mid z)\,
\1\{\mathcal R^{\mathrm{obs}}=\calZ_{s,s'}(U,z)\},
\]
where \(f(U\mid z)\) is the focal-sampling probability. In the present design,
\[
f(U\mid z)
=
\prod_{j\in J_{s,s'}(z)}
\binom{n_j-m_j(a_j)}{k_j}^{-1}
\1\{U_j\subseteq\{i\in\calI_j:d_i=0\}\},
\]
provided \(U\) contains exactly \(k_j\) units in each contrasted cluster, and
\(f(U\mid z)=0\) otherwise.

We first identify the assignments that can have positive conditional probability
under \(C\). If \(m(C\mid z)>0\), then
\(\mathcal R^{\mathrm{obs}}=\calZ_{s,s'}(U,z)\). By the definition of
\(\calZ_{s,s'}(U,z)\), this equality forces the set of contrasted clusters to be
\(J\), forces the saturation labels outside \(J\) to equal
\(A_{-J}^{\mathrm{obs}}\), and allows the saturation labels inside \(J\) only to
be \(s\) or \(s'\). Because the first-stage design has fixed saturation margins,
the only possible values of \(A_J\) are therefore the relabelings
\[
\calA_J^{\mathrm{obs}}
=
\left\{
 a_J\in\{s,s'\}^{|J|}:
 \sum_{j\in J}\1\{a_j=s\}=K_s
\right\}.
\]
Thus, for any \(a_J\notin\calA_J^{\mathrm{obs}}\),
\(\Pr(A_J=a_J\mid C)=0\). It remains to show that all
\(a_J\in\calA_J^{\mathrm{obs}}\) have the same conditional probability.

Fix \(a_J\in\calA_J^{\mathrm{obs}}\), and set
\[
a=(a_J,A_{-J}^{\mathrm{obs}}).
\]
By Bayes' rule, over \(a_J\in\calA_J^{\mathrm{obs}}\),
\begin{align*}
\Pr(A_J=a_J\mid C)
&\propto
\Pr(C\mid A_J=a_J,A_{-J}=A_{-J}^{\mathrm{obs}})
\Pr(A_J=a_J,A_{-J}=A_{-J}^{\mathrm{obs}}) \\
&\propto
\Pr(C\mid A=a)
\Pr(A_J=a_J\mid A_{-J}=A_{-J}^{\mathrm{obs}}),
\end{align*}
where factors depending only on \(A_{-J}^{\mathrm{obs}}\) are absorbed into the
constant of proportionality. We now show that \(\Pr(C\mid A=a)\) is constant in
\(a_J\).

Conditional on \(A=a\), the second-stage treatment assignment factorizes across
clusters, so
\[
\Pr(D=d\mid A=a)
=
\prod_{\ell=1}^K
\binom{n_\ell}{m_\ell(a_\ell)}^{-1}
\1\{d_{\calI_\ell}\in\mathcal D_\ell(a_\ell)\}.
\]
Therefore
\begin{align*}
\Pr(C\mid A=a)
&=
\sum_{d:(a,d)\in\calZ}
 m(C\mid (a,d))\Pr(D=d\mid A=a) \\
&=
\sum_{d:(a,d)\in\calZ}
 f(U\mid (a,d))
 \1\{\mathcal R^{\mathrm{obs}}=\calZ_{s,s'}(U,(a,d))\}
 \Pr(D=d\mid A=a).
\end{align*}
For the fixed \(a=(a_J,A_{-J}^{\mathrm{obs}})\) under consideration, every
\(a_J\in\calA_J^{\mathrm{obs}}\) has \(J_{s,s'}(a,d)=J\). Moreover, once \(U\)
and \(a\) are fixed, the set \(\calZ_{s,s'}(U,(a,d))\) does not depend on the
particular second-stage vector \(d\), except that the focal-sampling probability
is zero unless every focal unit is untreated. Hence the indicator
\(\1\{\mathcal R^{\mathrm{obs}}=\calZ_{s,s'}(U,(a,d))\}\) equals one whenever
\(f(U\mid (a,d))>0\), and the previous display reduces to
\[
\Pr(C\mid A=a)
=
\sum_{d:(a,d)\in\calZ}
 f(U\mid (a,d))\Pr(D=d\mid A=a).
\]
Using the product structure of both the second-stage assignment and the focal
sampling, this sum factorizes as
\begin{align*}
\Pr(C\mid A=a)
&=
\prod_{j\in J}
\left[
\sum_{d_{\calI_j}\in\mathcal D_j(a_j)}
 f_j(U_j\mid A_j=a_j,D_{\calI_j}=d_{\calI_j})
 \Pr(D_{\calI_j}=d_{\calI_j}\mid A_j=a_j)
\right] \\
&\quad\times
\prod_{j\notin J}
\left[
\sum_{d_{\calI_j}\in\mathcal D_j(a_j)}
 \Pr(D_{\calI_j}=d_{\calI_j}\mid A_j=a_j)
\right].
\end{align*}
The second product is equal to one. Thus only the contrasted clusters contribute
to \(\Pr(C\mid A=a)\).

Now fix a contrasted cluster \(j\in J\). Conditional on \(A_j=a_j\), exactly
\(m_j(a_j)\) units are treated, and each treatment vector in
\(\mathcal D_j(a_j)\) is equally likely:
\[
\Pr(D_{\calI_j}=d_{\calI_j}\mid A_j=a_j)
=
\binom{n_j}{m_j(a_j)}^{-1},
\qquad
 d_{\calI_j}\in\mathcal D_j(a_j).
\]
Given \(D_{\calI_j}=d_{\calI_j}\), the focal set \(U_j\) is sampled uniformly
from the \(n_j-m_j(a_j)\) untreated units, so
\[
f_j(U_j\mid A_j=a_j,D_{\calI_j}=d_{\calI_j})
=
\binom{n_j-m_j(a_j)}{k_j}^{-1}
\1\{U_j\subseteq\{i\in\calI_j:d_i=0\}\}.
\]
The condition \(1\le k_j\le \min\{n_j-m_j(s),n_j-m_j(s')\}\) guarantees that
this probability is well-defined for both possible labels \(a_j=s\) and
\(a_j=s'\).

The number of vectors \(d_{\calI_j}\in\mathcal D_j(a_j)\) for which every unit
in \(U_j\) is untreated is
\[
\binom{n_j-k_j}{m_j(a_j)},
\]
because all \(m_j(a_j)\) treated units must be chosen from the \(n_j-k_j\)
non-focal units. Therefore the contrasted-cluster contribution is
\begin{align*}
&\sum_{d_{\calI_j}\in\mathcal D_j(a_j)}
 f_j(U_j\mid A_j=a_j,D_{\calI_j}=d_{\calI_j})
 \Pr(D_{\calI_j}=d_{\calI_j}\mid A_j=a_j) \\
&\qquad=
\binom{n_j-k_j}{m_j(a_j)}
\binom{n_j-m_j(a_j)}{k_j}^{-1}
\binom{n_j}{m_j(a_j)}^{-1} \\
&\qquad=
\frac{(n_j-k_j)!}{m_j(a_j)!\{n_j-k_j-m_j(a_j)\}!}
\cdot
\frac{k_j!\{n_j-m_j(a_j)-k_j\}!}{\{n_j-m_j(a_j)\}!}
\cdot
\frac{m_j(a_j)!\{n_j-m_j(a_j)\}!}{n_j!} \\
&\qquad=
\frac{k_j!(n_j-k_j)!}{n_j!}
=
\binom{n_j}{k_j}^{-1}.
\end{align*}
This final expression no longer contains \(m_j(a_j)\), and hence is the same
whether \(a_j=s\) or \(a_j=s'\). Multiplying over \(j\in J\) gives
\[
\Pr(C\mid A=a)
=
\prod_{j\in J}\binom{n_j}{k_j}^{-1},
\]
which is independent of the relabeling \(a_J\in\calA_J^{\mathrm{obs}}\).

We have therefore shown that
\[
\Pr(A_J=a_J\mid C)
\propto
\Pr(A_J=a_J\mid A_{-J}=A_{-J}^{\mathrm{obs}}),
\qquad a_J\in\calA_J^{\mathrm{obs}}.
\]
Under first-stage complete randomization with fixed saturation margins, every
full first-stage assignment satisfying the margins has the same probability.
After conditioning on \(A_{-J}=A_{-J}^{\mathrm{obs}}\), the only remaining
freedom is the placement of the observed \(K_s\) labels equal to \(s\) among the
\(|J|\) clusters in \(J\); the remaining \(K_{s'}\) clusters receive label
\(s'\). Hence
\[
\Pr(A_J=a_J\mid A_{-J}=A_{-J}^{\mathrm{obs}})
=
\binom{|J|}{K_s}^{-1}
=
|\calA_J^{\mathrm{obs}}|^{-1},
\qquad a_J\in\calA_J^{\mathrm{obs}}.
\]
Combining this with the constancy of \(\Pr(C\mid A=a)\) and normalizing over
\(\calA_J^{\mathrm{obs}}\) yields
\[
\Pr(A_J=a_J\mid C)
=
|\calA_J^{\mathrm{obs}}|^{-1},
\qquad a_J\in\calA_J^{\mathrm{obs}}.
\]
This is exactly
\(A_J\mid C\sim\mathrm{Unif}(\calA_J^{\mathrm{obs}})\).
\end{proof}

\subsection{Proof of Theorem~\ref{thm:two-level-exact}(a)}
\label{app:proof-ps}

\begin{proof}
Take any assignment \(z'=(a',d')\in\calZ_{s,s'}(U,Z^{\mathrm{obs}})\). By the
construction of the restricted assignment set,
\[
d_i'=0
\quad\text{and}\quad
a'_{[i]}\in\{s,s'\}
\qquad
\text{for every } i\in U.
\]
Thus every focal unit has exposure either \((0,s)\) or \((0,s')\) under every
assignment in the conditioning set. Under the partially sharp spillover null
\(H_{0,\mathrm{PS}}^{s,s'}\),
\[
Y_i(0,s)=Y_i(0,s')
\qquad
\text{for every } i.
\]
By Assumption~\ref{ass:hpi}, each focal outcome under any assignment in the
conditioning set is therefore equal to its observed value:
\[
Y_i(z')=Y_i(Z^{\mathrm{obs}})
\qquad
\text{for all } i\in U
\text{ and all } z'\in\calZ_{s,s'}(U,Z^{\mathrm{obs}}).
\]
Hence any statistic that depends only on focal outcomes and the saturation label vector
is imputable over the entire conditional assignment set.

By Lemma~\ref{lem:conditional-relabeling},
\[
A_J\mid C\sim \mathrm{Unif}(\calA_J^{\mathrm{obs}}).
\]
The conditional \(p\)-value in \eqref{eq:generic-crt-pvalue} is therefore an
ordinary finite randomization \(p\)-value on the permutation space
\(\calA_J^{\mathrm{obs}}\). Applying Lemma~\ref{lem:randomization-pvalue} with
\(\mathcal R=\calA_J^{\mathrm{obs}}\), \(R=A_J\), and
\(t=T_U\) gives
\[
\Pr\{p_T\le \alpha\mid C\}\le \alpha,
\qquad \alpha\in[0,1].
\]
\end{proof}

\subsection{Proof of Theorem~\ref{thm:two-level-exact}(b)}
\label{app:proof-bounded}

\begin{proof}
Condition on $C=(U,\calZ_{s,s'}(U,Z^{\mathrm{obs}}))$.
By Lemma~\ref{lem:conditional-relabeling}, $A_J$ is uniform over
$\calA_J^{\mathrm{obs}}$. For $q\in\{s,s'\}$, let
$\widetilde Y_j(q)=k_j^{-1}\sum_{i\in U_j}Y_i(0,q)$, and write
$\widetilde Y(a_J)=(\widetilde Y_j(a_j):j\in J)$.
These potential-outcome arrays are fixed conditional on $C$.
The unitwise bounded null implies
$\widetilde Y_j(s)-\widetilde Y_j(s')\le\delta$ for every $j\in J$.
For any candidate $a_J$, the imputation in
\eqref{eq:boundary-imputation} consequently satisfies
\[
\begin{cases}
\widetilde Y_{j,\delta}^{\mathrm{imp}}(a_j)\ge\widetilde Y_j(a_j),
    &a_j=s,\\
\widetilde Y_{j,\delta}^{\mathrm{imp}}(a_j)\le\widetilde Y_j(a_j),
    &a_j=s',
\end{cases}
\]
with equality whenever $a_j=A_j^{\mathrm{obs}}$.
Indeed, a switch from $s'$ to $s$ imputes
$\widetilde Y_j(s')+\delta\ge\widetilde Y_j(s)$, and the opposite switch
imputes $\widetilde Y_j(s)-\delta\le\widetilde Y_j(s')$.
Effect increasingness gives
\[
t\bigl(a_J,\widetilde Y_{\delta}^{\mathrm{imp}}(a_J)\bigr)
\ge t\bigl(a_J,\widetilde Y(a_J)\bigr),
\qquad a_J\in\calA_J^{\mathrm{obs}},
\]
and equality holds at $A_J^{\mathrm{obs}}$.

Consider the infeasible $p$-value based on the true schedule,
\[
p^{\mathrm{true}}
=
\frac{1}{|\calA_J^{\mathrm{obs}}|}
\sum_{a_J\in\calA_J^{\mathrm{obs}}}
\1\!\left[
 t\bigl(a_J,\widetilde Y(a_J)\bigr)
 \ge t\bigl(A_J^{\mathrm{obs}},\widetilde Y(A_J^{\mathrm{obs}})\bigr)
\right].
\]
Conditional on $C$, the map $a_J\mapsto t\bigl(a_J,\widetilde Y(a_J)\bigr)$ is a
fixed function of a uniformly randomized label vector. Lemma~\ref{lem:randomization-pvalue}
therefore gives $\Pr(p^{\mathrm{true}}\le\alpha\mid C)\le\alpha$.
The preceding pointwise comparison and equality of observed thresholds imply
$p_T\ge p^{\mathrm{true}}$, so
\[
\Pr(p_T\le\alpha\mid C)
\le\Pr(p^{\mathrm{true}}\le\alpha\mid C)\le\alpha.
\]
Iterated expectations yield the claimed statement conditional on $U$.
For independent Monte Carlo relabelings, the pointwise comparison also
holds for the plus-one $p$-values. The true-schedule plus-one $p$-value is
super-uniform by exchangeability of the observed and reference label vectors,
so the same conclusion holds jointly over assignment and Monte Carlo draws.
\end{proof}

\subsection{Heterogeneous individual bounds}
\label{app:heterogeneous-bounds}

\begin{proposition}[Conditional tests with heterogeneous bounds]
\label{prop:heterogeneous-bounds}
Under the design and focal-sampling conditions of
Theorem~\ref{thm:two-level-exact}, fix constants $\delta_i$
before observing assignments or outcomes and suppose
\[
H_{\boldsymbol\delta,\mathrm B}^{s,s'}:
\quad
Y_i(0,s)-Y_i(0,s')\le\delta_i
\quad\text{for every }i.
\]
Set $d_j=\bar\delta_{U,j}
=k_j^{-1}\sum_{i\in U_j}\delta_i$ and define
\begin{equation}\label{eq:heterogeneous-boundary-imputation}
\widetilde Y_{j,\boldsymbol\delta}^{\mathrm{imp}}(a_j)
=
\widetilde Y_j^{\mathrm{obs}}
+d_j\bigl[
\1\{a_j=s\}-\1\{A_j^{\mathrm{obs}}=s\}
\bigr].
\end{equation}
For a prespecified effect-increasing statistic $t$, use
\[
T_U(a_J)
=
t\bigl(a_J,\widetilde Y_{\boldsymbol\delta}^{\mathrm{imp}}(a_J)\bigr).
\]
The conditional randomization $p$-value satisfies
$\Pr(p_T\le\alpha\mid C)\le\alpha$, and hence
$\Pr(p_T\le\alpha\mid U)\le\alpha$,
for every $\alpha\in[0,1]$.
\end{proposition}

\begin{proof}
Conditional on $C$, the $d_j$ are fixed and
$\widetilde Y_j(s)-\widetilde Y_j(s')\le d_j$.
Replacing $\delta$ by $d_j$ in the proof of
Theorem~\ref{thm:two-level-exact}(b) gives the same
pointwise comparison between imputed and true reference
statistics, with equality at the observed labeling.
The conditional validity statements and their plus-one
Monte Carlo counterparts follow by the same argument.
\end{proof}

\section{Regularity conditions and Proof of Theorem~\ref{thm:weak-null-asymptotic}}
\label{app:proof-weak}

\subsection{Regularity conditions}\label{app:weak-regularity}

For the weak-null result in Section \ref{sec:weak-null}, let \(J_K=\{j:A_j\in\{s,s'\}\}\) and \(M_K=|J_K|=K_s+K_{s'}\).
Independently of the first-stage assignment, draw auxiliary focal samples \(U_j\), independently across all \(K\) clusters and uniformly among the \(k_j\)-subsets of \(\calI_j\). For each such sample, define
\[
\widetilde Y_j(a)
=
\frac{1}{k_j}\sum_{i\in U_j}Y_i(0,a),
\qquad a\in\{s,s'\}.
\]
For a prespecified weight vector \(\lambda\), let
\[
w_{jK}:=K\lambda_{jK},
\qquad
X_j(a):=w_{jK}\widetilde Y_j(a),
\qquad
X_j^\tau:=X_j(s)-X_j(s').
\]
Also define
\[
\mu_{X,a,j}:=w_{jK}\mu_{a,j},
\qquad
\mu_{X,\tau,j}:=\mu_{X,s,j}-\mu_{X,s',j}.
\]
The randomness in \(X_j(a)\) and \(X_j^\tau\) comes from the focal sampling.

\begin{assumption}[Weighted weak-null regularity]\label{ass:weak-regularity}
Fix \(s\neq s'\) and a sequence of nonnegative, nonrandom weights \(\lambda_{1K},\ldots,\lambda_{KK}\) satisfying \(\sum_{j=1}^K\lambda_{jK}=1\).
Let \(w_{jK}=K\lambda_{jK}\).
Along a sequence of experiments indexed by \(K\), the following conditions hold.
\begin{enumerate}[(i)]
\item \textit{Stable design fractions.}
\[
\frac{K_s}{K}\to\rho_s\in(0,1),
\qquad
\frac{K_{s'}}{K}\to\rho_{s'}\in(0,1).
\]
Let \(\kappa=\rho_s+\rho_{s'}\), \(\pi_s=\rho_s/\kappa\), and \(\pi_{s'}=\rho_{s'}/\kappa\).

\item \textit{Transformed fourth-moment bounds.}
For \(a\in\{s,s'\}\),
\[
\sup_K
\frac{1}{K}\sum_{j=1}^K
w_{jK}^4
\frac{1}{n_j}\sum_{i\in\calI_j}Y_i(0,a)^4
<\infty,
\]
and the same bound holds with \(Y_i(0,a)\) replaced by \(Y_i(0,s)-Y_i(0,s')\).

\item \textit{Nondegenerate limiting variances.}
The two finite-population outcome variances converge in probability to positive finite limits.
The finite-population variance of \(\{X_j^\tau\}_{j=1}^K\) converges in probability to a finite limit.

\item \textit{Maximal negligibility.}
For \(a\in\{s,s'\}\),
\[
\max_{1\le j\le K}
\left|
X_j(a)
-
K^{-1}\sum_{\ell=1}^K X_\ell(a)
\right|
=o_p(K^{1/2}),
\]
and the analogous condition holds for \(X_j^\tau\).
\end{enumerate}
\end{assumption}

The fourth-moment bounds in part (ii) also imply the effect bound and part (iv). Indeed, \(|x-y|^4\le8(|x|^4+|y|^4)\), and a union bound gives
\(\Pr\{\max_j|X_j(a)|>\epsilon\sqrt K\}=O(K^{-1})\); centering preserves maximal negligibility. We retain the explicit conditions to identify the requirements used in the proof.

\begin{remark}[Role of the weights]\label{rem:weak-weights}
Assumption~\ref{ass:weak-regularity} is imposed on the transformed cluster-level array because the weak-null target determines the scale of the statistic.
For the equally weighted cluster-average target, \(w_{jK}=1\).
For the unit-average target, \(w_{jK}=n_j/(N/K)\), so the assumption permits unit-average weak-null inference with unequal cluster sizes but requires the corresponding weighted moments and maximal terms to remain controlled.
\end{remark}

\subsection{Proof of Theorem~\ref{thm:weak-null-asymptotic}}

Fix the weight sequence \(\lambda_{1K},\ldots,\lambda_{KK}\), and write
\[
w_{jK}:=K\lambda_{jK}.
\]
For \(a\in\{s,s'\}\), define
\[
\mu_{a,j}
=
\frac{1}{n_j}\sum_{i\in\calI_j}Y_i(0,a),
\qquad
\tau_{ij}=Y_i(0,s)-Y_i(0,s'),
\qquad
\mu_{\tau,j}=\mu_{s,j}-\mu_{s',j}.
\]
For the auxiliary focal sample \(U_j\), write
\[
\widetilde Y_j(a)=\frac{1}{k_j}\sum_{i\in U_j}Y_i(0,a),
\qquad
\widetilde\tau_j=\widetilde Y_j(s)-\widetilde Y_j(s').
\]
Define the transformed focal-cluster potential outcomes
\[
X_j(a):=w_{jK}\widetilde Y_j(a),
\qquad
X_j^\tau:=X_j(s)-X_j(s')=w_{jK}\widetilde\tau_j,
\]
and the transformed cluster-level means
\[
\mu_{X,a,j}:=w_{jK}\mu_{a,j},
\qquad
\mu_{X,\tau,j}:=\mu_{X,s,j}-\mu_{X,s',j}=w_{jK}\mu_{\tau,j}.
\]
Because \(w_{jK}=K\lambda_{jK}\),
\[
\frac{1}{K}\sum_{j=1}^K\mu_{X,\tau,j}
=
\sum_{j=1}^K\lambda_{jK}\tau_j^{s,s'}.
\]
Thus the weighted weak null restricts the expected transformed focal effects:
\[
K^{-1}\sum_j\E_U X_j^\tau=K^{-1}\sum_j\mu_{X,\tau,j}=0.
\]
The realized focal average \(K^{-1}\sum_jX_j^\tau\) need not be zero; the decomposition below accounts for focal-sampling randomness.

Let
\[
v_{X,a,j}:=\Var\{X_j(a)\},
\qquad
v_{X,\tau,j}:=\Var(X_j^\tau),
\]
where the variance is over focal sampling within cluster \(j\).
Equivalently,
\[
v_{X,a,j}=w_{jK}^2v_{a,j},
\qquad
v_{X,\tau,j}=w_{jK}^2v_{\tau,j},
\]
where \(v_{a,j}=\Var\{\widetilde Y_j(a)\}\) and \(v_{\tau,j}=\Var(\widetilde\tau_j)\).

Let
\[
J_K=J_{s,s'}(Z),
\qquad
M_K=|J_K|=K_s+K_{s'}.
\]
The proof uses the following sequential representation of the contrasted part of the design.

\begin{lemma}[Equivalent sequential representation]
\label{lem:sequential-representation}
For fixed \(s\neq s'\), the joint law of the contrasted clusters, focal sets, and saturation labels can be represented as follows:
\begin{enumerate}[(i)]
\item draw \(J_K\) uniformly from all subsets of \([K]\) of size \(M_K\);
\item conditional on \(J_K\), draw \(U_j\subseteq\calI_j\) independently across \(j\in J_K\), uniformly among all subsets of size \(k_j\);
\item conditional on \((J_K,U)\), assign exactly \(K_s\) of the clusters in \(J_K\) to saturation \(s\), and the remaining \(K_{s'}\) clusters to saturation \(s'\), uniformly over all such relabelings.
\end{enumerate}
\end{lemma}

\begin{proof}
The first statement follows from first-stage complete randomization with fixed margins: the set of clusters assigned to either \(s\) or \(s'\) is a simple random sample without replacement of size \(M_K=K_s+K_{s'}\) from \([K]\).
Given \(J_K\), the first-stage labels inside \(J_K\) are uniform over all allocations of \(K_s\) labels \(s\) and \(K_{s'}\) labels \(s'\).

It remains to show that the focal sets can be drawn before the relabeling and that their distribution does not depend on the label \(s\) or \(s'\).
Fix a cluster \(j\), a label \(a\in\{s,s'\}\), and a subset \(u\subseteq\calI_j\) with \(|u|=k_j\).
Conditional on \(A_j=a\),
\begin{align*}
\Pr(U_j=u\mid A_j=a)
&=
\sum_{d_{\calI_j}\in\mathcal D_j(a)}
\Pr(U_j=u\mid A_j=a,D_{\calI_j}=d_{\calI_j})
\Pr(D_{\calI_j}=d_{\calI_j}\mid A_j=a) \\
&=
\binom{n_j-k_j}{m_j(a)}
\binom{n_j-m_j(a)}{k_j}^{-1}
\binom{n_j}{m_j(a)}^{-1} \\
&=
\binom{n_j}{k_j}^{-1}.
\end{align*}
The last expression is independent of \(a\in\{s,s'\}\).
The independence of the focal sets across clusters follows from the product form of the second-stage randomization and from the independent focal sampling across clusters.
Hence the same joint law is obtained by drawing \(J_K\), then drawing the focal sets, and then relabeling the clusters in \(J_K\) with fixed margins.
\end{proof}

The next proposition records the moment convergence facts used in the weak-null proof.
It is stated for the transformed array because the weak-null target determines the scale of the cluster-level statistic.

\begin{proposition}[Convergence of transformed focal finite-population moments]
\label{prop:focal-moment-convergence}
For \(a\in\{s,s'\}\), define
\[
\bar X_{J_K}(a)
=
\frac{1}{M_K}\sum_{j\in J_K}X_j(a),
\qquad
S_{X,a,F}^2(J_K,U)
=
\frac{1}{M_K-1}
\sum_{j\in J_K}
\{X_j(a)-\bar X_{J_K}(a)\}^2,
\]
and define analogously
\[
\bar X_{J_K}^\tau
=
\frac{1}{M_K}\sum_{j\in J_K}X_j^\tau,
\qquad
S_{X,\tau,F}^2(J_K,U)
=
\frac{1}{M_K-1}
\sum_{j\in J_K}
(X_j^\tau-\bar X_{J_K}^\tau)^2.
\]
Let
\[
\bar\mu_{X,a,J_K}
=
\frac{1}{M_K}\sum_{j\in J_K}\mu_{X,a,j},
\qquad
\bar\mu_{X,\tau,J_K}
=
\frac{1}{M_K}\sum_{j\in J_K}\mu_{X,\tau,j},
\]
and define the deterministic focal-sampling targets
\[
S_{X,a,K}^2(J_K)
=
\frac{1}{M_K-1}
\sum_{j\in J_K}(\mu_{X,a,j}-\bar\mu_{X,a,J_K})^2
+
\frac{1}{M_K}\sum_{j\in J_K}v_{X,a,j},
\]
\[
S_{X,\tau,K}^2(J_K)
=
\frac{1}{M_K-1}
\sum_{j\in J_K}(\mu_{X,\tau,j}-\bar\mu_{X,\tau,J_K})^2
+
\frac{1}{M_K}\sum_{j\in J_K}v_{X,\tau,j}.
\]
Under Assumption~\ref{ass:weak-regularity}, along any subsequence on which the finite-population limits below exist,
\[
S_{X,a,F}^2(J_K,U)-S_{X,a,K}^2(J_K)=o_p(1),
\qquad a\in\{s,s'\},
\]
and
\[
S_{X,\tau,F}^2(J_K,U)-S_{X,\tau,K}^2(J_K)=o_p(1).
\]
Moreover, if the finite-population variances of \(\{X_j(a)\}_{j=1}^K\) and \(\{X_j^\tau\}_{j=1}^K\) converge in probability to \(S_{X,a}^2\) and \(S_{X,\tau}^2\), respectively, then
\[
S_{X,a,F}^2(J_K,U)\xrightarrow{p}S_{X,a}^2,
\qquad
S_{X,\tau,F}^2(J_K,U)\xrightarrow{p}S_{X,\tau}^2.
\]
Consequently, conditional on \((J_K,U)\), the arm-specific sample variances in \eqref{eq:stat-studentized} satisfy
\[
\hat S_{X,U}^2(a;A_J)-S_{X,a,F}^2(J_K,U)=o_p(1),
\qquad a\in\{s,s'\},
\]
and
\[
M_K\hat V_{X,U}^{\mathrm{Ney}}(A_J)
\xrightarrow{p}
D_X
:=
\frac{S_{X,s}^2}{\pi_s}+
\frac{S_{X,s'}^2}{\pi_{s'}}.
\]
\end{proposition}

\begin{proof}
We give the argument for a generic sequence \(Q_j\), where \(Q_j\) denotes either \(X_j(a)\) for \(a\in\{s,s'\}\) or \(X_j^\tau\).
Conditional on \(J_K\), the focal samples are independent across \(j\in J_K\).
Write
\[
\mu^Q_j=\E(Q_j\mid J_K),
\qquad
r^Q_j=\Var(Q_j\mid J_K),
\qquad
\xi_j=Q_j-\mu^Q_j.
\]
For \(Q_j=X_j(a)\), \((\mu^Q_j,r^Q_j)=(\mu_{X,a,j},v_{X,a,j})\).
For \(Q_j=X_j^\tau\), \((\mu^Q_j,r^Q_j)=(\mu_{X,\tau,j},v_{X,\tau,j})\).
Also define
\[
\bar\mu^Q_J=\frac{1}{M_K}\sum_{j\in J_K}\mu^Q_j.
\]

First, conditional on \(J_K\),
\[
\Var\left(\frac{1}{M_K}\sum_{j\in J_K}\xi_j\,\middle|\,J_K\right)
=
\frac{1}{M_K^2}\sum_{j\in J_K}r^Q_j=O_p(M_K^{-1}),
\]
where the order follows from the transformed fourth-moment bound in Assumption~\ref{ass:weak-regularity}(ii).
Chebyshev's inequality gives
\begin{equation}\label{eq:generic-mean-lln-new}
\frac{1}{M_K}\sum_{j\in J_K}Q_j-\bar\mu^Q_J=o_p(1).
\end{equation}

Next, let \(Z_{ij}\) denote the transformed unit-level quantity being averaged: either \(w_{jK}Y_i(0,a)\) or \(w_{jK}\tau_{ij}\).
Since \(U_j\) is a simple random sample without replacement of size \(k_j\), convexity gives
\[
Q_j^4
=
\left(\frac{1}{k_j}\sum_{i\in U_j}Z_{ij}\right)^4
\le
\frac{1}{k_j}\sum_{i\in U_j}Z_{ij}^4.
\]
Taking expectations over the focal sample,
\[
\E(Q_j^4\mid J_K)
\le
\frac{1}{k_j}\sum_{i\in\calI_j}Z_{ij}^4\Pr(i\in U_j\mid J_K)
\le
\frac{1}{n_j}\sum_{i\in\calI_j}Z_{ij}^4.
\]
By Assumption~\ref{ass:weak-regularity}(ii),
\[
\frac{1}{M_K}\sum_{j\in J_K}\E(Q_j^4\mid J_K)=O_p(1).
\]
Since \(|x-y|^4\le 8(x^4+y^4)\) and \(|\mu^Q_j|^4\le \E(Q_j^4\mid J_K)\) by Jensen's inequality,
\[
\E(\xi_j^4\mid J_K)\le 16\,\E(Q_j^4\mid J_K),
\]
and hence
\[
\frac{1}{M_K}\sum_{j\in J_K}\E(\xi_j^4\mid J_K)=O_p(1).
\]
Let
\[
V_K=
\frac{1}{M_K}\sum_{j\in J_K}\{\xi_j^2-r^Q_j\}.
\]
Conditional on \(J_K\), the variables \(\{\xi_j^2-r^Q_j:j\in J_K\}\) are independent and mean zero, so
\[
\Var(V_K\mid J_K)
\le
\frac{1}{M_K^2}\sum_{j\in J_K}\E(\xi_j^4\mid J_K)
=
O_p(M_K^{-1}).
\]
Therefore, for every \(\varepsilon>0\), conditional Chebyshev gives
\[
\Pr(|V_K|>\varepsilon\mid J_K)
\le
\varepsilon^{-2}\Var(V_K\mid J_K)
=
o_p(1).
\]
Thus \(V_K=o_p(1)\), or equivalently,
\begin{equation}\label{eq:generic-second-lln-new}
\frac{1}{M_K}\sum_{j\in J_K}\xi_j^2
-
\frac{1}{M_K}\sum_{j\in J_K}r^Q_j
=
o_p(1).
\end{equation}

We also need the cross term.
Conditional on \(J_K\),
\[
\Var\left(\frac{1}{M_K}\sum_{j\in J_K}\mu^Q_j\xi_j\,\middle|\,J_K\right)
=
\frac{1}{M_K^2}\sum_{j\in J_K}(\mu^Q_j)^2r^Q_j.
\]
By Cauchy--Schwarz,
\[
\frac{1}{M_K}\sum_{j\in J_K}(\mu^Q_j)^2r^Q_j
\le
\left(\frac{1}{M_K}\sum_{j\in J_K}(\mu^Q_j)^4\right)^{1/2}
\left(\frac{1}{M_K}\sum_{j\in J_K}(r^Q_j)^2\right)^{1/2}
=O_p(1),
\]
again by the transformed fourth-moment bound and Jensen's inequality.
Hence
\begin{equation}\label{eq:generic-cross-lln-new}
\frac{1}{M_K}\sum_{j\in J_K}\mu^Q_j\xi_j=o_p(1).
\end{equation}
Combining \eqref{eq:generic-mean-lln-new}, \eqref{eq:generic-second-lln-new}, and \eqref{eq:generic-cross-lln-new},
\begin{align*}
&\frac{M_K}{M_K-1}
\left[
\frac{1}{M_K}\sum_{j\in J_K}Q_j^2
-
\left\{\frac{1}{M_K}\sum_{j\in J_K}Q_j\right\}^2
\right] \\
&\qquad =
\frac{1}{M_K-1}\sum_{j\in J_K}(\mu^Q_j-\bar\mu^Q_J)^2
+
\frac{1}{M_K}\sum_{j\in J_K}r^Q_j
+o_p(1).
\end{align*}
Applying this identity with \(Q_j=X_j(a)\) yields
\[
S_{X,a,F}^2(J_K,U)-S_{X,a,K}^2(J_K)=o_p(1),
\qquad a\in\{s,s'\},
\]
and applying it with \(Q_j=X_j^\tau\) yields
\[
S_{X,\tau,F}^2(J_K,U)-S_{X,\tau,K}^2(J_K)=o_p(1).
\]

The convergence of the deterministic targets over \(J_K\) follows from the fact that \(J_K\) is a simple random sample without replacement from \([K]\) with \(M_K/K\to\kappa\in(0,1]\).
Standard finite-population laws of large numbers, using the same transformed fourth-moment bounds, give convergence of the sampled first and second moments and of the sampled averages of \(v_{X,a,j}\) and \(v_{X,\tau,j}\) to their full finite-population counterparts.
Thus the variances over \(J_K\) have the same limits as the full finite-population variances of the transformed focal quantities.
This proves
\[
S_{X,a,F}^2(J_K,U)\xrightarrow{p}S_{X,a}^2,
\qquad
S_{X,\tau,F}^2(J_K,U)\xrightarrow{p}S_{X,\tau}^2.
\]

Finally, conditional on \((J_K,U)\), the observed labels inside \(J_K\) form a complete randomization with fixed arm sizes \(K_s\) and \(K_{s'}\).
Since \(K_s/M_K\to\pi_s\in(0,1)\) and \(K_{s'}/M_K\to\pi_{s'}\in(0,1)\), the usual finite-population law of large numbers for complete randomization gives
\[
\hat S_{X,U}^2(a;A_J)-S_{X,a,F}^2(J_K,U)=o_p(1),
\qquad a\in\{s,s'\}.
\]
Therefore
\[
M_K\hat V_{X,U}^{\mathrm{Ney}}(A_J)
=
\frac{M_K}{K_s}\hat S_{X,U}^2(s;A_J)
+
\frac{M_K}{K_{s'}}\hat S_{X,U}^2(s';A_J)
\xrightarrow{p}
\frac{S_{X,s}^2}{\pi_s}+
\frac{S_{X,s'}^2}{\pi_{s'}}.
\]
\end{proof}

\begin{proof}[Proof of Theorem~\ref{thm:weak-null-asymptotic}]
Write \(J=J_K\) and \(M=M_K\) to simplify notation.
We prove the result along an arbitrary subsequence on which all bounded deterministic finite-population quantities used below have limits.
This is sufficient for the stated limsup claim because every sequence has a further subsequence of this type under the moment bounds in Assumption~\ref{ass:weak-regularity}.

Define
\[
\bar\mu_{X,\tau,J}:=\frac{1}{M}\sum_{j\in J}\mu_{X,\tau,j},
\qquad
\bar X_J^\tau:=\frac{1}{M}\sum_{j\in J}X_j^\tau.
\]
At the observed relabeling, the unstudentized transformed focal-cluster statistic is
\[
\hat\tau_{X,U}(A_J)
=
\frac{1}{K_s}\sum_{j\in J}X_j(s)\1\{A_j=s\}
-
\frac{1}{K_{s'}}\sum_{j\in J}X_j(s')\1\{A_j=s'\}.
\]
Under the weighted weak null, \(\bar\mu_{X,\tau,K}:=K^{-1}\sum_{j=1}^K\mu_{X,\tau,j}=0\).
Add and subtract \(\bar\mu_{X,\tau,J}\) and \(\bar X_J^\tau\) to obtain the exact three-term decomposition
\begin{equation}\label{eq:weak-three-term-decomp}
\hat\tau_{X,U}(A_J)
=
\underbrace{(\bar\mu_{X,\tau,J}-\bar\mu_{X,\tau,K})}_{=:R_K^{(1)}}
+
\underbrace{(\bar X_J^\tau-\bar\mu_{X,\tau,J})}_{=:R_K^{(2)}}
+
\underbrace{\{\hat\tau_{X,U}(A_J)-\bar X_J^\tau\}}_{=:R_K^{(3)}}.
\end{equation}
The first term is the error from sampling the contrasted clusters; the second is the error from focal sampling within those clusters; the third is the complete-randomization contrast generated by relabeling \(s\) and \(s'\) within \(J\).

Let
\[
V_{X,\mu,K}
:=
\frac{1}{K-1}\sum_{j=1}^K(\mu_{X,\tau,j}-\bar\mu_{X,\tau,K})^2,
\qquad
W_{X,\tau,K}^2
:=
\frac{1}{K}\sum_{j=1}^K v_{X,\tau,j},
\]
and write \(V_{X,\mu}\) and \(W_{X,\tau}^2\) for their subsequential limits.
By Proposition~\ref{prop:focal-moment-convergence}, the subsequential limit of the focal finite-population variance of \(\{X_j^\tau\}_{j=1}^K\) satisfies
\[
S_{X,\tau}^2=V_{X,\mu}+W_{X,\tau}^2.
\]
Similarly, let \(S_{X,s}^2\) and \(S_{X,s'}^2\) be the subsequential limits of the focal finite-population variances of \(\{X_j(s)\}_{j=1}^K\) and \(\{X_j(s')\}_{j=1}^K\).
Define
\[
D_X
:=
\frac{S_{X,s}^2}{\pi_s}
+
\frac{S_{X,s'}^2}{\pi_{s'}}.
\]
The nondegeneracy condition in Assumption~\ref{ass:weak-regularity}(iii) implies \(0<D_X<\infty\), and Proposition~\ref{prop:focal-moment-convergence} gives
\[
M\hat V_{X,U}^{\mathrm{Ney}}(A_J)\xrightarrow{p}D_X.
\]
Thus it is enough to analyze the scaled components
\[
C_K:=\frac{\sqrt M R_K^{(1)}}{\sqrt{D_X}},
\qquad
A_K:=\frac{\sqrt M R_K^{(2)}}{\sqrt{D_X}},
\qquad
B_K:=\frac{\sqrt M R_K^{(3)}}{\sqrt{D_X}}.
\]

\emph{Step 1: the contrasted-set sampling term.}
By Lemma~\ref{lem:sequential-representation}, \(J\) is a simple random sample without replacement of size \(M\) from \([K]\).
Therefore
\[
\Var\{\sqrt M(\bar\mu_{X,\tau,J}-\bar\mu_{X,\tau,K})\}
=
\left(1-\frac{M}{K}\right)
\frac{1}{K-1}\sum_{j=1}^K(\mu_{X,\tau,j}-\bar\mu_{X,\tau,K})^2
=
\left(1-\frac{M}{K}\right)V_{X,\mu,K}.
\]
The transformed fourth-moment bound in Assumption~\ref{ass:weak-regularity}(ii) implies the usual H\'ajek maximal-negligibility condition for the deterministic array \(\{\mu_{X,\tau,j}\}_{j=1}^K\).
Indeed, Jensen's inequality gives \(K^{-1}\sum_j\mu_{X,\tau,j}^4=O(1)\), and hence
\[
\max_j|\mu_{X,\tau,j}-\bar\mu_{X,\tau,K}|=o(K^{1/2}).
\]

If \((1-\kappa)V_{X,\mu}>0\), the finite-population central limit theorem for sampling without replacement yields
\[
C_K
\Rightarrow
N\!\left(0,\frac{(1-\kappa)V_{X,\mu}}{D_X}\right),
\]
or equivalently
\[
\sqrt M R_K^{(1)}
\Rightarrow
N\{0,(1-\kappa)V_{X,\mu}\}.
\]
If \((1-\kappa)V_{X,\mu}=0\), then
\[
\Var\{\sqrt M R_K^{(1)}\}
=
\left(1-\frac{M}{K}\right)V_{X,\mu,K}
\longrightarrow 0,
\]
and Chebyshev's inequality gives
\[
\sqrt M R_K^{(1)}\xrightarrow{p}0.
\]
Thus, in all cases,
\[
\sqrt M R_K^{(1)}
\Rightarrow
N\{0,(1-\kappa)V_{X,\mu}\},
\]
where \(N(0,0)\) denotes the degenerate distribution at zero.

\emph{Step 2: the focal-sampling term.}
Conditional on \(J\), the random variables
\[
\eta_j:=X_j^\tau-\mu_{X,\tau,j},
\qquad j\in J,
\]
are independent, mean zero, and satisfy \(\Var(\eta_j\mid J)=v_{X,\tau,j}\).
Hence
\[
A_K
=
\frac{1}{\sqrt{D_X}}\cdot\frac{1}{\sqrt M}\sum_{j\in J}\eta_j.
\]
The transformed fourth-moment bound in Assumption~\ref{ass:weak-regularity}(ii) and the same convexity calculation used in Proposition~\ref{prop:focal-moment-convergence} imply
\[
\frac{1}{M}\sum_{j\in J}\E(\eta_j^4\mid J)=O_p(1).
\]
Also, by the finite-population law of large numbers for the sampled set \(J\),
\[
\frac{1}{M}\sum_{j\in J}v_{X,\tau,j}\xrightarrow{p}W_{X,\tau}^2.
\]
If \(W_{X,\tau}^2>0\), the Lyapunov ratio satisfies
\[
\frac{\sum_{j\in J}\E(\eta_j^4\mid J)}
{\left(\sum_{j\in J}v_{X,\tau,j}\right)^2}
=O_p(M^{-1})\to 0.
\]
Therefore Lyapunov's central limit theorem, conditional on \(J\), yields
\[
A_K
\Rightarrow
N\!\left(0,\frac{W_{X,\tau}^2}{D_X}\right).
\]
If \(W_{X,\tau}^2=0\), then
\[
\Var(A_K\mid J)
=
\frac{1}{D_X}\cdot\frac{1}{M}\sum_{j\in J}v_{X,\tau,j}
\xrightarrow{p}0,
\]
and conditional Chebyshev's inequality gives \(A_K\xrightarrow{p}0\).
Consequently, in all cases, the conditional characteristic function of \(A_K\) given \(J\) converges in probability to
\[
\exp\{-t^2W_{X,\tau}^2/(2D_X)\}.
\]

\emph{Step 3: the label-randomization term.}
Conditional on \((J,U)\), the vector \((A_j)_{j\in J}\) is a two-arm completely randomized assignment on the \(M\) cluster-level units with potential outcomes
\[
\{X_j(s),X_j(s')\}_{j\in J}.
\]
Let
\[
\pi_{s,K}:=\frac{K_s}{M},
\qquad
\pi_{s',K}:=\frac{K_{s'}}{M},
\]
and define the finite-population complete-randomization variance
\[
V_{X,3,K}(J,U)
:=
\frac{S_{X,s,F}^2(J,U)}{\pi_{s,K}}
+
\frac{S_{X,s',F}^2(J,U)}{\pi_{s',K}}
-
S_{X,\tau,F}^2(J,U).
\]
This quantity is nonnegative. Indeed,
\[
V_{X,3,K}(J,U)
=
\frac{
S_{\pi_{s',K}X(s)+\pi_{s,K}X(s'),F}^2(J,U)
}{
\pi_{s,K}\pi_{s',K}
},
\]
where the numerator denotes the finite-population variance over \(j\in J\) of \(\pi_{s',K}X_j(s)+\pi_{s,K}X_j(s')\).

Proposition~\ref{prop:focal-moment-convergence} gives
\[
S_{X,a,F}^2(J,U)\xrightarrow{p}S_{X,a}^2,
\qquad a\in\{s,s'\},
\qquad
S_{X,\tau,F}^2(J,U)\xrightarrow{p}S_{X,\tau}^2.
\]
Together with \(\pi_{s,K}\to\pi_s\) and \(\pi_{s',K}\to\pi_{s'}\), this implies
\[
V_{X,3,K}(J,U)\xrightarrow{p}D_X-S_{X,\tau}^2.
\]
Since \(V_{X,3,K}(J,U)\ge 0\), necessarily \(D_X-S_{X,\tau}^2\ge 0\).
Assumption~\ref{ass:weak-regularity}(iv) gives the maximal-negligibility condition needed for the random focal finite population \(\{X_j(s),X_j(s'),X_j^\tau\}_{j=1}^K\), and hence also for the subpopulation indexed by \(J\).

If \(D_X-S_{X,\tau}^2>0\), then \(V_{X,3,K}(J,U)\) is bounded away from zero with probability approaching one, and the standard H\'ajek finite-population CLT for complete randomization applies conditionally on \((J,U)\).
Hence
\[
\frac{\sqrt M\{\hat\tau_{X,U}(A_J)-\bar X_J^\tau\}}
{\sqrt{V_{X,3,K}(J,U)}}
\Rightarrow N(0,1).
\]
Since \(V_{X,3,K}(J,U)\xrightarrow{p}D_X-S_{X,\tau}^2\), Slutsky's theorem gives
\[
\sqrt M R_K^{(3)}
\Rightarrow
N(0,D_X-S_{X,\tau}^2),
\]
and therefore
\[
B_K
\Rightarrow
N\!\left(0,1-\frac{S_{X,\tau}^2}{D_X}\right).
\]

If \(D_X-S_{X,\tau}^2=0\), the standardized CLT above is not invoked.
Instead, using the finite-population variance formula for complete randomization,
\[
\E\{\sqrt M R_K^{(3)}\mid J,U\}=0,
\qquad
\Var\{\sqrt M R_K^{(3)}\mid J,U\}=V_{X,3,K}(J,U).
\]
Because \(V_{X,3,K}(J,U)\xrightarrow{p}0\), conditional Chebyshev's inequality implies
\[
\sqrt M R_K^{(3)}\xrightarrow{p}0.
\]
Thus the same conclusion holds with \(N(0,0)\) interpreted as the degenerate normal distribution:
\[
\sqrt M R_K^{(3)}
\Rightarrow
N(0,D_X-S_{X,\tau}^2),
\qquad
B_K
\Rightarrow
N\!\left(0,1-\frac{S_{X,\tau}^2}{D_X}\right).
\]
Equivalently, in all cases, the conditional characteristic function of \(B_K\) given \((J,U)\) converges in probability to
\[
\exp\left\{
-\frac{t^2}{2}\left(1-\frac{S_{X,\tau}^2}{D_X}\right)
\right\}.
\]

For the conditional permutation limit, the observed pooled array is \(\{X_j^{\mathrm{obs}}:j\in J\}\). Along the subsequence under consideration, its variance converges to
\[
\pi_s S_{X,s}^2+\pi_{s'}S_{X,s'}^2
+\pi_s\pi_{s'}(\bar\mu_{X,s}-\bar\mu_{X,s'})^2>0,
\]
where \(\bar\mu_{X,a}\) is a subsequential limit of \(K^{-1}\sum_j\mu_{X,a,j}\); under the weak null the last term is zero. The moment bounds and complete randomization give this convergence, and the maximal-negligibility condition transfers to the observed pooled array because each observed value is one of the two potential values. Under a uniform permutation, each arm variance converges in conditional probability to the same pooled variance. The finite-population central limit theorem followed by studentization therefore gives the conditional permutation limit.
Let \(T_{U,\lambda}^{\mathrm{Ney},\pi}\) denote the statistic in \eqref{eq:stat-studentized} computed under a uniform relabeling \(\pi\in\calA_J^{\mathrm{obs}}\), using the same fixed transformed focal outcomes.
Then
\[
\sup_{t\in\mathbb R}
\left|
\Pr_\pi\!\left\{
T_{U,\lambda}^{\mathrm{Ney},\pi}\le t
\,\middle|\,
J,U,Y^{\mathrm{obs}}
\right\}
-
\Phi(t)
\right|
\xrightarrow{p}0,
\]
where \(\Pr_\pi\) denotes probability under the uniform relabeling distribution.
Studentization is essential in this step because the permutation distribution uses the observed focal-cluster outcomes as a single finite population, whereas the weak null does not make the missing focal potential outcomes equal to the observed ones.

\emph{Step 4: combine the three limits.}
The components are generated sequentially.
The term \(C_K\) is measurable with respect to \(J\).
Conditional on \(J\), the characteristic function of \(A_K\) converges in probability to a nonrandom limit.
Therefore bounded convergence and iterated expectations imply
\[
C_K+A_K
\Rightarrow
N\!\left(0,
\frac{(1-\kappa)V_{X,\mu}+W_{X,\tau}^2}{D_X}
\right).
\]
Conditioning once more on \((J,U)\), and using the conditional characteristic function limit for \(B_K\), gives
\[
C_K+A_K+B_K
\Rightarrow
N\!\left(0,
\frac{(1-\kappa)V_{X,\mu}+W_{X,\tau}^2}{D_X}
+
1-
\frac{S_{X,\tau}^2}{D_X}
\right).
\]
Because \(S_{X,\tau}^2=V_{X,\mu}+W_{X,\tau}^2\), the limiting variance simplifies as
\[
\frac{(1-\kappa)V_{X,\mu}+W_{X,\tau}^2}{D_X}
+
1-
\frac{S_{X,\tau}^2}{D_X}
=
1-
\frac{\kappa V_{X,\mu}}{D_X}
\le 1.
\]
The variance is nonnegative because it is the limit of the variance of \(C_K+A_K+B_K\).
Equivalently, it can be written as
\[
\frac{(1-\kappa)V_{X,\mu}+W_{X,\tau}^2+D_X-S_{X,\tau}^2}{D_X}\ge 0.
\]
Moreover, by Proposition~\ref{prop:focal-moment-convergence},
\[
T_{U,\lambda}^{\mathrm{Ney}}(A_J)-(C_K+A_K+B_K)=o_p(1),
\]
so
\[
T_{U,\lambda}^{\mathrm{Ney}}(A_J)
\Rightarrow
N\!\left(0,
1-
\frac{\kappa V_{X,\mu}}{D_X}
\right).
\]
This proves that the observed studentized statistic has a limiting normal variance no larger than one, while the conditional permutation distribution is asymptotically standard normal.

Let \(q_{1-\alpha,K}\) be the conditional \((1-\alpha)\)-quantile of the studentized relabeling distribution.
The conditional permutation CLT above implies
\[
q_{1-\alpha,K}\xrightarrow{p} z_{1-\alpha},
\]
where \(z_{1-\alpha}\) is the \((1-\alpha)\)-quantile of a standard normal random variable.
Because \(\alpha\in(0,1/2)\), \(z_{1-\alpha}>0\).
Hence, along the subsequence under consideration,
\begin{align*}
\limsup_{K\to\infty}\Prob\{p_{T,\lambda}\le\alpha\}
&\le
\limsup_{K\to\infty}
\Prob\{T_{U,\lambda}^{\mathrm{Ney}}(A_J)\ge q_{1-\alpha,K}\} \\
&\le
\Prob\{Z_\sigma\ge z_{1-\alpha}\}
\le
\alpha,
\end{align*}
where \(Z_\sigma\sim N(0,\sigma^2)\) with
\[
\sigma^2=1-\frac{\kappa V_{X,\mu}}{D_X}\in[0,1].
\]
Since every subsequence has a further subsequence satisfying the same bound, the full sequence obeys
\[
\limsup_{K\to\infty}\Prob\{p_{T,\lambda}\le\alpha\}\le\alpha.
\]
This proves asymptotic validity for the weighted weak null.
\end{proof}

\subsection{Unconditional variance and centering}
\label{app:weak-centering}
\paragraph{Variance over focal selection and assignment.}
Use the independent auxiliary focal samples in all \(K\) clusters defined in Appendix~\ref{app:weak-regularity}. Write
\[
V_{\mu,a,K}=\frac{1}{K-1}\sum_{j=1}^K(\mu_{X,a,j}-\bar\mu_{X,a,K})^2,
\qquad
\bar v_{X,a,K}=K^{-1}\sum_{j=1}^K v_{X,a,j},
\]
for \(a=s,s'\), and let \(V_{\mu,\tau,K}\) be the corresponding variance of \(\mu_{X,\tau,j}\). Under complete randomization, the observed numerator has the exact variance
\begin{equation}\label{eq:weak-unconditional-variance}
\Var(\hat\tau_{X,U})
=\frac{V_{\mu,s,K}+\bar v_{X,s,K}}{K_s}
 +\frac{V_{\mu,s',K}+\bar v_{X,s',K}}{K_{s'}}
 -\frac{V_{\mu,\tau,K}}{K}.
\end{equation}
For \(K_s,K_{s'}\ge2\), the expectation of \(\hat V_{X,U}^{\mathrm{Ney}}\) is the sum of the first two terms, so its excess over the true variance is \(V_{\mu,\tau,K}/K\ge0\).
To verify \eqref{eq:weak-unconditional-variance}, condition on the full auxiliary focal array and apply the complete-randomization variance formula. The expectation of each focal-array variance is the variance of its expectations plus the average focal-sampling variance. The law of total variance adds \(K^{-1}\bar v_{X,\tau,K}\), which cancels the corresponding term in the conditional treatment-effect variance. This identity explains why the weak-null guarantee averages over focal selection, but does not replace the central limit theorem.

The preceding exact variance and expectation calculations also apply under deterministic centering with the corresponding centered moments; validity with the estimated center \(\widehat c_K\) is established asymptotically in Proposition~\ref{prop:weak-centered}.

\paragraph{Centering before weighting.}
When the weights differ across clusters, the original transformed statistic need not be invariant to adding a constant to all outcomes. A location-invariant alternative uses
\begin{equation}\label{eq:weak-centered-outcome}
\widehat c_K=\frac{1}{M_K}\sum_{j\in J_K}\widetilde Y_j^{\mathrm{obs}},
\qquad
X_j^{\mathrm{obs},c}=w_{jK}(\widetilde Y_j^{\mathrm{obs}}-\widehat c_K),
\end{equation}
where \(M_K=K_s+K_{s'}\). Hold \(\widehat c_K\) fixed under every permutation and apply the same studentized statistic to \(X^{\mathrm{obs},c}\). Adding a constant to the outcomes leaves this array unchanged; replacing \(Y\) by \(1-Y\) negates it and reverses the tail of the test. In particular, centering is performed before multiplication by \(w_{jK}\), not afterward.

\begin{proposition}[Centered weighted weak-null test]\label{prop:weak-centered}
Suppose the weighted weak null holds. Let \(c_K\) be a bounded deterministic sequence such that \(\widehat c_K-c_K=o_p(1)\), and assume \(K^{-1}\sum_jw_{jK}^4=O(1)\). If Assumption~\ref{ass:weak-regularity} holds for the centered auxiliary array
\(X_j^c(a)=w_{jK}\{\widetilde Y_j(a)-c_K\}\), then the permutation test based on \eqref{eq:weak-centered-outcome} satisfies the conclusion of Theorem~\ref{thm:weak-null-asymptotic}.
\end{proposition}
\begin{proof}
Subtracting \(c_K\) from both potential outcomes leaves every effect and the weighted null unchanged, so Theorem~\ref{thm:weak-null-asymptotic} applies to the oracle-centered array. Replacing \(c_K\) by \(\widehat c_K\) changes its observed numerator by
\[
-(\widehat c_K-c_K)\left\{
\frac{1}{K_s}\sum_{j\in J_K}w_{jK}\1\{A_j=s\}
-\frac{1}{K_{s'}}\sum_{j\in J_K}w_{jK}\1\{A_j=s'\}\right\}.
\]
The expression in braces is \(O_p(K^{-1/2})\) by complete randomization and the weight moment bound. Thus the numerator changes by \(o_p(K^{-1/2})\). Cauchy--Schwarz and the same moment bounds imply that multiplying the change in its variance estimator by \(K\) gives \(o_p(1)\). The observed studentized statistics are therefore asymptotically equivalent.

Conditional on the observed array, a uniform permutation has the same \(O_{p_\pi}(K^{-1/2})\) weight contrast: its variance is the pooled weight variance times \(1/K_s+1/K_{s'}\). The preceding numerator and variance arguments apply in conditional probability to the permutation statistics. Hence the centered permutation distribution is asymptotically standard normal, while the observed statistic has the same conservative normal limit as its oracle-centered counterpart. The positive-critical-value argument in Theorem~\ref{thm:weak-null-asymptotic} completes the proof.
\end{proof}

For the pooled mean in \eqref{eq:weak-centered-outcome}, a sufficient moment condition, under the sampling-fraction conditions above, is
\[
\sup_K\frac1K\sum_{j=1}^K\frac1{n_j}
\sum_{i\in\calI_j}|Y_i(0,a)|^4<\infty,
\qquad a\in\{s,s'\}.
\]
The finite-population laws of large numbers then give
\[
\widehat c_K-c_K=o_p(1),\qquad
c_K=\frac{K_s}{M_K}\frac1K\sum_j\mu_{s,j}
   +\frac{K_{s'}}{M_K}\frac1K\sum_j\mu_{s',j}.
\]
The additional weight moment condition must also be checked; it is not implied for arbitrary weights by a bound on the original transformed outcomes alone. These conditions hold, for example, for bounded outcomes and uniformly bounded weights when the centered limiting variances are positive.
Under the individual equality null, fixing the pooled center preserves exact conditional validity by Theorem~\ref{thm:two-level-exact}(a). No finite-sample weak-null or bounded-null claim is made for the studentized statistic. A two-sided weak-null test uses the absolute studentized statistic and the same argument; its asymptotic bound holds for every \(\alpha\in(0,1)\). If Monte Carlo rather than enumeration is used for weak-null inference, taking the number of reference draws to infinity gives the same asymptotic conclusion.

\section{Proofs for the Monotone PIRT}\label{app:pirt-proofs}

\subsection{Proof of Proposition~\ref{prop:pairwise-ordering}}
\label{app:proof-pairwise-ordering}

\begin{proof}
Take any two assignments \(z=(a,d)\) and \(z'=(a',d')\) in \(\calZ\). For any
\(i\in\mathbb I_M(z,z')\), by definition,
\[
d_i=d_i'=0,
\qquad
a_{[i]},a'_{[i]}\in\calS_M.
\]
By Assumption~\ref{ass:hpi},
\[
Y_i(z)=Y_i(0,a_{[i]}),
\qquad
Y_i(z')=Y_i(0,a'_{[i]}).
\]
Under the monotone null \(H_{0,\mathrm M}(\calS_M)\), if
\(a_{[i]}>a'_{[i]}\), then
\[
Y_i(z)=Y_i(0,a_{[i]})
\ge
Y_i(0,a'_{[i]})=Y_i(z').
\]
If \(a_{[i]}<a'_{[i]}\), then
\[
Y_i(z)=Y_i(0,a_{[i]})
\le
Y_i(0,a'_{[i]})=Y_i(z').
\]
If \(a_{[i]}=a'_{[i]}\), Assumption~\ref{ass:hpi} gives
\[
Y_i(z)=Y_i(z').
\]
Therefore, for every \(i\in\mathbb I_M(z,z')\),
\[
(a_{[i]}-a'_{[i]})\{Y_i(z)-Y_i(z')\}\ge 0.
\]
By the monotonicity condition in
Definition~\ref{def:pairwise-monotone-stat},
\[
T\bigl(Y(z),z,z'\bigr)
\ge
T\bigl(Y(z'),z,z'\bigr).
\]
\end{proof}

\subsection{Proof of Theorem~\ref{thm:pirt-valid}}
\label{app:proof-pirt-valid}

\begin{proof}
Write
\[
\pi(z)=P(z),
\qquad z\in\calZ.
\]
For a fixed realized assignment \(z\), define
\[
p(z)
=
\sum_{\tilde z\in\calZ}
\1\!\left[
T\bigl(Y(z),z,\tilde z\bigr)
\ge
T\bigl(Y(z),\tilde z,z\bigr)
\right]\pi(\tilde z).
\]
By pairwise imputability, when \(Z^{\mathrm{obs}}=z\), this quantity equals the
observable PIRT value \(p_{\mathrm M}^{\mathrm{PIRT}}(z)\). Therefore
\[
p_{\mathrm M}^{\mathrm{PIRT}}(Z^{\mathrm{obs}})=p(Z^{\mathrm{obs}}).
\]
For \(\alpha\in(0,1)\), let
\[
\calZ_\alpha
=
\{z\in\calZ:p(z)\le \alpha/2\},
\qquad
w_\alpha
=
\sum_{z\in\calZ_\alpha}\pi(z).
\]
Then, since \(Z^{\mathrm{obs}}\sim P\),
\[
\E_P\!\left[
\1\{p_{\mathrm M}^{\mathrm{PIRT}}(Z^{\mathrm{obs}})\le \alpha/2\}
\right]
=
w_\alpha.
\]
If \(w_\alpha=0\), there is nothing to prove. Suppose \(w_\alpha>0\).

Define
\[
H(z,\tilde z)
=
\1\!\left[
T\bigl(Y(z),z,\tilde z\bigr)
\ge
T\bigl(Y(\tilde z),\tilde z,z\bigr)
\right],
\qquad
z,\tilde z\in\calZ.
\]
For every pair \((z,\tilde z)\), at least one of the two weak inequalities
\[
T\bigl(Y(z),z,\tilde z\bigr)
\ge
T\bigl(Y(\tilde z),\tilde z,z\bigr)
\]
or
\[
T\bigl(Y(\tilde z),\tilde z,z\bigr)
\ge
T\bigl(Y(z),z,\tilde z\bigr)
\]
must hold. Hence
\[
H(z,\tilde z)+H(\tilde z,z)\ge 1.
\]

Next, Proposition~\ref{prop:pairwise-ordering}, applied to the ordered pair
\((\tilde z,z)\), gives
\[
T\bigl(Y(\tilde z),\tilde z,z\bigr)
\ge
T\bigl(Y(z),\tilde z,z\bigr).
\]
Therefore, for every \(z\in\calZ\),
\begin{align*}
p(z)
&=
\sum_{\tilde z\in\calZ}
\1\!\left[
T\bigl(Y(z),z,\tilde z\bigr)
\ge
T\bigl(Y(z),\tilde z,z\bigr)
\right]\pi(\tilde z) \\
&\ge
\sum_{\tilde z\in\calZ}
\1\!\left[
T\bigl(Y(z),z,\tilde z\bigr)
\ge
T\bigl(Y(\tilde z),\tilde z,z\bigr)
\right]\pi(\tilde z) \\
&=
\sum_{\tilde z\in\calZ}H(z,\tilde z)\pi(\tilde z).
\end{align*}

Now define
\[
\Sigma
=
\sum_{z\in\calZ_\alpha}
\sum_{\tilde z\in\calZ}
H(z,\tilde z)\pi(z)\pi(\tilde z).
\]
On the one hand,
\[
\Sigma
\le
\sum_{z\in\calZ_\alpha}p(z)\pi(z)
\le
\frac{\alpha}{2}
\sum_{z\in\calZ_\alpha}\pi(z)
=
\frac{\alpha w_\alpha}{2}.
\]

On the other hand,
\begin{align*}
\Sigma
&\ge
\sum_{z\in\mathcal Z_\alpha}
\sum_{\tilde z\in\mathcal Z_\alpha}
H(z,\tilde z)\pi(z)\pi(\tilde z) \\
&=
\frac{1}{2}
\sum_{z\in\mathcal Z_\alpha}
\sum_{\tilde z\in\mathcal Z_\alpha}
\bigl\{H(z,\tilde z)+H(\tilde z,z)\bigr\}\pi(z)\pi(\tilde z) \\
&\ge
\frac{1}{2}
\sum_{z\in\mathcal Z_\alpha}
\sum_{\tilde z\in\mathcal Z_\alpha}
\pi(z)\pi(\tilde z)
+
\frac{1}{2}
\sum_{z\in\mathcal Z_\alpha}\pi(z)^2 \\
&>
\frac{w_\alpha^2}{2},
\end{align*}
where the extra term uses \(H(z,z)=1\) for every \(z\in\mathcal Z_\alpha\).

Combining the two bounds yields
\[
\frac{w_\alpha^2}{2}
<
\frac{\alpha w_\alpha}{2}.
\]
Since \(w_\alpha>0\), it follows that \(w_\alpha\le\alpha\). Therefore
\[
\E_P\!\left[
\1\{p_{\mathrm M}^{\mathrm{PIRT}}(Z^{\mathrm{obs}})\le \alpha/2\}
\right]
=
w_\alpha
<
\alpha.
\]
This proves finite-sample validity of Procedure~\ref{proc:pirt}.
\end{proof}

\subsection{Conditional monotone PIRT and Monte Carlo implementation}
\label{app:conditional-pirt}

Let $G:\calZ\to\mathcal G$ be a map fixed before assignments and outcomes
are observed. For $g\in G(\calZ)$ with positive design probability, define
$\calZ_g=\{z:G(z)=g\}$ and $P_g(\cdot)=P\{\cdot\mid G(Z)=g\}$.
The observed assignment selects the cell $g=G(Z^{\mathrm{obs}})$.
For that cell define the raw conditional PIRT value
\begin{equation}\label{eq:conditional-pirt-pvalue}
p_{M,g}^{\mathrm{PIRT}}(Z^{\mathrm{obs}})
=
\sum_{z\in\calZ_g}
\1\!\left[
T\bigl(Y^{\mathrm{obs}},Z^{\mathrm{obs}},z\bigr)
\ge T\bigl(Y^{\mathrm{obs}},z,Z^{\mathrm{obs}}\bigr)
\right]P_g(z).
\end{equation}

\begin{corollary}[Conditional validity of monotone PIRT]
\label{cor:conditional-monotone-pirt}
Suppose the assumptions of Theorem~\ref{thm:pirt-valid} hold, and let $G$
be prespecified as above. For every positive-probability cell and
$\alpha\in(0,1)$,
\[
\Pr\!\left\{
\min\{1,2p_{M,g}^{\mathrm{PIRT}}(Z^{\mathrm{obs}})\}\le\alpha
\mid G(Z^{\mathrm{obs}})=g
\right\}\le\alpha.
\]
The same conclusion holds for the Monte Carlo implementation: conditional
on the observed cell, draw $Z^{(1)},\ldots,Z^{(B)}$ independently from
$P_g$ and define
\begin{equation}\label{eq:conditional-pirt-mc}
\widehat p_{M,g}^{\mathrm{PIRT}}
=
\frac{1+\sum_{b=1}^{B}
\1\!\left[
T\bigl(Y^{\mathrm{obs}},Z^{\mathrm{obs}},Z^{(b)}\bigr)
\ge T\bigl(Y^{\mathrm{obs}},Z^{(b)},Z^{\mathrm{obs}}\bigr)
\right]}{B+1}.
\end{equation}
Then $\min\{1,2\widehat p_{M,g}^{\mathrm{PIRT}}\}$ is super-uniform
conditional on $g$, with probability taken jointly over the original
assignment and the Monte Carlo draws. Both conclusions also hold
unconditionally by iterated expectations. A constant $G$ recovers the
original design law and its plus-one Monte Carlo implementation.
\end{corollary}

\begin{proof}
For exact enumeration, Proposition~\ref{prop:pairwise-ordering} continues
to hold for every pair in $\calZ_g$. The proof of
Theorem~\ref{thm:pirt-valid}, with $\calZ_g$ and $P_g$ in place of
$\calZ$ and $P$, therefore proves the conditional assertion.

For the finite Monte Carlo claim, set $n=B+1$, $Z^{(0)}=Z^{\mathrm{obs}}$,
and condition on $g$. The assignments $Z^{(0)},\ldots,Z^{(B)}$ are
independent draws from $P_g$. For $i,j\in\{0,\ldots,B\}$, let
\[
h_{ij}
=\1\!\left[
T\bigl(Y(Z^{(i)}),Z^{(i)},Z^{(j)}\bigr)
\ge T\bigl(Y(Z^{(i)}),Z^{(j)},Z^{(i)}\bigr)
\right],
\qquad p_i=\frac1n\sum_{j=0}^{B}h_{ij}.
\]
The diagonal has $h_{ii}=1$. Pairwise ordering implies
$h_{ij}+h_{ji}\ge1$: if $h_{ij}=0$, then
\[
\begin{split}
T\bigl(Y(Z^{(j)}),Z^{(j)},Z^{(i)}\bigr)
&\ge T\bigl(Y(Z^{(i)}),Z^{(j)},Z^{(i)}\bigr)\\
&> T\bigl(Y(Z^{(i)}),Z^{(i)},Z^{(j)}\bigr)\\
&\ge T\bigl(Y(Z^{(j)}),Z^{(i)},Z^{(j)}\bigr),
\end{split}
\]
so $h_{ji}=1$.
Let $S=\{i:p_i\le\alpha/2\}$ and $r=|S|$.
Summing over the selected rows and then their internal comparisons gives
\[
\frac{r(r+1)}{2}
\le \sum_{i\in S}\sum_{j\in S}h_{ij}
\le \sum_{i\in S}\sum_{j=0}^{B}h_{ij}
\le \frac{r\alpha n}{2}.
\]
If $r>0$, this implies $r/n\le\alpha$; the same bound is immediate
when $r=0$. Exchangeability of row indices gives
$\Pr(p_0\le\alpha/2\mid g)=\E(r/n\mid g)\le\alpha$.
Since $p_0=\widehat p_{M,g}^{\mathrm{PIRT}}$, the claimed Monte Carlo
validity follows.
\end{proof}

For the specialization in Remark~\ref{rem:conditional-pirt}, take
$G=\Gamma_{\mathrm{out}}$. Its cells form a partition because any
assignment in the selected cell has the same outside cluster-label pairs
and hence selects the same cell again. Under complete first-stage
randomization, the remaining labels are uniformly relabeled with their
fixed margins. The second-stage law is unchanged conditional on those
labels. Assignments in outside clusters can be integrated out because
the statistic excludes their units. The conditional construction therefore
does not require a fixed set of individual controls within the retained
clusters. The partition formulation parallels
\citet[Section~4.2]{liu2026randomizationtestsswitchbackexperiments}; the
present proof uses the monotone ordering of
Proposition~\ref{prop:pairwise-ordering}, not an equality-null imputation.

\subsection{Pairwise mean and rank statistics}
\label{app:pirt-statistics}
For a pair \(z=(a,d),z'=(a',d')\), write
\[
\mathbb I_+=\{i\in\mathbb I_M(z,z'):a_{[i]}>a'_{[i]}\},\qquad
\mathbb I_-=\{i\in\mathbb I_M(z,z'):a_{[i]}<a'_{[i]}\}.
\]
We use the difference between the mean outcomes in these two groups. If either group is empty, both ordered scores are zero, so the comparison is a tie. This symmetric rule preserves Definition~\ref{def:pairwise-monotone-stat} and invariance to a common outcome shift.

A normalized rank alternative is
\[
T^{\mathrm{rk},c}(y,z,z')
=\frac{1}{|\mathbb I_+|}\sum_{i\in\mathbb I_+}\varphi(r_i)
-\frac{1}{|\mathbb I_-|}\sum_{i\in\mathbb I_-}\varphi(r_i),
\]
where ranks are taken within \(\mathbb I_+\cup\mathbb I_-\), \(\varphi\) is nondecreasing, and ties use either the fixed unit ordering or the average rank-score convention described in Section~\ref{subsec:monotone-statistics}. Use the same empty-group convention. For a fixed pair of assignments, raising outcomes in \(\mathbb I_+\) and lowering outcomes in \(\mathbb I_-\) cannot decrease the higher-group rank-score sum. The total rank-score sum is fixed, so the lower-group sum cannot increase. The denominators are outcome independent, verifying Definition~\ref{def:pairwise-monotone-stat}.

For binary outcomes with Wilcoxon scores \(\varphi(r)=r\) and average-rank ties, the normalized rank contrast has the same sign as the ordinary mean contrast, so it does not give a distinct PIRT comparison. A different choice restricts the two directional means to units whose two saturation labels are adjacent in the ordered set \(\calS_M\). More generally, nonnegative weights fixed by the unordered assignment pair may be used in each directional outcome mean, with a tie when either total weight is zero; the displayed normalized rank statistic remains unweighted. The adjacent-exposure choice can reduce cancellation between comparisons at different parts of the saturation range, but its power relative to the pooled mean score is a numerical question.

The Kruskal--Wallis, average cross-difference, and omnibus regression statistics in \citet{hoshino2023} apply to the multi-saturation equality test in Appendix~\ref{app:global-no-spillover}, but their two-sided construction need not satisfy the order restriction required here.

\section{Empirical Literature of Randomized Saturation Designs}
\label{app:empirical-literature}

Table~\ref{tab:empirical} summarizes empirical studies that use randomized saturation designs or closely related clustered designs with variation in treatment intensity across clusters. The table reports the number of clusters, the number of saturation levels, the average number of clusters per level, and the total number of units, showing that these designs appear in a wide range of applications. Many of these papers fit our setting because treatment intensity varies at the cluster level, while outcomes are measured at the unit level under partial interference. The prevalence of designs with multiple saturation levels and, in some cases, few clusters per level motivates finite-sample randomization-based inference methods.

\begin{table}[htbp]
    \centering
    \caption{Summary of Empirical Papers}
    \label{tab:empirical}
    
    \resizebox{\textwidth}{!}{ 
        \begin{tabular}{lcccc}
        \toprule
        Paper Reference & Number of Clusters & Number of Saturation Levels & Avg Clusters/Level & Total Units \\
        \midrule
        \citet{Baird2018} & 130 & 4 & 32.5 & 2226 \\
        \citet{Beuermann_2015} & 28 & 2 & 14 & NR \\
        \citet{repec:ucp:ecdecc:doi:10.1086/701213} & 17 & 2 & 8.5 & 875 \\
        \citet{Cr_pon_2013} & 235 & 5 & 47 & 29636 \\
        \citet{Cruces_2025} & 3982 & 4 & 995.5 & 68806 \\
        \citet{Duflo_2015} & 140 & 2 & 70 & 13500 \\
        \citet{10.1257/app.20130480} & 37 & 2 & 18.5 & NR \\
        \citet{Haushofer_2016} & 120 & 2 & 60 & 1440 \\
        \citet{Ichino_2012} & 39 & 2 & 19.5 & 868 \\
        \citet{Kinnan2020} & 424 & 5 & 84.8 & 10879 \\
        \citet{Kremer_2011} & 184 & 2 & 92 & 1384 \\
        \citet{10.1257/app.20180340} & 157 & 2 & 78.5 & 3537 \\
        \citet{Melis_2005} & 54 & 2 & 27 & 155 \\
        \citet{Miguel_2004} & 50 & 2 & 25 & 19493 \\
        \citet{Muralidharan_2015} & 180 & 2 & 90 & 6433 \\
        \citet{Sinclair_2012} & 4897 & 4 & 1224.25 & 64445 \\
        \bottomrule
        \end{tabular}     
    }
\end{table}
\clearpage
\section{Application and Simulation Details}
\label{app:malawi-numerics}

\subsection{Assignment law and test statistics}
\label{app:malawi-implementation}
The original roster has 665 schoolgirls in 42 EAs. The analysis excludes the singleton EA whose untreated outcome at saturation \(0.66\) is infeasible, leaving 664 schoolgirls in 41 EAs and 639 observed values for each outcome. Exclusion uses roster size, not outcomes. Conditional on the excluded EA's observed label and the geographic margins, the retained first-stage distribution is uniform over the corresponding within-stratum label allocations. Table~\ref{tab:malawi-block-margins} records these retained margins; Appendix~\ref{app:stratified-crt} gives the formal justification. Baseline dropouts were offered CCTs in both the no-schoolgirl-offer and UCT intervention EAs; the experiment's pure-control EAs are outside this analysis \citep[Section~2.3]{baird2011cash}. The maintained exposure restriction treats schoolgirl offers as sufficient within this restricted experiment, abstracting from peer identities and variation in transfer amounts.

\begin{table}[htbp]
\centering
\caption{Retained geographic margins in the reference distribution}
\label{tab:malawi-block-margins}
\begin{tabular}{lrrrr}\toprule Stratum & $0$ & $.33$ & $.66$ & $1$\\\midrule
far & 3 & 0 & 1 & 0 \\
near & 8 & 7 & 7 & 7 \\
urban & 3 & 2 & 1 & 2 \\
\bottomrule\end{tabular}

\end{table}

Given a label, the reference design draws exactly \(m_j(s)=\lfloor n_js+1/2\rfloor\) offers uniformly from the EA roster. Within-stratum exchangeability and this second-stage assignment mechanism are maintained for the illustration; matching observed offer counts alone does not identify the full assignment distribution. The margins are observed counts on which the analysis conditions, rather than additional quotas attributed to the original allocation algorithm. EA 194, observed at zero saturation in the nearby-rural block, is conditioned out; its label is subtracted from that block. All retained observed offer counts match the rounding rule. The pairwise relabeling spaces have sizes 64350, 10296, 102960, respectively. The distant-rural block has no 0.33 cell and contributes no label variation to comparisons involving that label.

The indicators \(O_{ik}\) of observed outcomes are treated as invariant over reference assignments. For each contrasted pair, select \(k_j=\min\{n_j-m_j(s),n_j-m_j(s')\}\ge1\) focal controls from the full retained roster and use \(F_{jk}=\{i\in U_j:O_{ik}=1\}\). An EA with no observed focal outcome receives zero outcome weight but retains its label in the randomization space. The exclusion of an EA for lack of an untreated exposure is distinct from this zero observed-outcome weight.

For candidate labels \(a_J\), let \(x_j=\1\{a_j=s\}\), \(v_j=|F_{jk}|\), and \(\bar y_j=v_j^{-1}\sum_{i\in F_{jk}}Y_i^k\) when \(v_j>0\). If \(h(j)\) denotes the geographic block, let \(\bar x_h\) be its \(v_j\)-weighted mean of \(x_j\). The weighted least-squares coefficient is
\begin{equation}\label{eq:malawi-block-stat}
T^{\mathrm{blk}}(a_J,\bar y)
=\frac{\sum_{j\in J}v_j(x_j-\bar x_{h(j)})\bar y_j}
{\sum_{j\in J}v_j(x_j-\bar x_{h(j)})^2}.
\end{equation}
A block with no outcome weight has \(\bar x_h=0\); a zero-denominator statistic is set to zero. For fixed labels and weights, the coefficient is nondecreasing in high-label outcomes and nonincreasing in low-label outcomes. All reported pairwise tests use \(\delta=0\). Focal sets and observation indicators remain fixed under label permutations.

Conditional PIRT fixes the excluded EA label, full-offer EA labels, and geographic margins, redraws the remaining labels and individual offers, and compares strict changed-exposure means on doubly untreated observed girls. If either directional group is empty, the comparison is a tie, as in Appendix~\ref{app:pirt-statistics}. The raw comparison fraction includes the observed assignment and weak ties; the reported \(p\)-value doubles this fraction and caps it at one. Adjacent and all-pairs CRT benchmarks use lower-minus-higher contrasts and Bonferroni correction. A single lowest-versus-highest CRT is also valid under the increasing-monotonicity null, without a multiplicity adjustment.

All observed-data CRT and PIRT \(p\)-values use Monte Carlo sampling with 1,999 independent reference draws plus the observed comparison, rather than enumeration of the reported reference spaces. The joint monotone null includes the supported 0-versus-0.66 restriction in the distant-rural stratum. Adjacent CRT--Bonferroni is valid for this joint null but does not test that restriction there, because 0.33 is unsupported.

\subsection{Calibration and potential-outcome schedules}
\label{app:malawi-calibration}
For outcome \(k\) in EA \(j\), let \(n_{jk0}\) be the number of observed untreated girls and \(S_{jk}\) their number of positive outcomes. The calibration fits
\begin{align}
S_{jk}\mid b_{jk}&\sim\operatorname{Binomial}(n_{jk0},\pi_{jk}),\label{eq:malawi-cal-binomial}\\
\operatorname{logit}(\pi_{jk})&=\alpha_k+\xi_k(A_j^{\mathrm{obs}})+b_{jk},
\qquad b_{jk}\sim N(0,\omega_k^2),\label{eq:malawi-cal-logit}
\end{align}
with \(\xi_k(\ell_0)=0\). Fitted parameters are fixed across replications. We draw EA effects from their fitted conditional distributions when untreated observations are available, and from the unconditional Gaussian distribution otherwise. Removing the saturation-cell effect gives \(p_{jk}^{0,*}=\operatorname{expit}(\widehat\alpha_k+b_{jk}^*)\).

\begin{table}[htbp]\centering
\caption{Outcome-specific calibration parameters on the retained roster}
\label{tab:malawi-calibration}
\begin{tabular}{lrr}\toprule Outcome & $\widehat\alpha_k$ & $\widehat\omega_k$\\\midrule
Currently enrolled & 0.3991 & 0.4310 \\
Reads an English letter & 2.8032 & 0.1339 \\
Ever married & -1.0896 & 0.5934 \\
Ever pregnant & -0.9383 & 0.0000 \\
\bottomrule\end{tabular}

\end{table}

Labels giving the same treated count share a shift, and the zero-count class remains at baseline. Shared unit uniforms in \eqref{eq:sim-potential-outcome} translate ordered probabilities into unitwise outcome restrictions.

Individual equality \eqref{eq:null-ps} is imposed by setting all shifts to zero. For the individual inequality nulls, let
\(\mathcal C_j(q_r)=\{q_\ell:m_j(q_\ell)=m_j(q_r)\}\)
on the tested grid \(0=q_0<\cdots<q_L\). For logit-template strength \(\zeta\in\{.25,.5,1\}\), set
\[
\gamma_{jk}(q_r;\zeta)=
\begin{cases}
0,&m_j(q_r)=0,\\[2pt]
\displaystyle\frac{\epsilon\zeta}{|\mathcal C_j(q_r)|}
\sum_{q\in\mathcal C_j(q_r)}\frac{q}{q_L},&m_j(q_r)>0.
\end{cases}
\]
Here \(\epsilon=1\) imposes increasing monotonicity \eqref{eq:null-monotone}, and \(\epsilon=-1\) imposes bounded individual spillovers \eqref{eq:null-bounded} with \(s>s'\) and \(\delta=0\). The count classes are contiguous, so averaging the ordered template preserves the required ordering in every EA; shared uniforms then impose each individual inequality. Positive strengths permit ties in binary outcomes and are not percentage-point effects. Average probability changes index alternative magnitudes, not the individual nulls.

Power magnitudes are averaged over observed schoolgirls in the 37 nearby-rural and urban EAs with positive assignment probability at every tested untreated label: 568 girls for enrollment, English-letter literacy, and pregnancy, and 570 for marriage. Assignment retains all 664 girls in 41 EAs.

Write $\Delta_k(q)=\sum_j w_{jk}\{p_{jk}(q)-p_{jk}(0)\}/\sum_jw_{jk}$, where $w_{jk}$ counts observed girls in this common-support population. Pairwise alternatives target $(\Delta_k(0),\Delta_k(.33),\Delta_k(.66))=(0,\tau/2,\tau)$. Diffuse declines target $(0,-\tau/2,-\tau)$ on three levels and $(0,-\tau/3,-2\tau/3,-\tau)$ on $(0,.22,.44,.66)$. Localized declines target zero at every earlier label and $-\tau$ at the final label. These are average probability changes, not identical effects in every EA.

For these power experiments, write $\beta=(\beta_1,\ldots,\beta_L)'$ and $\beta_0=0$. A count-class matrix $A_j$ maps the shifts to $\gamma_j=A_j\beta$: each positive feasible count class receives the average of its component shifts, and the zero-count class receives zero. The average probability-change equations determine the alternative magnitudes, and the required within-EA direction is verified before assignment. For localized paths, only a terminal class containing no earlier label can shift. For hump and trough schedules, this averaging map is replaced by the class-level rule assigning zero shift to every class containing an endpoint and a common signed shift to the remaining eligible interior classes. Shared uniforms give equal endpoint outcomes for every unit; magnitude denotes the maximal average interior departure.

Candidate magnitudes are $0,.02,.04,.05,.10,.15,.20,.30,.40,.50$. We retain a magnitude only when its probability schedules are feasible in all replications, without clipping. The common four-outcome grids end at .20 for diffuse/localized declines and troughs, and .05 for humps. Outcome-specific curves retain their longer feasible ranges.

The synthetic four-level design splits the original positive partial-offer slots into three groups, preserving the observed zero and full-offer slots. Its retained full margins are $(14,6,6,6,9)$ at $(0,.22,.44,.66,1)$. In the nearby-rural block the positive split is $(5,5,4)$, in the urban block $(1,1,1)$, and in the distant-rural block $(0,0,1)$. Three-level margins are $(14,9,9,9)$.

\subsection{Regression benchmark and rounded treatment counts}
\label{app:regression-benchmark}
The all-control benchmark regresses untreated individual outcomes on the higher-saturation indicator and geographic fixed effects and uses EA-clustered CR1 standard errors. It adapts the regression approach of \citet[Section~3.2]{baird2011cash} to saturation comparisons as in \citet[Section~3.1]{Baird2018}.

Because the regressors are constant within an EA, individual OLS has the same coefficient as weighted least squares of control-EA means with weights equal to their observed control counts. These weights change with the saturation label. For cluster types \(t\) with shares \(\pi_t\), baseline means \(\mu_t\), and control counts \(c_t(q)\), an unadjusted all-control contrast under individual equality can converge to
\[
\frac{\sum_t\pi_tc_t(s)\mu_t}{\sum_t\pi_tc_t(s)}
-\frac{\sum_t\pi_tc_t(s')\mu_t}{\sum_t\pi_tc_t(s')}.
\]
If \(c_t(q)=n_t(1-q)\), the label-specific factor cancels; rounded counts need not be proportional in this way.

For example, consider equal shares of size-two clusters with untreated outcome one and size-three clusters with untreated outcome zero. Both outcomes are invariant to saturation. Nearest-integer allocation at \(0.33\) leaves one and two controls, respectively, whereas allocation at \(0.66\) leaves one control in both types. The pooled control means converge to \(1/3\) and \(1/2\), and their difference converges to \(1/6\). With increasing numbers of clusters, the cluster-robust standard error tends to zero around this nonzero regression contrast. The problem is the contrast's composition, not a requirement that cluster sizes themselves tend to infinity.

Table~\ref{tab:matched-diagnostic} compares inference using all controls with inference on the CRT's focal outcomes and fixed weights.
Individual OLS and count-weighted EA-mean WLS give the same coefficient and EA score contributions. CR1 uses the individual sample size in its finite-sample correction and a $t_{G-1}$ reference distribution. Entries average the three contrasts under individual no-spillover equality; parentheses give Monte Carlo standard errors.

\begin{table}[htbp]\centering\small
\caption{All-control and matched-focal clustered inference}\label{tab:matched-diagnostic}
\begin{tabular}{lrrr}\toprule
Outcome & CRT & LPM-CR1 & Matched-focal CR1 \\ \midrule
Currently enrolled & 4.98 (0.30) & 7.67 (0.39) & 7.32 (0.36) \\
Reads an English letter & 5.25 (0.30) & 12.28 (0.49) & 8.73 (0.40) \\
Ever married & 4.75 (0.29) & 6.08 (0.34) & 6.80 (0.35) \\
Ever pregnant & 4.68 (0.30) & 5.95 (0.34) & 6.55 (0.35) \\
\bottomrule\end{tabular}\end{table}

Across the retained roster, implied control fractions at .33 range from 0.625 to 0.750, and at .66 from 0.250 to 0.375. The mean absolute difference between all-control and focal coefficients under individual equality ranges from 0.0104 to 0.0419 across outcome/contrast cells. These diagnostics do not isolate a unique cause of any CR1 size distortion.

\subsection{Additional simulation results}
\label{app:malawi-robustness}
The supplementary comparisons use the same simulated populations. The fixed-weight focal comparison uses complete generated outcomes, nested focal samples, equal EA weights, and geographic fixed effects. The shape comparisons retain the pooled PIRT statistic used in the main text.

\begin{figure}[htbp]\centering\includegraphics[width=\textwidth]{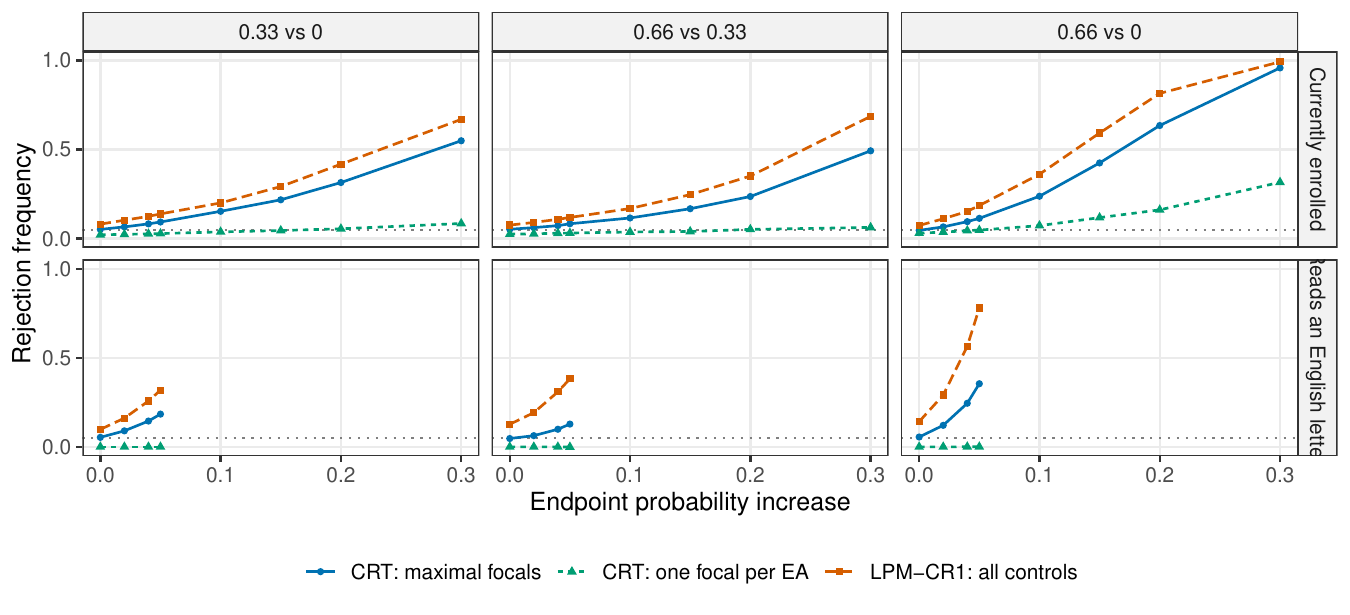}\caption{Enrollment and English-letter pairwise rejection frequencies on their complete feasible grids.}\label{fig:additional-pairwise}\end{figure}

\begin{figure}[htbp]\centering\includegraphics[width=\textwidth]{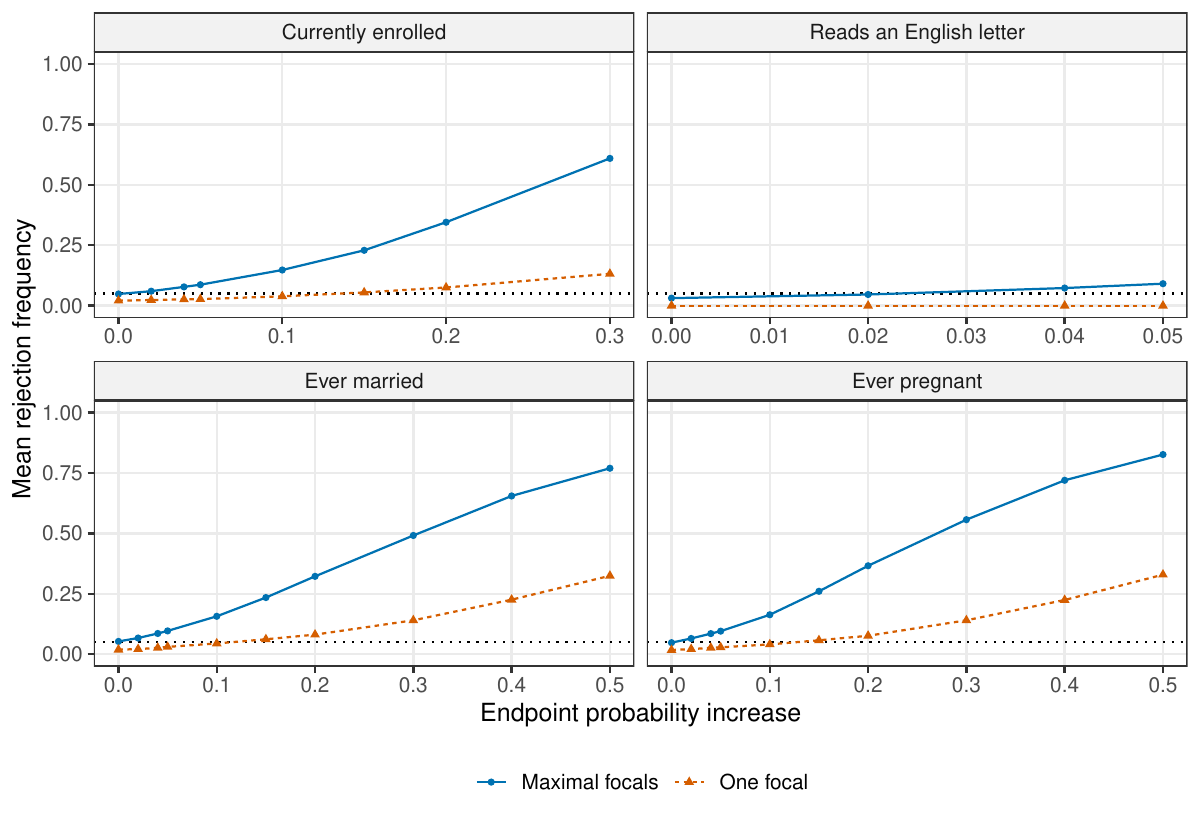}\caption{One versus maximal focal controls with equal prespecified EA weights.}\label{fig:focal-power}\end{figure}

\begin{figure}[htbp]\centering\includegraphics[width=\textwidth]{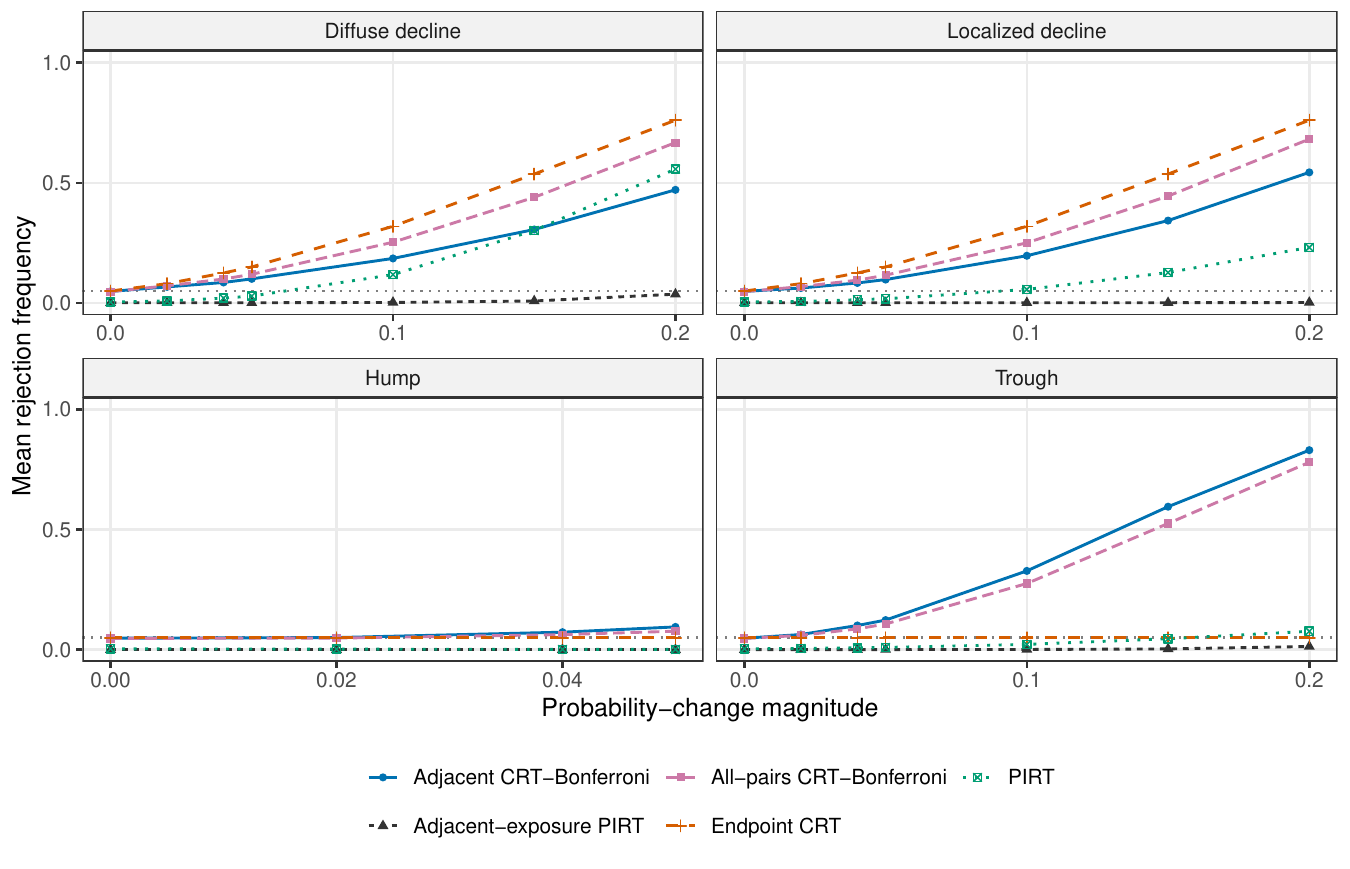}\caption{Three-level monotonicity comparisons across alternative shapes. The adjacent-exposure PIRT statistic restricts the pooled contrast to adjacent label changes. Hump magnitude stops at .05 on the common four-outcome grid because larger upward shifts are infeasible.}\label{fig:monotone-three}\end{figure}

\begin{figure}[htbp]\centering\includegraphics[width=\textwidth]{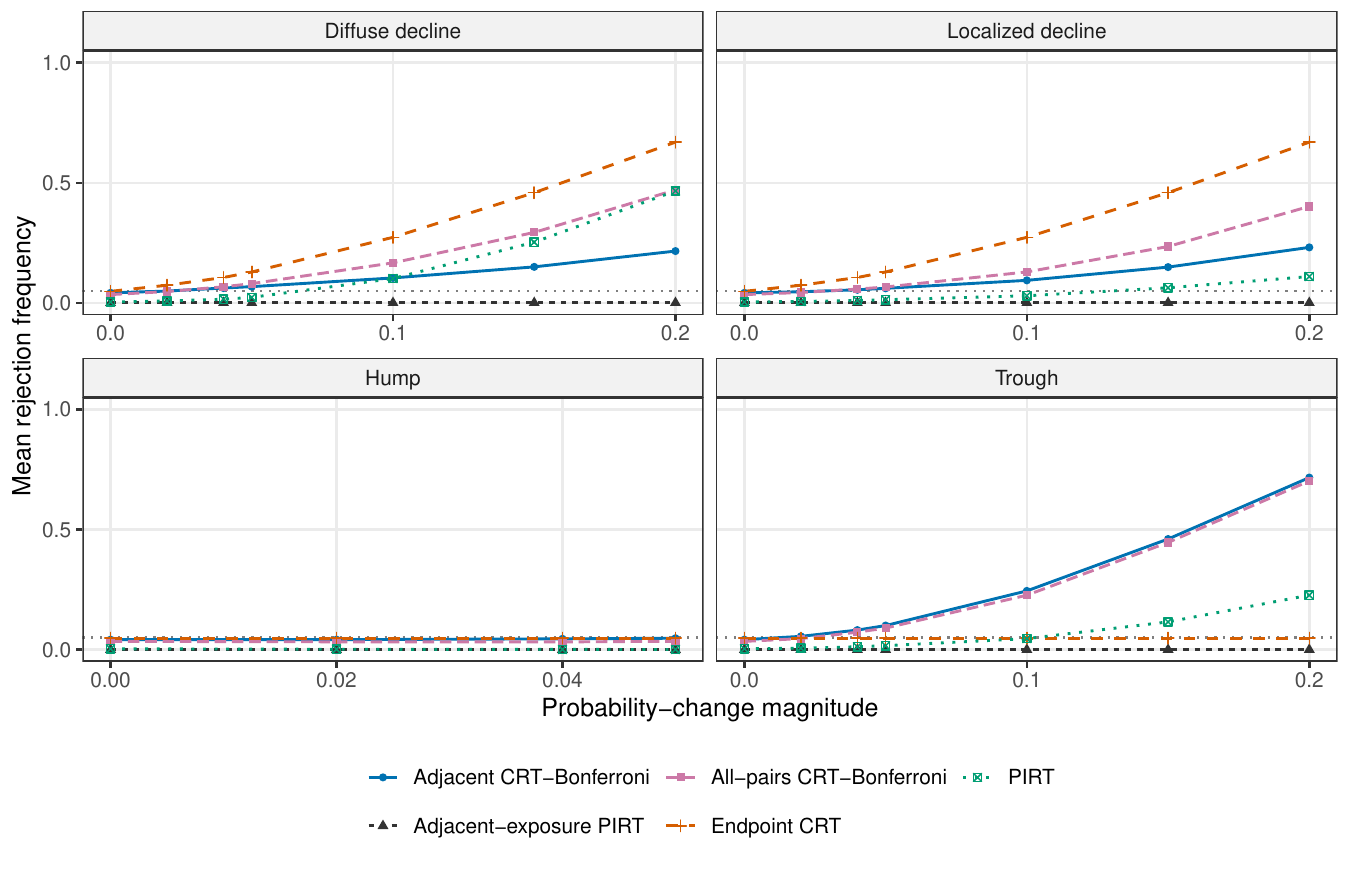}\caption{Four-level monotonicity comparisons across alternative shapes, including adjacent-exposure PIRT. Hump magnitude stops at .05 on the common four-outcome grid because larger upward shifts are infeasible.}\label{fig:monotone-secondary}\end{figure}

Hump and trough paths have equal endpoints for every unit, so the endpoint CRT's null remains true despite the interior violation. Adjacent and all-pairs CRT--Bonferroni have substantially greater power than pooled PIRT for these paths. Moving from three to four levels changes cell sizes and adjacent increments as well as the number of comparisons. Complete outcome-specific curves and numerical summaries are provided below and in the replication package.

\begin{table}[htbp]\centering\small
\caption{Rejection percentages under individual no-spillover equality}\label{tab:contrast-size}
\begin{tabular}{llrrr}\toprule
Outcome & Contrast & CRT: max & CRT: one & LPM--CR1 \\ \midrule
Currently enrolled & 0.33 vs 0 & 5.15 (0.49) & 2.10 (0.32) & 8.10 (0.61) \\
Currently enrolled & 0.66 vs 0.33 & 5.20 (0.50) & 2.50 (0.35) & 7.55 (0.59) \\
Currently enrolled & 0.66 vs 0 & 4.60 (0.47) & 3.05 (0.38) & 7.35 (0.58) \\
Reads an English letter & 0.33 vs 0 & 5.40 (0.51) & 0.05 (0.05) & 9.85 (0.67) \\
Reads an English letter & 0.66 vs 0.33 & 4.75 (0.48) & 0.05 (0.05) & 12.75 (0.75) \\
Reads an English letter & 0.66 vs 0 & 5.60 (0.51) & 0.05 (0.05) & 14.25 (0.78) \\
Ever married & 0.33 vs 0 & 4.85 (0.48) & 2.90 (0.38) & 6.65 (0.56) \\
Ever married & 0.66 vs 0.33 & 4.60 (0.47) & 2.50 (0.35) & 5.60 (0.51) \\
Ever married & 0.66 vs 0 & 4.80 (0.48) & 2.05 (0.32) & 6.00 (0.53) \\
Ever pregnant & 0.33 vs 0 & 4.45 (0.46) & 2.50 (0.35) & 6.45 (0.55) \\
Ever pregnant & 0.66 vs 0.33 & 4.15 (0.45) & 2.00 (0.31) & 6.05 (0.53) \\
Ever pregnant & 0.66 vs 0 & 5.45 (0.51) & 2.90 (0.38) & 5.35 (0.50) \\
\bottomrule\end{tabular}\end{table}

\begin{table}[htbp]\centering\small
\caption{Fixed-weight focal comparison: size and power at $\tau=.04$ (percent)}\label{tab:focal-comparison}
\begin{tabular}{lrrrrr}\toprule
Outcome & One: size & Max: size & One: power & Max: power & $\Delta$ (MCSE) \\ \midrule
Currently enrolled & 2.1 & 4.8 & 2.6 & 7.8 & 5.2 (0.4) \\
Reads an English letter & 0.0 & 3.1 & 0.0 & 7.3 & 7.3 (0.3) \\
Ever married & 1.9 & 5.3 & 2.6 & 8.6 & 5.9 (0.4) \\
Ever pregnant & 1.7 & 4.8 & 2.6 & 8.5 & 5.9 (0.4) \\
\bottomrule\end{tabular}\end{table}

\begin{table}[htbp]\centering\small
\caption{Rejection frequencies under the individual inequality nulls}\label{tab:nonflat-checks}
\begin{tabular}{lllr}\toprule
Design & Null & Method & Max. rejection (\%) \\ \midrule
four & $H_{0,\mathrm M}$ & Adjacent CRT-Bonferroni & 4.10 \\
three & $H_{0,\mathrm M}$ & Adjacent CRT-Bonferroni & 3.95 \\
four & $H_{0,\mathrm M}$ & Adjacent-exposure PIRT & 0.00 \\
three & $H_{0,\mathrm M}$ & Adjacent-exposure PIRT & 0.00 \\
four & $H_{0,\mathrm M}$ & All-pairs CRT-Bonferroni & 2.85 \\
three & $H_{0,\mathrm M}$ & All-pairs CRT-Bonferroni & 3.40 \\
three & $H_{0,\mathrm B}^{s,s'}$ & CRT & 4.00 \\
four & $H_{0,\mathrm M}$ & Endpoint CRT & 2.45 \\
three & $H_{0,\mathrm M}$ & Endpoint CRT & 2.80 \\
three & $H_{0,\mathrm B}^{s,s'}$ & LPM-CR1 & 9.50 \\
three & $H_{0,\mathrm B}^{s,s'}$ & Matched-focal CR1 & 6.95 \\
four & $H_{0,\mathrm M}$ & PIRT & 0.05 \\
three & $H_{0,\mathrm M}$ & PIRT & 0.10 \\
\bottomrule\end{tabular}\par\medskip\begin{minipage}{\textwidth}\footnotesize Entries maximize rejection percentages over the four outcomes and logit-template strengths $\zeta\in\{.25,.5,1\}$. Bounded-null rows use $s>s'$, $\delta=0$ in \eqref{eq:null-bounded}, and also maximize over the three contrasts. Increasing-monotonicity rows use \eqref{eq:null-monotone} on the stated three- or four-level design. These finite-grid maxima are not suprema over the composite nulls.\end{minipage}\end{table}

Observed means below weight each observed untreated girl equally; their counts are estimation counts, not eligible-population or design-roster sizes. Fully offered cells have no untreated observations.

\begin{table}[htbp]\centering\small
\caption{Observed untreated cell means}\label{tab:descriptive-controls}
\begin{tabular}{lrrr}\toprule
Outcome & Saturation & N & Mean \\ \midrule
Currently enrolled & 0 & 192 & 0.604 \\
Currently enrolled & 0.33 & 130 & 0.569 \\
Currently enrolled & 0.66 & 41 & 0.585 \\
Currently enrolled & 1 & 0 & n/a \\
Reads an English letter & 0 & 192 & 0.943 \\
Reads an English letter & 0.33 & 130 & 0.915 \\
Reads an English letter & 0.66 & 41 & 0.951 \\
Reads an English letter & 1 & 0 & n/a \\
Ever married & 0 & 191 & 0.251 \\
Ever married & 0.33 & 131 & 0.229 \\
Ever married & 0.66 & 42 & 0.214 \\
Ever married & 1 & 0 & n/a \\
Ever pregnant & 0 & 192 & 0.281 \\
Ever pregnant & 0.33 & 130 & 0.308 \\
Ever pregnant & 0.66 & 41 & 0.220 \\
Ever pregnant & 1 & 0 & n/a \\
\bottomrule\end{tabular}\end{table}

\begin{figure}[p]\centering\includegraphics[width=\textwidth,height=.85\textheight,keepaspectratio]{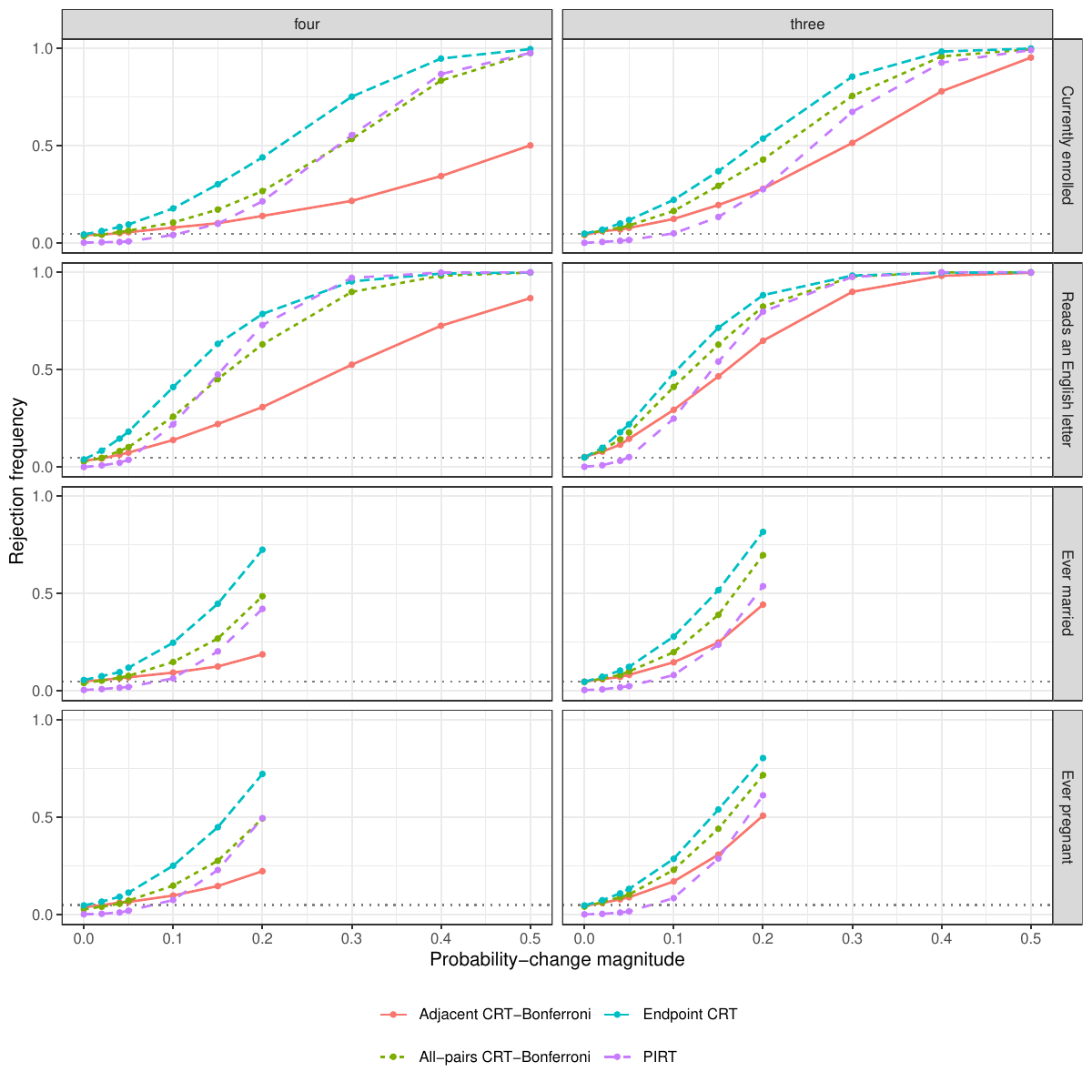}\caption{Outcome-specific diffuse monotonicity comparisons on complete feasible grids.}\label{fig:monotone_outcomes_diffuse}\end{figure}
\begin{figure}[p]\centering\includegraphics[width=\textwidth,height=.85\textheight,keepaspectratio]{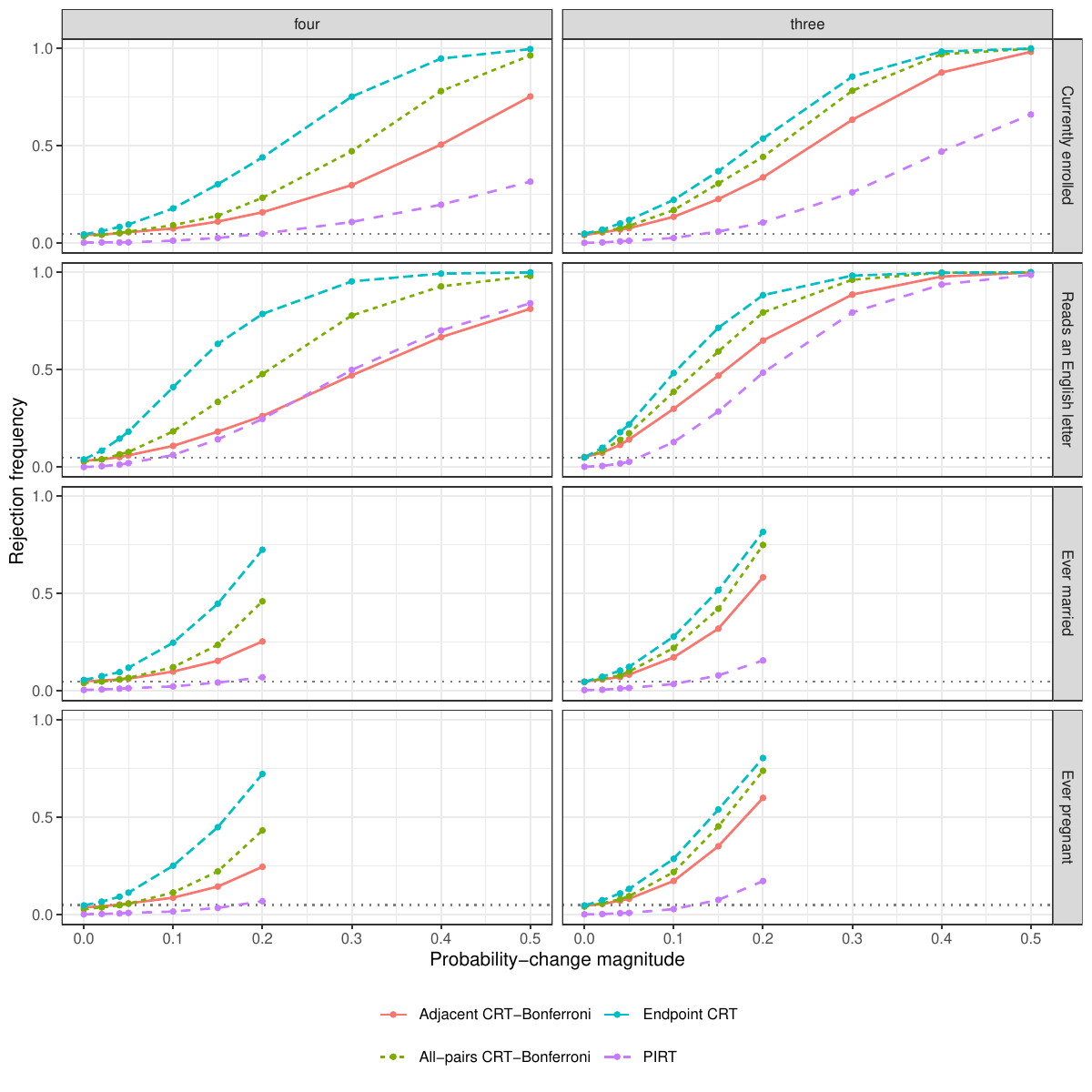}\caption{Outcome-specific localized monotonicity comparisons on complete feasible grids.}\label{fig:monotone_outcomes_localized}\end{figure}
\begin{figure}[p]\centering\includegraphics[width=\textwidth,height=.85\textheight,keepaspectratio]{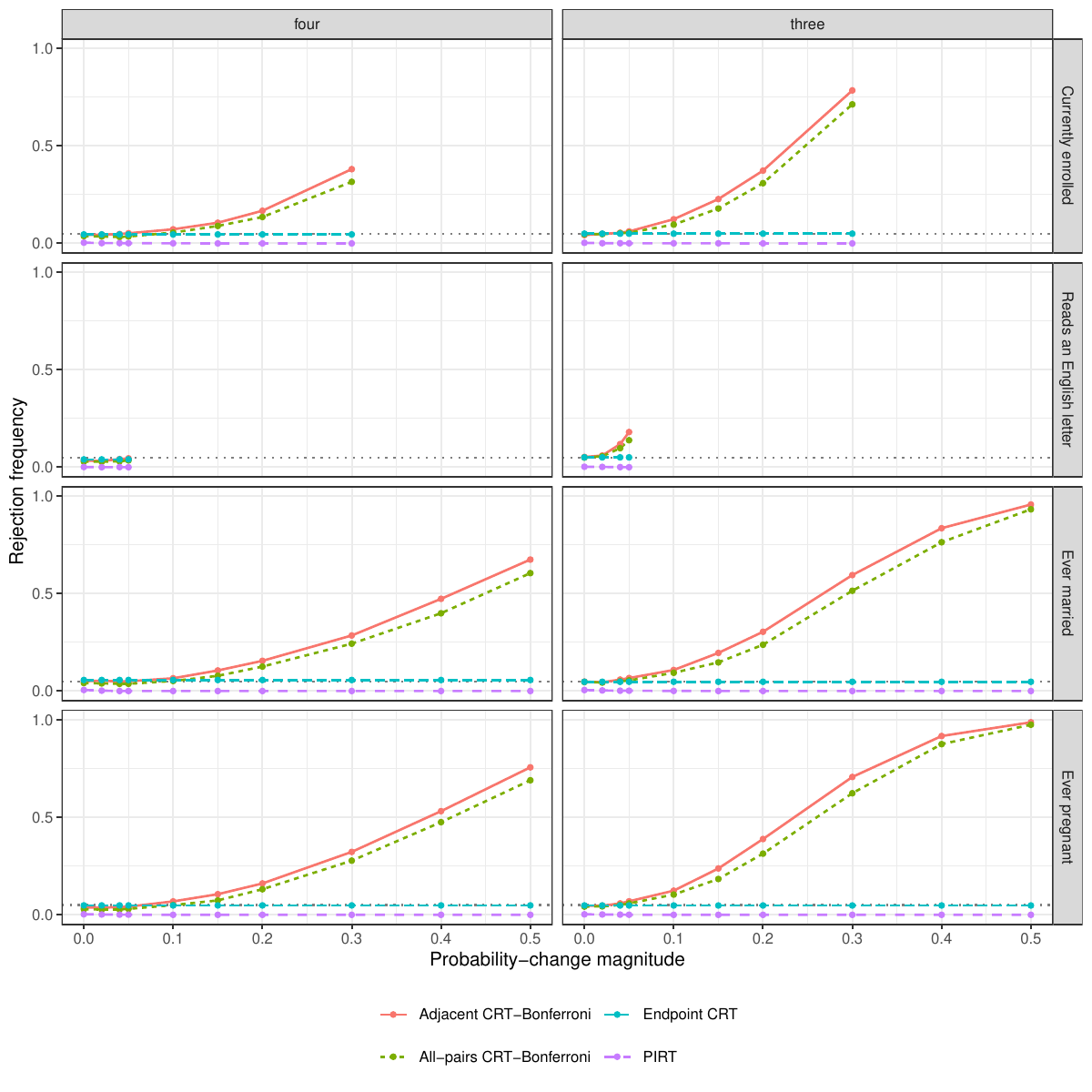}\caption{Outcome-specific hump monotonicity comparisons on complete feasible grids.}\label{fig:monotone_outcomes_hump}\end{figure}
\begin{figure}[p]\centering\includegraphics[width=\textwidth,height=.85\textheight,keepaspectratio]{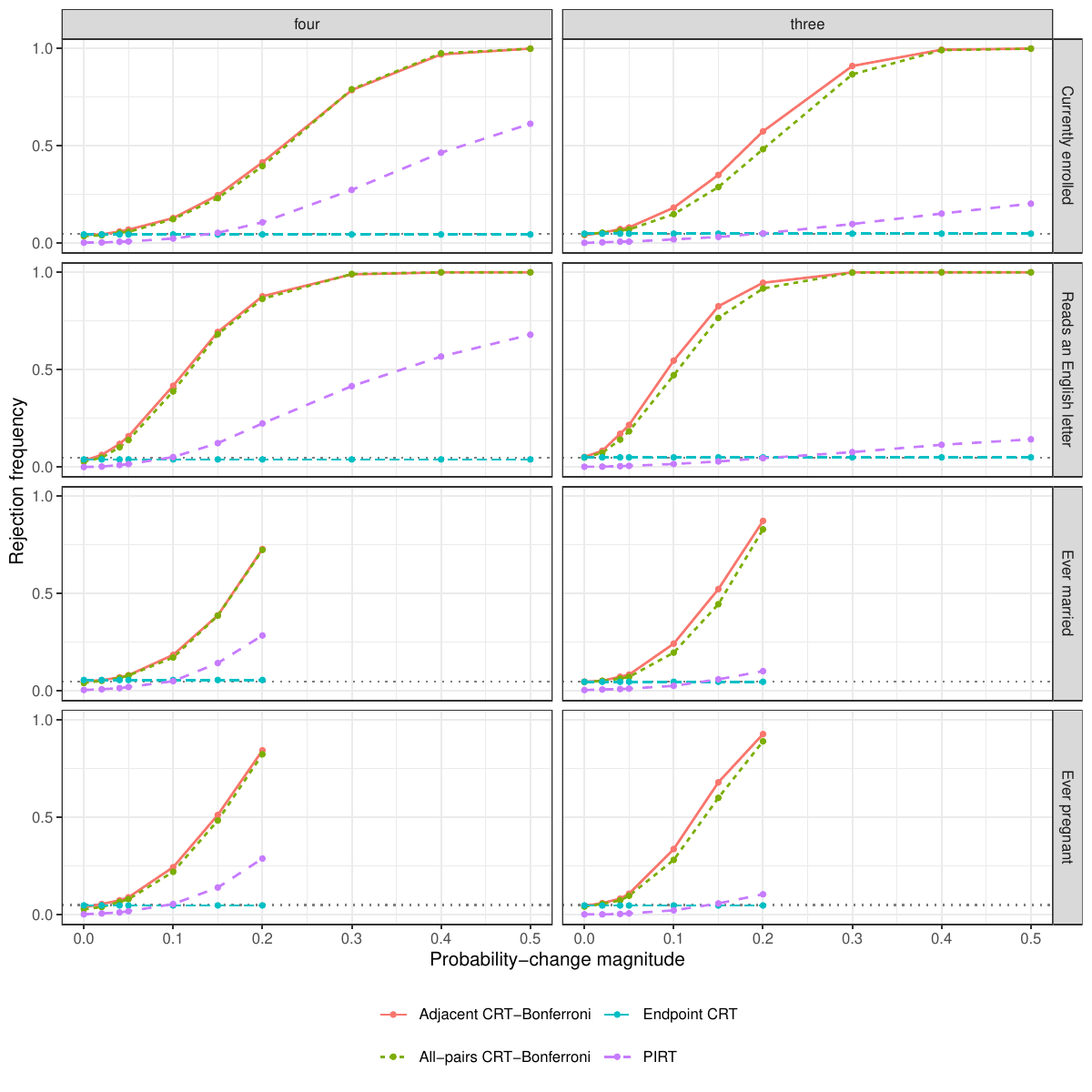}\caption{Outcome-specific trough monotonicity comparisons on complete feasible grids.}\label{fig:monotone_outcomes_trough}\end{figure}

\clearpage
\subsection{Numerical illustration of the weighted weak null}
\label{app:weak-numerics}
The synthetic populations use $K\in\{20,40,80,160\}$ EAs with repeated sizes $(4,8,12,20)$ and labels .25 and .50. With $\kappa=1$, each label receives $K/2$ EAs; with $\kappa=.5$, each receives $K/4$ and the remaining EAs are fully treated. Write $v_{ij}=\{i-(n_j+1)/2\}/n_j$ and $\bar n=K^{-1}\sum_j n_j$. The untreated potential outcomes are

\[Y_{ij}(0,.25)=2+.7\sin(j\sqrt{2})+.3v_{ij},\]
\[h_{ij}=.7(n_j/\bar n-1)+.4\cos(j\sqrt{3})+.8v_{ij},\]
\[Y_{ij}(0,.50)=Y_{ij}(0,.25)+h_{ij}-K^{-1}\sum_{\ell=1}^K w_{\ell K}\bar h_\ell+\tau,\]
where $\bar h_\ell$ is the within-EA mean and $w_{jK}=K\lambda_{jK}$ equals one for the equal-cluster target and $n_j/\bar n$ for the unit-weighted target. Each target is centered separately, so the weighted average-spillover null \eqref{eq:null-weak-weighted} holds exactly at $\tau=0$. Power uses $\tau\in\{.10,.20,.30\}$.

Both focal CRTs use maximal feasible focal counts. The centered version subtracts the pooled focal-EA mean before multiplication by $K\lambda_j$, holding this center fixed over permutations. The centered Neyman test uses the same focal array; the all-control benchmark uses all controls and the same target weights. All four procedures studentize using \eqref{eq:stat-studentized} on their respective EA arrays. The CRTs use upper-tail permutation probabilities; the normal benchmarks reject when the observed statistic exceeds $\Phi^{-1}(.95)$. Rejection frequencies average over assignment and focal selection. Table~\ref{tab:weak-equal} reports identical rejection frequencies for the original and centered tests under equal cluster weights, consistent with their algebraic equivalence. For the unit-average target, Figure~\ref{fig:weak-power} illustrates substantial power gains from centering at larger positive shifts in these synthetic populations.

\begin{table}[htbp]\centering\small
\caption{Weak-null rejection percentages: equal target}\label{tab:weak-equal}
\begin{tabular}{lrrrrr}\toprule
$K$ & $\kappa$ & Original & Centered & Neyman & All controls \\ \midrule
20 & 1 & 4.5 (0.5) & 4.5 (0.5) & 5.1 (0.5) & 5.2 (0.5) \\
20 & 0.5 & 5.5 (0.5) & 5.5 (0.5) & 7.3 (0.6) & 7.5 (0.6) \\
40 & 1 & 3.3 (0.4) & 3.3 (0.4) & 3.8 (0.4) & 3.7 (0.4) \\
40 & 0.5 & 4.6 (0.5) & 4.6 (0.5) & 5.1 (0.5) & 5.1 (0.5) \\
80 & 1 & 3.1 (0.4) & 3.1 (0.4) & 3.3 (0.4) & 3.3 (0.4) \\
80 & 0.5 & 4.4 (0.5) & 4.4 (0.5) & 4.6 (0.5) & 4.7 (0.5) \\
160 & 1 & 3.8 (0.4) & 3.8 (0.4) & 3.8 (0.4) & 4.1 (0.4) \\
160 & 0.5 & 4.8 (0.5) & 4.8 (0.5) & 5.0 (0.5) & 5.2 (0.5) \\
\bottomrule\end{tabular}\end{table}

\begin{table}[htbp]\centering\small
\caption{Weak-null rejection percentages: unit target}\label{tab:weak-unit}
\begin{tabular}{lrrrrr}\toprule
$K$ & $\kappa$ & Original & Centered & Neyman & All controls \\ \midrule
20 & 1 & 4.7 (0.5) & 2.5 (0.3) & 3.0 (0.4) & 5.4 (0.5) \\
20 & 0.5 & 4.1 (0.4) & 2.1 (0.3) & 3.3 (0.4) & 6.2 (0.5) \\
40 & 1 & 3.9 (0.4) & 2.6 (0.4) & 3.0 (0.4) & 4.6 (0.5) \\
40 & 0.5 & 4.2 (0.4) & 2.5 (0.4) & 3.2 (0.4) & 4.8 (0.5) \\
80 & 1 & 3.9 (0.4) & 2.1 (0.3) & 2.4 (0.3) & 3.9 (0.4) \\
80 & 0.5 & 4.0 (0.4) & 2.6 (0.4) & 2.9 (0.4) & 4.8 (0.5) \\
160 & 1 & 4.3 (0.5) & 3.3 (0.4) & 3.4 (0.4) & 4.6 (0.5) \\
160 & 0.5 & 5.0 (0.5) & 4.2 (0.5) & 4.6 (0.5) & 5.0 (0.5) \\
\bottomrule\end{tabular}\end{table}

\begin{figure}[htbp]\centering\includegraphics[width=\textwidth]{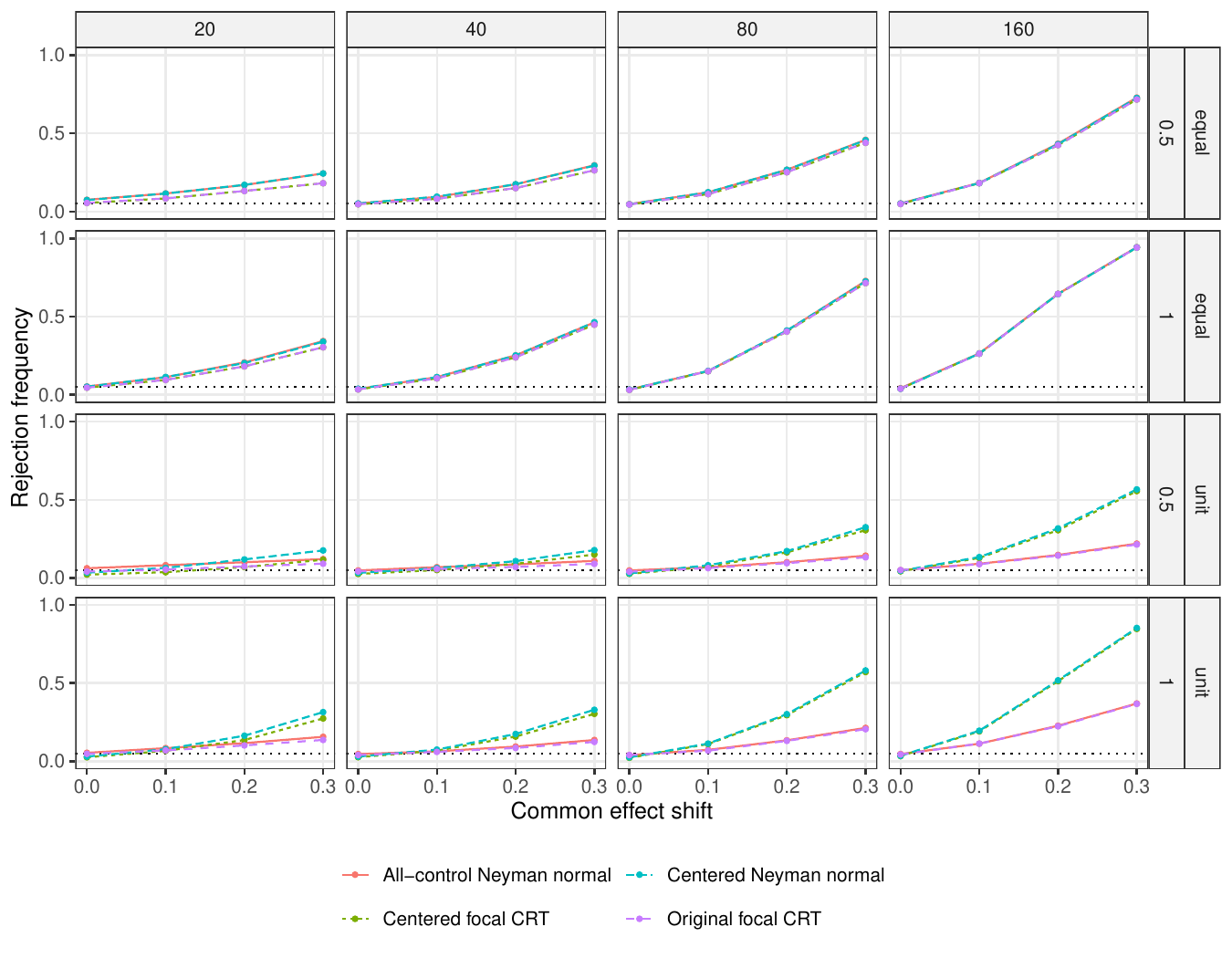}\caption{Weighted weak-null size and power over fixed finite populations.}\label{fig:weak-power}\end{figure}

\clearpage
\putbib[biblio]
\end{bibunit}
\end{document}